\documentclass[12pt]{article}

\usepackage{graphicx}
\usepackage{amsmath}
\usepackage{amssymb}
\usepackage{amsthm,mathtools,xcolor}
\usepackage[numbers,sort&compress]{natbib}

\usepackage[colorlinks = false]{hyperref}

\usepackage{subcaption}
\usepackage[margin=1in]{geometry}

\newcommand{\footremember}[2]{%
    \footnote{#2}
    \newcounter{#1}
    \setcounter{#1}{\value{footnote}}%
}
\newcommand{\footrecall}[1]{%
    \footnotemark[\value{#1}]%
}

\newcommand{\be}{\begin{equation}}
\newcommand{\ee}{\end{equation}}
\newcommand{\ba}{\begin{array}}
\newcommand{\ea}{\end{array}}
\newcommand{\bea}{\begin{eqnarray}}
\newcommand{\eea}{\end{eqnarray}}

\newcommand{\ra}{\rangle}
\newcommand{\la}{\langle}

\newcommand{\calH}{{\cal H }}

\newcommand{\calN}{{\cal N }}

\newcommand{\calF}{{\cal F }}

\newcommand{\calJ}{{\cal J }}

\newcommand{\calS}{{\cal S }}

\newcommand{\EE}{\mathbb{E}}
\newcommand{\ZZ}{\mathbb{Z}}
\newcommand{\CC}{\mathbb{C}}

\newcommand{\RR}{\mathbb{R}}
\newcommand{\HH}{\mathbb{H}}

\newcommand{\multi}[1]{\mathbf{{#1}}}
\newcommand{\nn}{\multi{n}}
\newcommand{\mm}{\multi{m}}

\newcommand{\herm}{\mathrm{He}}

\newtheorem{prop}{Proposition}
\newtheorem{lemma}{Lemma}

\newtheorem{corol}{Corollary}

\newtheorem{theorem}{Theorem}
\newtheorem*{theorem*}{Theorem}
\newtheorem*{lemma*}{Lemma}
\newtheorem*{corol*}{Corollary}

\usepackage{acro}

\DeclareAcronym{nse}{
  short=NSE,
  long=Navier--Stokes equations
}

\title{Quantum simulation of a noisy classical nonlinear dynamics}

\author{Sergey Bravyi\footremember{ibm}{IBM Quantum, IBM T.J. Watson Research Center, Yorktown Heights, NY 10598 (USA)}
\and
Robert Manson-Sawko\footremember{daresbury}{IBM Research Europe,  Daresbury Laboratory, Warrington, United Kingdom}
\and 
Mykhaylo Zayats\footremember{Dublin}{IBM Quantum, IBM Research Europe, Trinity Business School, Dublin,  D02 F6N2 (Ireland)} \and Sergiy Zhuk\footrecall{Dublin}}

\begin{document}
\maketitle

\begin{abstract}
We present an end-to-end quantum algorithm   for simulating nonlinear dynamics
described by a system of stochastic  dissipative differential equations
with a quadratic nonlinearity.
The stochastic part of the system is modeled   by a Gaussian  noise in the equation of motion and  in the initial conditions.  
Our algorithm can approximate the expected value of any correlation function that depends on $O(1)$ variables with rigorous bounds on the approximation error.
The runtime scales polynomially with $\log{N}$, $t$,  $J$, and  $\lambda_1^{-1}$, where
$N$ is the total number of variables, $t$ is the evolution time,  $J$ is the nonlinearity strength, and $\lambda_1$ is the smallest dissipation rate.
However, the runtime scales exponentially with a parameter quantifying inverse relative error in the initial conditions. 
To the best of our knowledge, this is the first rigorous  quantum algorithm capable of simulating strongly nonlinear systems with
$J\gg \lambda_1$ at the cost poly-logarithmic in $N$ and polynomial in $t$.
The considered simulation problem is shown to be BQP-complete, providing a strong evidence for a quantum advantage. 
We benchmark the quantum algorithm via numerical experiments by simulating a vortex flow in the 2D Navier Stokes equation.
\end{abstract}

\newpage

\tableofcontents

\newpage

\section{Introduction}
\label{sec:intro}

Numerical simulation of nonlinear dynamical systems such as turbulent fluid flows is a formidable challenge for modern computers.
Nonlinearity often precludes an accurate prediction of the dynamics over long periods of time due to a chaotic behavior which  amplifies modeling errors and uncertainty of initial conditions. 
Furthermore, nonlinear dynamics often spans a wide range of spatial and temporal scales necessitating discretized models
with an extremely large number of variables~\cite{yang2021grid}. Such models are usually expressed by  ordinary differential equations (ODEs) 
$dX(t)/dt = f(X(t))$, where $t\ge 0$ is the evolution time, $X(t)$ is a state vector with $N$ components, and $f\, : \, \RR^N\to \RR^N$ is a nonlinear function. 
The goal is to determine some property of ODE's solution $X(t)$ for a given initial condition $X(0)$. 

The seminal work by Liu et al.~\cite{liu2021efficient} opened the door to a quantum simulation of nonlinear classical systems.
It was shown that  quantum computers can solve certain quadratic ODEs 
astonishingly fast:  the runtime scales only 
logarithmically
with the number of variables $N$ and polynomially with the evolution time $t$. 
The main technical tool employed by~\cite{liu2021efficient} is the Carleman linearization method.
It maps the classical  solution $X(t)$  to a quantum state that lives in an infinite dimensional Hilbert space
and evolves in time linearly.   Similar classical-to-quantum mappings have been investigated by many authors~\cite{koopman1931hamiltonian,neumann1932operatorenmethode,kowalski1997nonlinear,engel2021linear,tanaka2023polynomial,novikau2025quantum,shi2023koopman}.

The quantum algorithm of Liu et al.~\cite{liu2021efficient} has been rigorously justified only for ODEs that include  {\em dissipation}.
An example of a dissipative  ODE  is $dX(t)/dt = -\lambda X(t)+f(X(t))$, where $\lambda>0$ is the dissipation rate and $f$ is a quadratic polynomial.
The algorithm of~\cite{liu2021efficient} is efficient only in the regime 
$R<1$, where $R$ is the ratio between a suitable norm $\|f\|$ of 
the quadratic term  and the minimum dissipation rate\footnote{More generally, $R$ also depends on the norm of the initial state $X(0)$. For simplicity, here we 
assume that $X(0)$  is a unit norm vector. We note that $R$ is qualitatively similar to the Reynolds number  in fluid dynamics.}. In words, nonlinearity must be weak compared with the dissipation. 
In contrast, Leyton and Osborne~\cite{leyton2008quantum} showed that a quantum computer can solve strongly nonlinear quadratic ODEs 
satisfying certain measure-preservation condition
at the cost scaling logarithmically with $N$ and exponentially with the evolution time $t$. The exponential scaling arguably 
rules out most of practical applications.
Given this situation, a natural question is whether a quantum or classical algorithm can solve quadratic ODEs
at the cost   polynomial in both
$\log{N}$ and $t$ in the strongly nonlinear regime $R\gg 1$. To the best of our knowledge, this question remains open.

In practice, any dynamical system is subject to some amount of  {\em noise}. The noise can be caused, for example, by unwanted interactions with the environment or small
 uncertainty in the initial conditions. Unless active efforts are made to undo the effect of noise,  
 noisy dynamics arguably provides a more complete model of real-world systems. 
 The presence of noise can also have profound 
 implications for the computational complexity of various simulation tasks. 
 For example, the breakthrough result by Aharonov et al.~\cite{aharonov2023polynomial} demonstrated that noisy random quantum circuits can be efficiently simulated 
 on a classical computer while this task is believed to be exponentially hard classically in the absence of noise~\cite{aaronson2016complexity}.
 In this example noise is our friend as it drastically simplifies the simulation task at hand.

Here we show that strongly nonlinear ($R{\gg}1$) classical dynamics  with a quadratic nonlinearity can be efficiently simulated
on a quantum computer provided that the simulated system is
{\em both dissipative and  noisy}. The noisy dynamics is described by a system of stochastic ODEs obtained from the noiseless ODE
 by adding a Gaussian  noise to the equation of motion and to the initial state $X(0)$. 
The runtime of our quantum algorithm scales polynomially with $\log{N}$,
evolution time $t$, inverse error tolerance, and parameters quantifying sparsity, dissipation, and nonlinearity strength
of the ODE. However, there are two caveats. First,
we have to assume that the ODE  satisfies a certain divergence-free condition analogous
to the measure-preservation property of Leyton and Osborne~\cite{leyton2008quantum}.
Second, runtime scales exponentially with a parameter quantifying inverse relative error in the initial conditions
induced  by the noise.
Thus the noise rate for the initial conditions cannot be too small. 
Despite these limitations, 
we show that the considered simulation problem is BQP-complete providing a strong  evidence for a quantum advantage.

In contrast, it is known that  some noise-free and dissipation-free 
 nonlinear systems may be hard to simulate even on a quantum computer
due to an extreme sensitivity to small changes of parameters, e.g. initial conditions, that grants the ability to discriminate non-orthogonal states and efficiently solve NP-hard problems
~\cite{abrams1998nonlinear,childs2016optimal,brustle2025quantum}.
The sensitivity to initial conditions  is exemplified by a rigid pendulum balanced in the  upward position. In theory, the pendulum will stay in this position forever. However, 
even a  tiny perturbation of the initial state will be amplified exponentially as the time advances until it results in a 
swinging motion. Intuitively, adding  noise to the initial conditions and to the equation of motion makes the system less sensitive to small
changes of parameters 
removing the main obstacle for quantum algorithms. Our work rigorously confirms this intuition. 

While the present work is mainly focused on the asymptotic runtime analysis, 
a natural next step is to develop practical quantum algorithms for solving nonlinear stochastic ODEs
relevant for real-world applications. An example of a notoriously hard to solve dissipative and quadratic system is given by the Navier--Stokes equation (NSE) in turbulent regimes, where nonlinear advection (quadratic term) dominates diffusion (dissipative term) on a so-called inertia manifold~\cite{Foias_Manley_Rosa_Temam_2001}: for example, resolving all dynamically relevant scales of turbulent atmospheric boundary layer requires direct numerical simulations of the NSE with Reynolds number of order $10^8$ (so $\lambda\approx 10^{-8}$) but as of 2020 the state-of-the-art methods have reached Reynolds number of just $10^4$~\cite{garcia-villalba_using_2020}. We note that the number of variables $N$ in discretizations of the NSE 
scales roughly quadratically with the Reynolds number~\cite{yang2021grid}. Application of the algorithm of~\cite{liu2021efficient} to quadratic ODEs associated with NSEs was discussed in~\cite{mikel-sanz-PRR25}: the authors concluded that the algorithm is efficient only for very low Reynolds number i.e. when the linear term (dissipation) dominates the nonlinear term (advection).\\ 
Other examples include nonlinear Markov diffusions, e.g. overdamped Langevin diffusion in global optimization~\cite{trillos2023optimization}, finance~\cite{alghassi2022variational} and uncertainty  quantification~\cite{frank2018detectability,yeo2024reducing}. The proposed method has a great potential to provide theoretical insights and computational breakthroughs if applied to these problems when large scale fault-tolerant quantum computers become available.

\subsection{Main results}
\label{sec:setup}

Consider a system of nonlinear ODEs with $N$ real variables
$X=(X_1,\ldots,X_N)$ and the equation of motion
\be
\label{ODE}
\frac{d}{dt} X_i(t) = -\lambda_i X_i(t) + f_i(X(t))
\ee
with a given  initial state $X(0)$.
Here $t\ge 0$ is the evolution time, $\lambda_i>0$ are dissipation rates, and $f_i\, : \, \RR^N\to \RR$ are quadratic polynomials, 
\be
\label{drift}
f_i(x) = \sum_{j=1}^N b_{ij} x_j +  \sum_{j,k=1}^N c_{ijk} x_j x_k
\ee
for some real coefficients 
$b_{ij}$ and $c_{ijk}$. Assume wlog that $c_{ijk}=c_{ikj}$.
We shall refer to the polynomials $f_i$ as {\em drift functions}.
The system is said to have  the  {\em drift strength} $J$ if $|b_{ij}|\le J$ and $|c_{ijk}|\le J$ for all indices $i,j,k$.
The system is called {\em $s$-sparse} if any row and any column
of the matrix $b=\{b_{ij}\}$ has at most $s$ nonzeros and any two-dimensional slice of the tensor $c=\{c_{ijk}\}$ has
at most $s$ nonzeros\footnote{In other words, for any $i$ there must be at most $s$ pairs $j,k$ with $c_{ijk}\ne 0$,
for any $j$ there must be at most $s$ pairs $i,k$ with $c_{ijk}\ne 0$,
and for any $k$ there must be at most $s$ pairs $i,j$ with $c_{ijk}\ne 0$.}. Assume wlog that $\lambda_1\le \lambda_2\le\ldots\le \lambda_N$.

We are interested in the regime when the number of variables $N$ is large enough that
storing the full description of the system and its solution in a classical computer memory
is impractical. Accordingly, we  assume that
the system  admits a succinct description by classical 
algorithms with the runtime scaling polynomially with $\log{N}$
that compute dissipation rates $\lambda_i$,
 values and locations of nonzero
 matrix elements $b_{ij}$ and $c_{ijk}$,
 see Section~\ref{sec:algorithm} for details.
 
For general nonlinear ODEs with quadratic $f_i$ the norm of the solution $X(t)$ can become infinite after a finite evolution time
and the problem becomes ill-defined, e.g. for Riccati equations~\cite{DASCALIUC201953}. To avoid this situation, we restrict our attention to a class
of ODEs satisfying certain zero divergence condition, which contains quadratic chaotic ODEs like Lorenz 96-model and Burgers-Hopf model~\cite{frank2018detectability}.  
To the best of our knowledge, this class was introduced by Engel, Smith, and Parker~\cite{engel2021linear} in the study of quantum algorithms based on linear embedding methods.
A closely related measure-preservation property was studied by Leyton and Osborne~\cite{leyton2008quantum}.
The ODE Eqs.~(\ref{ODE},\ref{drift})  is said to be {\em divergence-free} if $b$ and $c$ have zero diagonal entries, that is, $b_{ii}=0$,
$c_{ijj}=c_{jij}=c_{jji}=0$ for all indices $i,j$. In addition, off-diagonal entries of $b$ and $c$ must obey
\be
\label{NDE3}
\lambda_i b_{ij} + \lambda_j b_{ji}=0
\ee
and
\be
\label{NDE2}
\lambda_i c_{ijk} + \lambda_j c_{jki} + \lambda_k c_{kij} = 0
\ee
for all $i\ne j\ne k$.
These conditions ensure that a norm $\|x\|_\lambda^2 = \sum_{i=1}^N \lambda_i x_i^2$ is non-increasing\footnote{This can be easily checked
by computing the derivative $(d/dt)\|X(t)\|_\lambda^2$ from Eq.~(\ref{ODE}) and 
 using the
identity $\sum_{i=1}^N \lambda_i x_i f_i(x)=0$ which follows from Eqs.~(\ref{NDE3},\ref{NDE2}).}
along the trajectories $X(t)$, that is, $(d/dt) \|X(t)\|_\lambda  \le 0$. Thus the divergence-free condition can be viewed
as a generalization of unitarity to non-linear dissipative ODEs.
We demonstrate that discretizations of the Navier-Stokes equations are described by divergence-free quadratic ODEs
(although these ODEs are not sparse in the Fourier basis), see Section~\ref{sec:examples}.

Our goal is to simulate a noisy version of the ODE Eq.~(\ref{ODE}) described by a stochastic differential equation (SDE) 
\be
\label{SDE}
dX_i(t) = -\lambda_i X_i(t) dt  + f_i(X(t))dt  + 
\sqrt{q}\, dW_i(t).
\ee
Here $q>0$ is the noise rate and $W(t)\in \RR^N$ is 
a vector of $N$ independent Wiener processes (Brownian motions). In discrete time this SDE  is related to the following random walk:
\[
X_i(t+dt) = X_i(t)  -\lambda_i X_i(t) dt + f_i(X(t)) dt+
\sqrt{q dt } \calN(0,1),
\]
where $\calN(0,1)$ is the standard normally distributed random variable.
We note that in the zero noise limit ($q\to 0$) the solution of this SDE  converges (on average) to the solution of the corresponding noiseless ODE Eq.~(\ref{ODE}),
see~\cite[p.245]{chow2007infinite}.

 Our quantum algorithm outputs a real number that approximates the expected value of a  given {\em observable function} $u_0 \, : \, \RR^N\to \RR$
evaluated at the solution $X(t)$. This function describes the property of the solution that we would like to observe. 
For example, $u_0(x)=x_i$ if we are  interested in the expected value of a single variable $X_i(t)$
or $u_0(x)=x_ix_j$ if we are interested in the expected value of $X_i(t) X_j(t)$. 
We assume that $u_0(x)=\prod_{i=1}^N x_i^{d_i}$ where $d_i\ge0$
are integers such that $\sum_{i=1}^N d_i=O(1)$. 

Let $v(t,x)$ be  the expected value of the obsvervable function $u_0(X(t))$ along the trajectory $X(t)$ starting from a vector $X(0)=x+z$ 
where $x\in \RR^N$ is the desired initial condition and 
$z\in \RR^N$ models noise in the initial condition. 
The expected value is taken over Wiener noise and over $z$. 
Our setup is illustrated on Figure~\ref{fig:diagram_classical}.
For technical reasons, we assume that each  component $z_i$ is drawn independently from  the normal distribution 
with the zero mean and variance $q/(2\lambda_i)$. Thus 
\be
\label{v(t,x)}
 v(t,x) =
  \EE_{W,z} u_0(X(t))
    \quad \mbox{subject to $X(0)=x+z$}
 \ee
 where  $z_i\sim \calN(0,q/(2\lambda_i))$ for all $i$. 
Our main result is the following.

\begin{figure}
\centering\includegraphics[width=10cm]{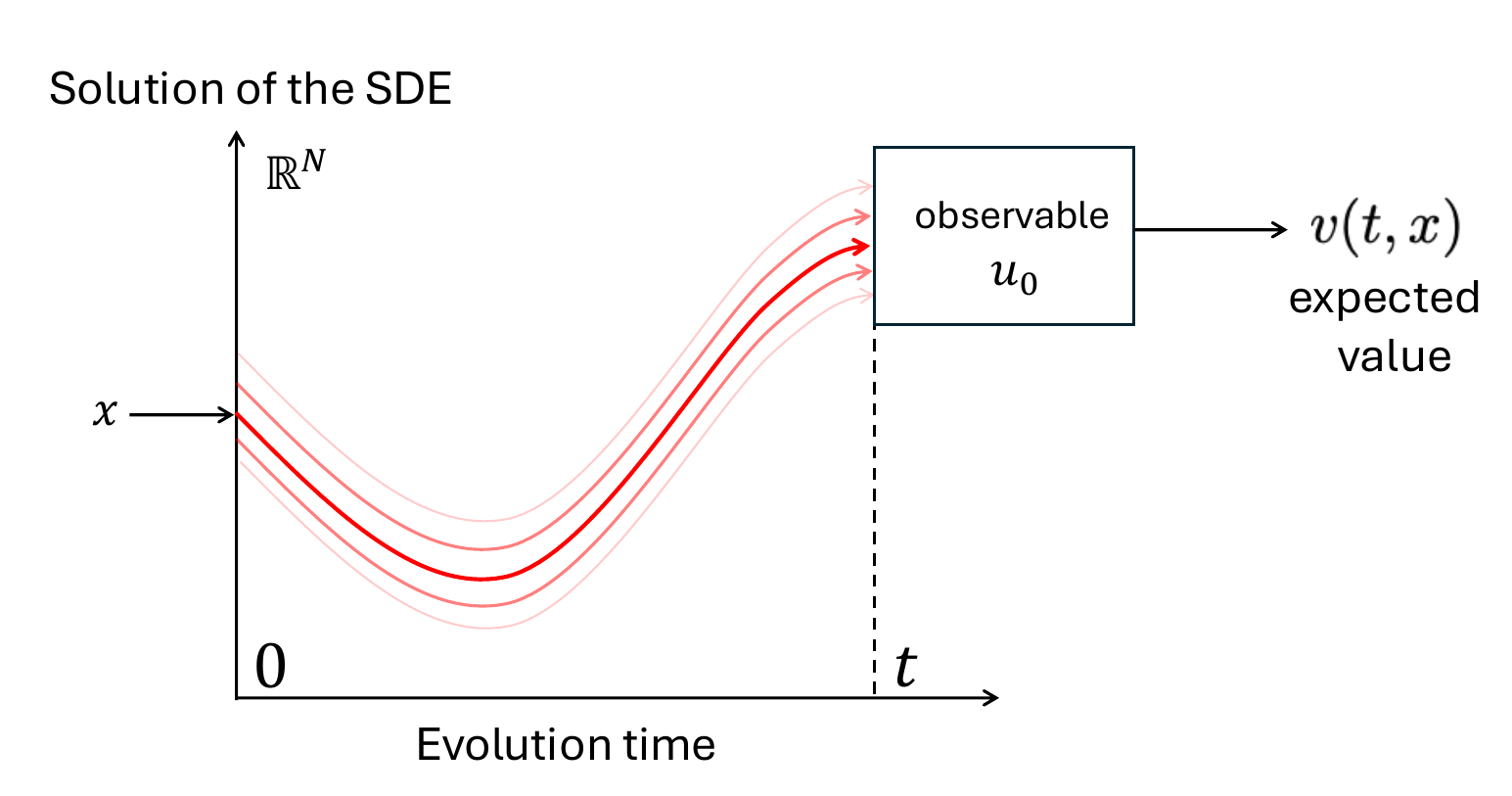}
        \caption{Cartoon of our simulation problem. We consider 
        a stochastic differential equation (SDE) for a state vector $X(t)\in \RR^N$ with  the initial condition $X(0)=x+z$, where $x$ is the desired initial condition and $z$ is a Gaussian noise.
        Red lines represent solutions of the SDE corresponding to different realization of noise, including Wiener noise in the SDE and noise in the initial condition.
        The goal is to  approximate the expected value of a given observable function $u_0\, : \, \RR^N\to \RR$ evaluated at the solution $X(t)$. The expected value is taken over noise realizations.  }
\label{fig:diagram_classical}
\end{figure}

\begin{theorem} 
\label{thm:algorithm}
There exists a quantum algorithm taking as input a vector $x\in \RR^N$ with $O(1)$ nonzeros, evolution time $t$,
an error tolerance $\epsilon$, and an observable function $u_0$ as above. 
The algorithm outputs a real number approximating the 
expected value $v(t,x)$ with an additive error $\epsilon$. 
The algorithm has runtime scaling
polynomially with  $\log{(N)}$,  $t$, $\epsilon^{-1}$, $J$, $s$,  $\lambda_1^{-1}$,
and $q$.
The runtime scales exponentially with the quantity $q^{-1} \|x\|_\lambda^2$,
where $\| x\|_\lambda^2  = \sum_{i=1}^N \lambda_i x_i^2$.
\end{theorem}
The exponential factor in the runtime is $\exp{\left(C q^{-1} \|x\|_\lambda^2\right)}$,
where $C=O(1)$ is a universal constant. 
By definition, the expected value and the variance of $X_i(0)$ is $x_i$ and $q/(2\lambda_i)$  respectively. Thus
\[
q^{-1} \| x\|_\lambda^2 =\frac12 \sum_{i=1}^N \frac{\left( \EE X_i(0)\right)^2}{\mathrm{Var}(X_i(0))}
\]
quantifies inverse relative error in the initial conditions. If the relative initialization error per variable 
is lower bounded by a positive constant and $x$ has $O(1)$ nonzeros, 
the runtime of our algorithm scales polynomially with $\log{(N)}$ and all other parameters. 
Let us emphasize that the runtime does not depend on the {\em maximum} dissipation rate $\lambda_N$.
This is a highly desirable feature since $\lambda_N$ scales as a positive power of $N$ for many practical applications, see Section~\ref{sec:examples}.

A natural question is whether the considered simulation problem can be solved by a classical algorithm with
a comparable or better runtime. Our second result shows that this  is highly unlikely since the problem of approximating 
the expected value $v(t,x)$ defined in Eq.~(\ref{v(t,x)})   is BQP-complete,
that is, as hard as simulating the universal quantum computer.
\begin{theorem}
\label{thm:BQP}
Suppose $U$ is $n$-qubit  quantum circuit with  $m$ gates.
There exists a linear divergence-free  SDE with $N=(m+1)2^n$ variables,  drift strength $J=O(m)$,
constant sparsity, noise rate, and dissipation rate such that 
\be
\label{U00}
\left|  \la 0^n|U|0^n\ra - v(t,x) \right|\le \frac1{10}
\ee
for $t=1$.
The initial state $x\in \RR^N$ has a single nonzero and obeys  $q^{-1} \|x\|_\lambda^2=1$.
\end{theorem}
Approximating the matrix element $\la 0^n|U|0^n\ra$ with a constant additive error, say  $1/3$, is the canonical example
of a BQP-complete problem~\cite{kitaev2002classical}. Theorem~\ref{thm:BQP} implies that approximating the quantity
$v(t,x)$  is also BQP-complete even if $t=1$ and $q^{-1} \|x\|_\lambda^2=1$.
In this regime the runtime of our quantum algorithm is polynomial in all parameters.

\subsection{Proof sketch}
Our main technical tool in the proof of Theorem~1 is the Kolmogorov equation  describing the time evolution of the expected value  
\be
\label{u(t,x)}
u(t,x) =
  \EE_{W} u_0(X(t))
    \quad \mbox{subject to $X(0)=x$}.
 \ee
Note that this coincides with the expected value $v(t,x)$ computed by our quantum algorithm except that the initial state $X(0)$ is now noiseless.
The Kolmogorov equation reads as
\[
\frac{\partial}{\partial t} u(t,x) = \sum_{i=1}^N (-\lambda_i x_i + f_i(x)) \frac{\partial}{\partial x_i} u(t,x)
+\frac{q}2 \sum_{i=1}^N  \frac{\partial^2}{\partial x_i^2} u(t,x)
\]
with the initial condition $u(0,x)=u_0(x)$. By mapping each variable $x_i$ to a quantum harmonic oscillator, we express the desired expected value $v(t,x)=\EE_z u(t,x+z)$ as an inner product of two (unnormalized) quantum states,
\[
v(t,x)=\la \psi_{out}(x)|\psi(t)\ra,
\]
where the ambient Hilbert space describes a system of $N$ quantum harmonic oscillators, $|\psi_{out}(x)\ra$ is the
coherent state embedding~\cite{engel2021linear} of the initial condition  $x$, 
the state $|\psi(0)\ra$
encodes the observable function $u_0$,
and $|\psi(t)\ra$ solves the second-quantized Kolmogorov equation  
\be
\label{ABC}
\frac{d}{dt} |\psi(t)\ra = (-A+B+C)|\psi(t)\ra
\ee
for certain  hermitian positive definite operator $A$, and anti-hermitian operators $B$ and $C$. Denoting the creation and annihilation operators
for the $i$-th oscillator by $a_i^\dag$ and $a_i$, one has $A=\sum_{i=1}^N \lambda_i a_i^\dag a_i$,
$B= \sum_{i,j=1}^N b_{ij}(\lambda_i/\lambda_j)^{1/2} \, a_j^\dag a_i$, while $C$ is a linear combination of operators $a_i a_j^\dag a_k^\dag- a_k a_j a_i^\dag$
with coefficients simply related  to $c_{ijk}$. Time evolution under $-A+B$ describes free bosons (linear optics). However, $C$  is cubic in the creation/annihilation operators and $C$ does not preserve the particle number, which complicates the analysis.
Moreover, $C$ cannot be considered as a small perturbation in the strong nonlinearity regime. 
To complicate the matter further, $A,B,C$ are unbounded operators. Thus,  it is far from obvious that Eq.~(\ref{ABC}) has a well-defined solution.

To address these challenges, we derive operator upper bounds\footnote{If $X$ and $Y$ are Hermitian operators, we write $X\le Y$ to indicate that $Y-X$ is positive semidefinite.} on the commutator of $A$ and $B+C$, namely,
$[A,B+C] \le A^{2} + \kappa A$, where $\kappa$
depends only on the parameters $q,s,J$, and $\lambda_1^{-1}$.
We derive similar upper bounds for commutators $[A^p,B+C]$ with $p\ge 1$.
 Likewise, we show that
  $|\la \phi|B+C|\varphi\ra| \le \gamma \sqrt{\la \varphi|A|\varphi\ra\la \phi|A^2|\phi\ra}$ for any states $\varphi,\phi$.
Here $\gamma$ depends only on $q,s,J$, and $\lambda_1^{-1}$.

Equipped with these bounds,
we were able to prove that the solution $\psi(t)$ of Eq.~(\ref{ABC}) is well-defined and can can be efficiently approximated
on a quantum computer. To this end we define a regularized Kolmogorov equation that depends on a cutoff parameter $k>0$.
Loosely speaking, the regularization "turns off" the term $B+C$ on a subspace spanned by Fock basis vectors $|\mm\ra$ satisfying $\la \mm|A|\mm\ra> k$.
Here $\mm$ specifies the occupation number of each oscillator. 
Let $\psi_k(t)$ be the solution of the regularized Kolmogorov equation with the initial condition $\psi_k(0)=\psi(0)$.
 We prove that the family of states $\psi_k(t)$ has a limit
$\psi(t)=\lim_{k\to \infty} \psi_k(t)$ and this limit solves the original (unregularized) Kolmogorov equation. Furthermore, the regularization error
$\epsilon=\|\psi(t)-\psi_k(t)\|$ is upper bounded by a polynomial function of $1/k$, $s$, $J$, $\lambda_1^{-1}$.
Thus one can choose $k=poly(\epsilon^{-1},s,J,\lambda_1^{-1})$. For a formal statement see Theorem~\ref{thm:regul} in Section~\ref{sec:regul-sz}.
Our regularization procedure is  the key for
reducing an infinite dimensional simulation problem to a finite dimensional one with a well controlled approximation error.

The regularized version of  Eq.~(\ref{ABC})  is  then projected onto the subspace $\calH_k$ spanned by Fock basis vectors $|\mm\ra$ satisfying $\la \mm|A|\mm\ra\le k$.
This subspace is finite-dimensional can can be expressed using $\log_2{\mathrm{dim}(\calH_k)}=O(k\lambda_1^{-1}\log{N})$ qubits in the worst case. 
The projected version of Eq.~(\ref{ABC}) can be efficiently simulated  on a quantum computer using off-the-shelf Hamiltonian simulation methods.
In general, the initial state $\psi_k(0)=\psi(0)$ does not belong to  $\calH_k$. However, we show that the dynamics of $\psi_k(t)$
on the orthogonal subspace $\calH_k^\perp$ can be efficiently simulated classically and we only need a quantum computer to simulate the dynamics within $\calH_k$.

The exponential scaling in the runtime of Theorem~1 results from normalization of the readout state $\psi_{out}(x)$.
We show that $\psi_{out}(x)$ has an infinite norm if the initial condition $X(0)=x$ is noiseless, while $\|\psi_{out}(x)\| = \exp{\left[ q^{-1}\|x\|_\lambda^2 \right]}$
in the noisy case, $X(0)=x+z$. A large norm of $\psi_{out}(x)$ amplifies errors incurred in the Hamiltonian simulation and in the regularization translating
into the runtime and qubit count scaling polynomially with $\|\psi_{out}(x)\|$, that is, exponentially with $q^{-1}\|x\|_\lambda^2$.
We leave as an open question whether this exponential factor can be removed. The workflow of our algorithm is illustrated on Figure~\ref{fig:diagram}.

Our "Kolmogorian simulation" approach drastically departs from the prior quantum algorithms for simulating nonlinear dynamics~\cite{liu2021efficient, koopman1931hamiltonian,neumann1932operatorenmethode,kowalski1997nonlinear,engel2021linear,tanaka2023polynomial,novikau2025quantum,shi2023koopman}
that
aim to track dynamics of the full state vector $X(t)$, as opposed to the expected value of a single scalar function of $X(t)$. 
At a high level, this idea of  mapping states to expectations   has similarities with the one of Refs.~\cite{aharonov2023polynomial,beguvsic2025simulating,schuster2024polynomial}
that investigated classical simulation of quantum circuits based on the Heisenberg time evolution of observables
expanded in the Pauli basis. 
We note that the Kolmogorov equation  is a well known tool in the literature on stochastic processes~\cite{gikhman2004theory,da2004kolmogorov}: this very tool is used to approach many open problems in mathematics, e.g. the Navier–-Stokes existence and uniqueness Millennium Prize Problem\footnote{\url{https://www.claymath.org/millennium/navier-stokes-equation/}} (see~\cite[p.235]{dalang2015stochastic}) and has already generated a set of exciting results recognized by Fields medals~\cite{hairer2006ergodicity}.

Let us  emphasize that our main technical contribution is Theorem~\ref{thm:algorithm} whereas Theorem~\ref{thm:BQP}
follows almost trivially  from the definitions.  Indeed, since the considered noise has 
zero mean and the SDE of Theorem~\ref{thm:BQP} is linear, noise does not contribute to the
observed expected value. Thus we just need to prove BQP-completeness for linear ODEs satisfying the sparsity and divergence-free conditions.
This requires only a minor modification of the known circuit-to-Hamiltonian mapping~\cite{kitaev2002classical}.

\subsection{Organization of the paper}

The rest of this paper is organized as follows. 
Section~\ref{sec:Kolmogorov} contains our main technical contributions.
It considers  the second-quantized version of the Kolmogorov equation (KE) describing a system of $N$ quantum harmonic oscillators and
proves that this equation has a well-defined solution that can be well approximated by a solution of the regularized KE. 
Section~\ref{sec:readout} highlights some challenges associated with extracting the desired expected values from the solution of KE.
We show how to overcome these challenges in the case then the initial condition $X(0)$ is noisy obtaining an efficient readout step. 
We combine the two ingredients to prove Theorem~\ref{thm:algorithm} in Section~\ref{sec:algorithm}.
Section~\ref{sec:BQP} contains the proof of Theorem~\ref{thm:BQP}.
Section~\ref{sec:examples} describes 
an example of 
a divergence-free system with a quadratic nonlinearity: 
the incompressible 
Navier–Stokes equation
 in two spacial dimensions on a torus discretized in the eigen-basis of the Stokes operator~\cite{flandoli1998kolmogorov}.
This system has natural dissipation (molecular diffusion) and is stochastically forced.  We numerically simulate our quantum algorithm
and compare its outputs with expected values obtained  by the brute force sampling. We use simple analytic solution (Taylor-Green vortex) to test our algorithm in the limit of very small noise without changing the dissipation.
 We provide proofs of some technical lemmas  in the appendix.

\begin{figure}
\centering\includegraphics[width=14cm]{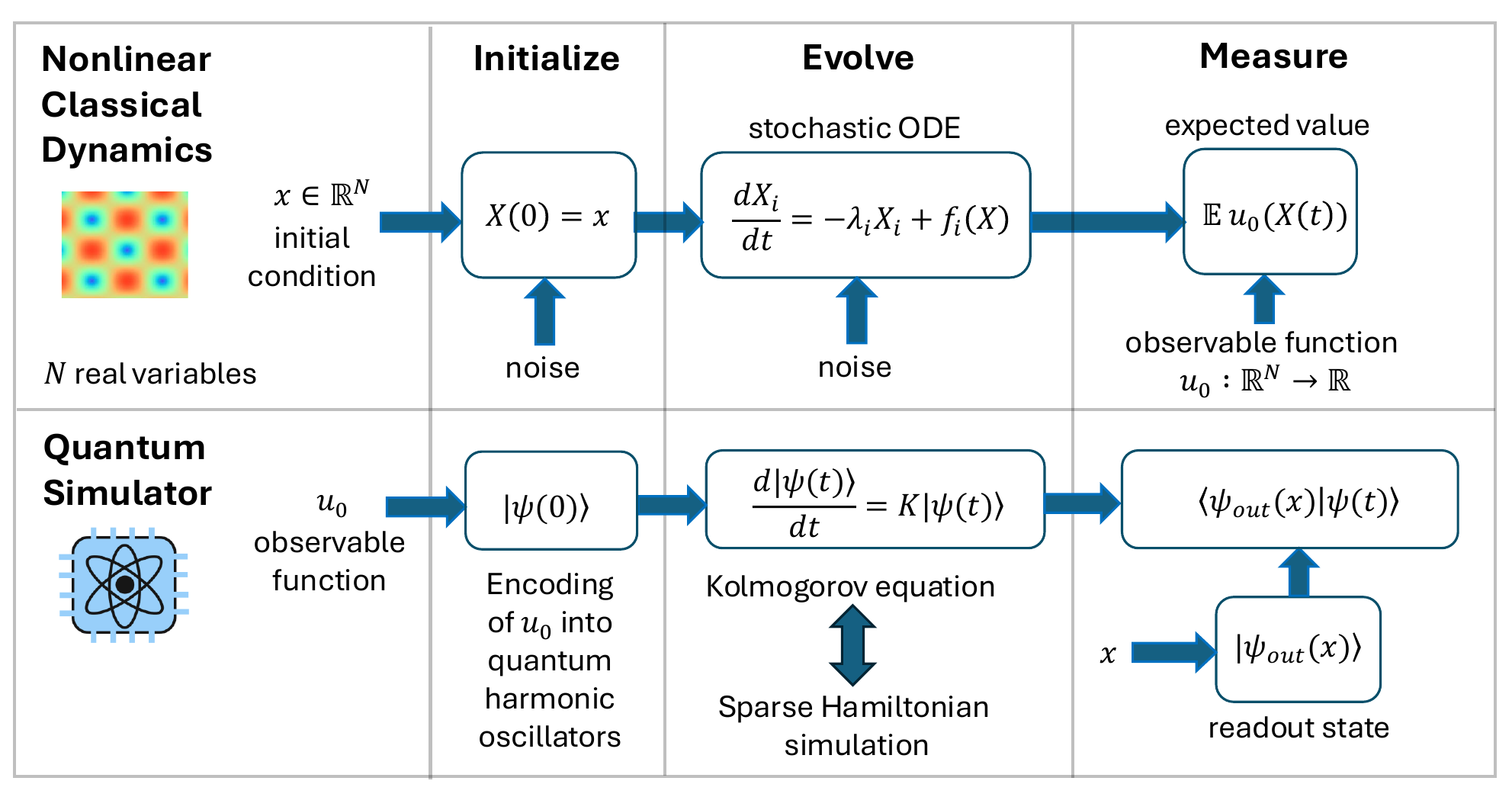}
        \caption{Quantum simulation of a nonlinear SDE. The workflow is divided into  initialization, time evolution, and   measurement stages. 
    The simulator is initialized in a state $|\psi(0)\ra$ encoding
     the observable function $u_0$ into the Hilbert space of $N$ quantum harmonic oscillators.
      Time evolution of $|\psi(t)\ra$ is governed by a "Kolmogorian" $K=-A+B+C$, 
     where $A=\sum_{i=1}^N \lambda_i a_i^\dag a_i$ describes noise and dissipation while $B$ and $C$ are anti-hermitian operators
     simply related to the drift functions $f_i(x)$. The Hilbert space is truncated to a finite dimension based on a novel 
     regularization procedure.  Time evolution under the truncated Kolmogorian is simulated using sparse Hamiltonian simulation methods. 
       The measurement estimates the inner product between the time evolved state $|\psi(t)\ra$
     and a readout state $|\psi_{out}(x)\ra$ encoding the desired  initial condition $x$
     of the SDE. The measured inner product gives an $\epsilon$-approximation of the classical expected value. }
    \label{fig:diagram}
\end{figure}

\section{Kolmogorov equation}
\label{sec:Kolmogorov}
Time evolution of the expected value $u(t,x)$ defined in Eq.~(\ref{u(t,x)}) is 
governed by the backward Kolmogorov equation~\cite{da2004kolmogorov}.
 Assuming that $u(t,x)$ is a sufficiently smooth function, the Kolmogorov equation (KE) reads as
\be
\label{Kolmogorov1}
\frac{\partial}{\partial t} u(t,x) = \sum_{i=1}^N (-\lambda_i x_i + f_i(x)) \frac{\partial}{\partial x_i} u(t,x)
+\frac{q}2 \sum_{i=1}^N  \frac{\partial^2}{\partial x_i^2} u(t,x)
\ee
for $t\ge 0$ with the initial condition $u(0,x)=u_0(x)$.  
(The inverse mapping from the partial differential equation Eq.~(\ref{Kolmogorov1}) to an SDE is known as Feynman–Kac formula~\cite{oksendal2013stochastic}.)

\subsection{Mapping to quantum harmonic oscillators}

Let 
\be
\label{mu(x)}
\mu(x)\sim \exp{\left(-q^{-1} \|x\|_\lambda^2\right)}
\ee
be a normalized Gaussian distribution such that $\int_{\RR^N} dx\, \mu(x)=1$.
Recall that $\|x\|_\lambda^2=
\sum_{i=1}^N \lambda_i x_i^2$.
We shall construct  a solution of Eq.~(\ref{Kolmogorov1}) in the form
\be
u(t,x)= \frac{\psi(t,x)}{
\sqrt{\mu(x)}},
\ee
where $\psi(t,x)$ is a function satisfying $\int_{\RR^N} dx\,  |\psi(t,x)|^2 <\infty$\footnote{This ansatz ensures that the solution $u(t,x)$ belongs to a Gauss-Sobolev space associated with the measure $\mu(x)$, see for instance~\cite{da2004kolmogorov}.}
A simple algebra shows that Eq.~(\ref{Kolmogorov1}) is equivalent to 
\be
\label{Kolmogorov1a}
\frac{\partial}{\partial t} \psi(t,x)
=-\frac12 \sum_{i=1}^N \left( q^{-1}\lambda_i^2  x_i^2 -q \frac{\partial^2}{\partial x_i^2} - \lambda_i\right)\psi(t,x)
+ \sum_{i=1}^N f_i(x) ( q^{-1} \lambda_i x_i + \frac{\partial}{\partial x_i}) \psi(t,x).
\ee
The term $( q^{-1}\lambda_i^2  x_i^2 -q \frac{\partial^2}{\partial x_i^2} - \lambda_i)/2$ is the Hamiltonian of
a quantum harmonic oscillator, up to an overall energy shift. 
The oscillator has mass $q^{-1}$, spring constant
$q^{-1}\lambda_i^2$, and frequency $\lambda_i$. 
It will be convenient to describe this Hamiltonian by the  creation ($a_i^\dag$) and annihilation ($a_i$) operators
\be
a_i = \frac1{\sqrt{2}}\left( \sqrt{\frac{\lambda_i}{q}} x_i + \sqrt{\frac{q}{\lambda_i}} \frac{\partial}{\partial x_i}\right) \quad \mbox{and} \quad 
a_i^\dag= \frac1{\sqrt{2}}\left( \sqrt{\frac{\lambda_i}{q}} x_i - \sqrt{\frac{q}{\lambda_i}} \frac{\partial}{\partial x_i}\right).
\ee
Then\footnote{Here we use the identity $(x_i\partial_{x_i}-\partial_{x_i}x_i)\psi = -\psi$.} $( q^{-1}\lambda_i^2  x_i^2 -q \frac{\partial^2}{\partial x_i^2} - \lambda_i)/2=\lambda_i a_i^\dag a_i$.
Using the identity
 $q^{-1} \lambda_i x_i + \frac{\partial}{\partial x_i}=a_i\sqrt{2\lambda_i/q}$ one can rewrite Eq.~(\ref{Kolmogorov1a}) as 
\be
\label{Kolmogorov1b}
\frac{\partial}{\partial t} \psi(t,x) = \left( -\sum_{i=1}^N \lambda_i a_i^\dag a_i  +   \sum_{i=1}^N \sqrt{\frac{2\lambda_i}{q}} f_i(x) a_i\right) \psi(t,x).
\ee
We shall work in the second quantization picture and 
express the wave function $\psi(t,x)$ of $N$ harmonic oscillators  in the Fock (occupation number) basis. 
Let $\ZZ_{\ge 0}$ be the set of nonnegative integers,  $m_i \in \ZZ_{\ge 0}$ be the occupation number of the $i$-th oscillator,
and $\mm = (m_1,\ldots,m_N)$ be a multi-index composed of $N$ occupation numbers. 
Let $|\mm\ra$ be the Fock basis vector
 corresponding to a multi-index $\mm \in \ZZ_{\ge 0}^N$.
These vectors form an orthonormal basis,  i.e. $\la \mm|\nn\ra=\delta_{\nn,\mm}$ for all multi-indices $\mm,\nn$, and 
\be
\label{ai_action}
a_i |\mm\ra = \left\{
\ba{rcl} 
\sqrt{m_i} |\mm-e^i\ra &\mbox{if} & m_i\ge 1,\\
0 &\mbox{if} & m_i=0.\\
\ea
\right.
\ee
Here $e^i\in \ZZ_{\ge 0}^N$ is a multi-index such that $(e^i)_i=1$ and $(e^i)_j=0$ for $j \ne i$.
Accordingly, $a_i^\dag |\mm\ra = \sqrt{m_i+1} |\mm+e^i\ra$ and
$a_i^\dag a_i |\mm\ra = m_i |\mm\ra$. Let $|{\bf 0}\ra=|00\ldots0\ra$ be the vacuum state
annihilated by all operators $a_i$. Recall that 
the Fock space $\calF$ is a Hilbert space formed by linear combinations $|\psi\ra = \sum_{\mm \in \ZZ_{\ge 0}^N}\;  \psi_\mm |\mm\ra$
with complex coefficients $\psi_\mm$ satisfying $\sum_{\mm\in \ZZ_{\ge 0}^N}\;  |\psi_\mm|^2 <\infty$. The inner product in $\calF$ is defined as
$\la \psi|\phi\ra = \sum_{\mm\in \ZZ_{\ge 0}^N}\; \psi_\mm^* \phi_\mm$.
The second-quantized version of Eq.~(\ref{Kolmogorov1b}) becomes
\be
\label{Kolmogorov1c}
\frac{d}{dt}|\psi(t)\ra  =  \left( -\sum_{i=1}^N \lambda_i a_i^\dag a_i  +  \sqrt{2\lambda_i/q} \sum_{i=1}^N f_i(x) a_i\right)|\psi(t)\ra,
\ee
where $|\psi(t)\ra \in \calF$ 
and $f_i(x)$ is expressed in terms of creation/annihilation operators using the identity
\be
\label{x_vs_a}
x_i = \sqrt{q/(2\lambda_i)} (a_i+a_i^\dag)
\ee
The correspondence between the first and the second quantization pictures is established by $\psi(t,x)=\sum_{\mm\in \ZZ_{\ge 0}^N}\;  \la \mm|\psi(t)\ra \Phi_\mm(x)$,
where $\Phi_\mm(x)$ are the eigenfunctions for the considered system of harmonic oscillators. An explicit expression for $\Phi_\mm(x)$  
in terms of Hermite polynomials will be needed only later in Section~\ref{sec:readout}.
Recalling that  $f_i(x)=\sum_{j=1}^N b_{ij} x_j + \sum_{j,k=1}^N c_{ijk} x_j x_k$ 
and expressing each variable $x_i$ in terms of creation/annihilation operators using Eq.~(\ref{x_vs_a})
one arrives at 
\be
\label{Kolmogorov1d}
\frac{d}{dt}|\psi(t)\ra  =  (-A+B+C)|\psi(t)\ra,
\ee
where
\be
\label{ABexplicit}
A = \sum_{i=1}^N \lambda_i a_i^\dag a_i, \qquad B =  \sum_{i,j=1}^N b_{ij} (\lambda_i/\lambda_j)^{1/2}  \, a_j^\dag a_i,
\ee
and
\be
\label{Csimplified}
C = (q/2)^{1/2}\sum_{i,j,k=1}^N   (\lambda_i\lambda_j  \lambda_k)^{-1/2}  \lambda_i  c_{ijk} (a_j+a_j^\dag) (a_k+a_k^\dag) a_i.
\ee
Obviously, $A$ is hermitian and positive semi-definite. 
We claim that $B$ and $C$ are skew-hermitian whenever the original ODE Eq.~(\ref{ODE}) is divergence-free.
\begin{lemma}
Suppose the ODE Eq.~(\ref{ODE}) is divergence-free. Then 
$B^\dag=-B$ and $C^\dag=-C$. One has 
\be
\label{Csimplified}
C=
(q/2)^{1/2}\sum_{i,j,k=1}^N   (\lambda_i\lambda_j  \lambda_k)^{-1/2}   \lambda_i  c_{ijk} (a_j^\dag a_k^\dag a_i -  a_i^\dag a_k a_j).
\ee
\end{lemma}
\begin{proof}
Indeed, from $b_{ii}=0$ and  $\lambda_i b_{ij}+\lambda_j b_{ji}=0$ one infers that the matrix $b_{ij} (\lambda_i/\lambda_j)^{1/2}$ is skew-symmetric, 
which implies $B^\dag = -B$.
Recall that 
$c_{ijk}=0$ unless all indices $i,j,k$ are distinct. Hence $C$ is a cubic multi-linear polynomial in the creation/annihilation operators.
Suppose $i\ne j \ne k$. The coefficient in front of $a_i a_j a_k$ in the formula for $C$
is proportional to the product $\lambda_i c_{ijk}$ symmetrized over all permutations of indices. From the divergence-free condition
$\lambda_i c_{ijk} + \lambda_j c_{jki} + \lambda_k c_{kij}=0$ one infers that the coefficient in front of $a_i a_j a_k$ is zero. 
Using the symmetry $c_{ijk}=c_{ikj}$ one gets 
\[
C = (q/2)^{1/2}\sum_{i,j,k=1}^N   (\lambda_i\lambda_j  \lambda_k)^{-1/2}  \lambda_i  c_{ijk} (a_j^\dag a_k^\dag + 2 a_j a_k^\dag) a_i.
\]
Replacing $2\lambda_i c_{ijk} a_j a_k^\dag a_i$ by $(\lambda_i c_{ijk}  + \lambda_j c_{jik})a_j a_k^\dag a_i $ under the sum
(which is justified since we sum over all pairs $i\ne j$), using the identity $c_{jik}=c_{jki}$ and  $\lambda_i c_{ijk} + \lambda_j c_{jki} = -\lambda_k c_{kij}$
proves Eq.~(\ref{Csimplified}).
\end{proof}

Our goal is to solve the Kolmogorov equation in the form Eq.~(\ref{Kolmogorov1d}). Since the vacuum state $|{\bf 0}\ra$ is annihilated by all
terms in the righthand side of Eq.~(\ref{Kolmogorov1d}), time evolution of the vacuum amplitude  is trivial, that is,
$\la {\bf 0}|\psi(t)\ra=\la {\bf 0}|\psi(0)\ra$ for all $t$.  Instead of keeping track of  this constant amplitude in all equations, we find it convenient
to work in the subspace of the Fock space $\calF$  orthogonal to the vacuum. Let this subspace be
\be
\calF_0 = \{ |\psi\ra \in \calF \, : \,  \la {\bf 0}|\psi\ra=0 \}.
\ee
Note that $\la \psi|A|\psi\ra \ge \lambda_1 \la \psi|\psi\ra$ for all $|\psi\ra \in \calF_0$.
We assume that the initial state $|\psi(0)\ra$ belongs to $\calF_0$. 
This can always be achieved by subtracting a multiple of $|{\bf 0}\ra$ from $|\psi(0)\ra$.

\subsection{Key technical lemmas}
\label{sec:ABCstructural}

In this section we establish  structural properties of the operators  $A,B,C$ that appear in the second-quantized version of the Kolmogorov equation.

\begin{lemma}
\label{lemma:AC1}
For any vectors $\psi,\phi\in \calF_0$ one has 
\be
\label{C_overlap_upper}
|\la \phi |B+C|\psi\ra |\le  \gamma \sqrt{\la \psi|A|\psi\ra \la \phi|A^2|\phi\ra}
\ee
where $\gamma = 2q^{1/2} sJ \lambda_1^{-2} (1+ \lambda_1^{1/2} q^{-1/2})$.
\end{lemma}
{\em Comment:} we  implicitly assume that $\la \psi|A|\psi\ra<\infty$ and $\la \phi|A^2|\phi\ra<\infty$ since otherwise Eq.~(\ref{C_overlap_upper})
is meaningless. The same comment applies to all other inequalities stated below: all expected value that include $A$
or powers of $A$ are assumed to be finite. 
We defer the proof of Lemma~\ref{lemma:AC1}  to Appendix~\ref{app:lemmaAC1proof}.
Given operators $X$ and $Y$ acting on the same Hilbert space, let 
 $[X,Y]\equiv XY-YX$.
\begin{lemma}
\label{lemma:commutator_order1}
For any vector $\psi\in \calF_0$ one has
\be
\label{commutator_bound_order1}
\la \psi|[A,B+C]|\psi\ra\le 
2Js \lambda_1^{-1/2} 
(1+ 4 q^{1/2}  \lambda_1^{-1/2}  ) \sqrt{ \la \psi|A|\psi\ra \la \psi |A^2|\psi\ra}.
\ee
\end{lemma}
Note that $[A,B+C]$ is a  hermitian operator since $A$ is hermitian while $B+C$ is anti-hermitian. Thus $\la \psi| [A,B+C]|\psi\ra$  is a real number
and Eq.~(\ref{commutator_bound_order1}) is well defined. 
We prove Lemma~\ref{lemma:commutator_order1} in Appendix~\ref{app:B}.
Note that the naive estimate $\la \psi|[A,B+C]|\psi\ra\le 2 |\la \psi |A (B+C)|\psi\ra|$ combined with Lemma~\ref{lemma:AC1}
would also results in the upper bound 
on $\la \psi|[A,B+C]|\psi\ra$. However, this bound is considerably weaker than Eq.~(\ref{commutator_bound_order1}).
We shall often use the following implication of Lemma~\ref{lemma:commutator_order1}.
\begin{corol}[\bf Commutator bound]
\label{corol:commutator_bound_order1}
For any vector $\psi\in \calF_0$ one has
\be
\label{commutator_kappa1}
\la \psi|[A,B+C]|\psi\ra\le  \la \psi |A^2|\psi\ra + \kappa_1 \la \psi|A|\psi\ra,
\ee
where $\kappa_1=2J^2 s^2 \lambda_1^{-1}(1+16q \lambda_1^{-1})$.
\end{corol}
\begin{proof}
For any real numbers $\alpha_1,\alpha_2\ge 0$, and $\beta>0$ one has 
$\sqrt{\alpha_1 \alpha_2}\le (1/2)(\beta \alpha_2 + \beta^{-1} \alpha_1)$. Choose  $\alpha_1= \la \psi|A|\psi\ra$ and $\alpha_2=\la \psi|A^2|\psi\ra$.
Choose $\beta$ such that $\beta Js \lambda_1^{-1/2} 
(1+ 4 q^{1/2}  \lambda_1^{-1/2}  )=1$. Now  Eq.~(\ref{commutator_kappa1}) follows from Eq.~(\ref{commutator_bound_order1}). 
\end{proof}
We shall need a generalization  of the bound Eq.~(\ref{commutator_kappa1}) for commutators $[A^p,B+C]$ with $p\ge 2$.

\begin{lemma}[\bf Order-$p$ commutator bound]
\label{lemma:AC2}
For any integer $p\ge 2$ there exists a real number  $\kappa_p\le poly(s,J,\lambda_1^{-1},\lambda_N)$ 
such that  for any vector $\psi \in \calF_0$ one has 
\be
\label{commutator_bound}
\la \psi| [A^p,B+C]|\psi\ra \le \la \psi | A^{p+1}|\psi\ra  +  \kappa_{p} \la \psi| A^p  |\psi\ra.
\ee
\end{lemma}
We defer the proof to Appendix~\ref{app:lemmaAC2proof}.
The bound Eq.~(\ref{commutator_bound}) is rather weak as it depends on the {\em largest} dissipation rate $\lambda_N$, which often scales as a positive power of $N$.
Luckily, we only need Lemma~\ref{lemma:AC2} to prove that the Kolmogorov equation has a well-defined solution.
The scaling of $\kappa_p$ with $p\ge 2$  is irrelevant for the runtime of our quantum algorithm, which depends on the coefficients $\gamma$ and $\kappa_1$.

Finally, we shall need the following very loose upper bound. 
\begin{lemma}
\label{lemma:loose}
Let $p\ge 0$ be a real number. 
 There exist a real number $\omega_p<\infty$ such that  
\be
\label{ABCloose}
\| A^p(B+C)\psi\| \le \omega_p \| A^{p+3/2}\psi\|
\ee
for any vector $\psi \in \calF_0$. Here $\omega_p$ depends on $p$, $N$, $J$, $\lambda_1$, and $\lambda_N$. 
\end{lemma}
We defer the proof to Appendix~\ref{app:D}.

\subsection{Regularization}
\label{sec:regul-sz}

In this section we prove that the Kolmogorov equation has a well-defined solution. Our proof is constructive in the sense that it provides 
a regularization procedure approximating the exact solution 
along with  rigorous bounds on the approximation error. The regularized Kolmogorov equation can be simulated on a quantum computer using
sparse Hamiltonian simulation methods. 
Define a set of multi-indices
\[
\calJ = \{\mm =(m_1,\ldots,m_N) \in \ZZ_{\ge 0}^N \, : \, m_1+\ldots+m_N\ge 1\}.
\]
By definition, any vector in the Fock space $\calF_0$ is a linear combination of basis states $|\mm\ra$ with $\mm \in \calJ$.
Given a real cutoff parameter $k>0$,  let
\[
\calJ_k = \{\mm \in \calJ \, : \, \sum_{i=1}^N \lambda_i m_i  \le k\}.
\]
In other words, $\mm \in \calJ_k$ iff $\la \mm|A|\mm\ra\le k$.
Note that  $\calJ_k$ is a finite set (possibly empty) for any $k>0$. Indeed, $\mm \in \calJ_k$ implies $|\mm|\equiv \sum_{i=1}^N m_i \le \lceil k/\lambda_1\rceil$.
It is well known that the number of partitions of an integer $L$  into a sum of $N$ non-negative integers is ${L+N-1\choose L}$.
Substituting $L=0,1,\ldots, \lceil k/\lambda_1 \rceil$ gives $|\calJ_k|\le O(N^{k/\lambda_1})$. Define a projector
\[
\Pi_k = \sum_{\mm \in \calJ_k} |\mm\ra\la \mm|
\]
and a subspace
\[
\calH_k = \mathrm{span}(|\mm\ra\, : \, \mm \in \calJ_k) \subseteq \calF_0.
\]
We shall refer to $\calH_k$ as a {\em low-dissipation subspace}.
Define the regularized Kolmogorov equation as 
\be
\label{KE_k}
\frac{d}{dt} |\psi_k(t)\ra = (-A + \Pi_k (B+C)\Pi_k ) |\psi_k(t)\ra
\ee
with the initial condition $|\psi_k(0)\ra = |\psi(0)\ra$, where $|\psi(0)\ra\in \calF_0$ is the initial condition for the original (unregularized) KE. 
We emphasize that all regularized KEs are defined on the same Hilbert space, namely the infinite-dimensional Fock space $\calF_0$.
However the operator $B+C$ is "turned on" only on the low-dissipation subspace $\calH_k$, which is finite-dimensional.
The orthogonal complement of $\calH_k$ evolves only under the diagonal operator $-A$ 
so that $\la \mm |\psi_k(t)\ra = e^{-\lambda_\mm t } \la \mm|\psi(0)\ra$ for all $\mm \notin \calJ_k$, where $\lambda_\mm = \sum_{i=1}^N \lambda_i m_i$.
  Using these observations
one can easily check that Eq.~(\ref{KE_k}) has an infinitely differentiable solution $|\psi_k(t)\ra \in \calF_0$.
The main result of this section is the following. 
\begin{theorem}[\bf Regularization]
\label{thm:regul}
Consider any vector $\psi_0\in \calF_0$ such that 
 $\la \psi_0|A^8|\psi_0\ra<\infty$. There exists a differentiable function $\psi\, : \, \RR_{\ge 0} \to \calF_0$ 
 with $\psi(0)=\psi_0$  such that for any evolution time $t\ge 0$ 
 one has $\la \psi(t)|A^3|\psi(t)\ra<\infty$ and $\psi(t)$ solves the Kolmogorov equation\footnote{More formally, 
 $\| \psi(t+\delta) - \psi(t) - \delta (-A+B+C)\psi(t)\| = O(\delta^2)$ in the limit $\delta\to 0$, where the norm is the usual $2$-norm of the Fock
 space $\calF_0$.}

 \[
 \frac{d}{dt}\psi(t)=(-A+B+C)\psi(t).
 \]
Furthermore,   for any cutoff parameter  $k>0$ and for all $t\ge 0$
\be
\label{regul_bound1}
\| \psi_k(t) - \psi(t)\|^2 \le  \frac{12 \gamma}{ k^{1/2}} \left(\la \psi_0|A|\psi_0\ra  + \kappa_1 \la\psi_0|\psi_0\ra\right).
\ee
Here $\gamma$ and $\kappa_1$ are the coefficients defined in Section~\ref{sec:ABCstructural}.
\end{theorem}
Note that the bound $\la \psi(t)|A^3|\psi(t)\ra<\infty$ implies $\| (B+C)\psi(t)\|<\infty$ due to Lemma~\ref{lemma:loose} and thus
$(-A+B+C)\psi(t)\in \calF_0$, that is, $\psi(t)$ belongs to the domain of the operator $-A+B+C$.
In the rest of this section we prove the theorem.
\begin{proof}
For a fixed cutoff parameter $k>0$  define moment functions
\[
f_p(t) =\la \psi_k(t)|A^p|\psi_k(t)\ra,
\]
where $p$ is a non-negative integer. 
We shall need uniform (independent of $k$) upper bounds on $f_p(t)$ 
established
in the following lemma.
\begin{lemma}[\bf Moment bounds]
\label{lemma:moment_bounds}
For any integer  $p\ge 0$ and evolution time $t\ge 0$ one has 
\be
\label{f_moment_p}
f_p(t) \le  \max{(f_p(0), (\kappa_p)^p f_0(0))},
\ee
where $\kappa_0=1$ and $\kappa_p$ with $p\ge 1$ are the coefficients defined in Section~\ref{sec:ABCstructural}.
Furthermore, 
\be
\label{f_moment_p_integral}
\int_0^t d\tau\, f_p(t) \le f_{p-1}(0) + \kappa_{p-1} f_{p-2}(0) + \kappa_{p-2} \kappa_{p-1} f_{p-3}(0) + \ldots + (\kappa_1 \kappa_2 \cdots \kappa_{p-1}) f_0(0)
\ee
for $p\ge 1$.
\end{lemma}
\begin{proof}
From the equation of motion Eq.~(\ref{KE_k}) one gets
\[
\frac{d}{dt} f_p(t) = -2f_{p+1}(t) - \la \psi_k(t)| [A^p, \Pi_k (B+C)\Pi_k  ] |\psi_k(t)\ra = -2f_{p+1}(t) -\la \psi_k(t)|\Pi_k [A^p,B+C] \Pi_k  |\psi_k(t)\ra.
\]
Here we used the fact that $A$ is hermitian while $B$ and $C$ are anti-hermitian. 
Also we used the identity $A\Pi_k=\Pi_k A$. 
Suppose $p=0$. Then the term with the commutator disappears, that is, $(d/dt)f_0(t) = -2 f_1(t)$.
Since $f_1(t)\ge 0$, this implies $f_0(t)\le f_0(0)$, proving Eq.~(\ref{f_moment_p}) with $p=0$.
Integrating $(d/dt)f_0(t) = -2 f_1(t)$ gives
\[
\int_0^t f_1(t) = \frac12( f_0(0)-f_0(t)) \le \frac12 f_0(0) \le f_0(0)
\]
proving Eq.~(\ref{f_moment_p_integral}) with $p=1$.
Suppose $p \ge 1$.
Applying the commutator bound of Corollary~\ref{corol:commutator_bound_order1} and Lemma~\ref{lemma:AC2} with $\psi = \Pi_k \psi_k(t)$ one gets
\[
|\la \psi_k(t)|\Pi_k [A^p,B+C] \Pi_k  |\psi_k(t)\ra| \le \la \psi_k(t)| A^{p+1} + \kappa_p A^p |\psi_k(t)\ra = f_{p+1}(t) + \kappa_p f_p(t).
\]
Hence
\be
\label{f_p_derivative_eq1}
\frac{d}{dt} f_p(t) \le  -f_{p+1}(t)  + \kappa_p f_p(t).
\ee
Jensen's inequality gives
\[
f_{p+1}(t) \ge (f_0(t))^{-1/p} (f_p(t))^{(p+1)/p} \ge (f_0(0))^{-1/p} (f_p(t))^{(p+1)/p},
\]
where the  second inequality follows from $f_0(t)\le f_0(0)$.
Thus
\[
\frac{d}{dt} f_p(t) \le -  (f_0(0))^{-1/p} (f_p(t))^{(p+1)/p} + \kappa_p f_p(t).
\]
The derivative of $f_p(t)$ is  negative whenever $f_p(t)> (\kappa_p)^p f_0(0)$.
Hence, if  $f_p(0)\le (\kappa_p)^p f_0(0)$ then $f_p(t)\le (\kappa_p)^p f_0(0)$ for all $t\ge 0$. Otherwise, if $f_p(0)>(\kappa_p)^p f_0(0)$, then
$f_p(t)$ is monotone decreasing until the first time moment $t$ such that  $f_p(t)\le (\kappa_p)^p f_0(0)$ and this bound holds for all subsequent times.
This proves Eq.~(\ref{f_moment_p}).
From Eq.~(\ref{f_p_derivative_eq1}) one gets
\[
\int_0^t d\tau\, f_{p+1}(\tau) \le f_p(0) +  \kappa_p \int_0^t d\tau \, f_p(\tau).
\]
Applying this recursively proves Eq.~(\ref{f_moment_p_integral}).
\end{proof}
As a corollary of Lemma~\ref{lemma:moment_bounds} one infers that solutions of the regularized KE  have most of their mass
on the low-dissipation subspace, as formally stated below. 
\begin{corol}[\bf Projection error]
For any cutoff parameters $k,\ell >0$ and evolution time $t\ge 0$ one has
\be
\label{projection_error}
\| (I-\Pi_k)\psi_\ell(t)\|^2 \le \frac{\la \psi(0)|A|\psi(0)\ra + \kappa_1 \la \psi(0)|\psi(0)\ra }{ k}.
\ee
Here $\kappa_1$ is the coefficient defined in Section~\ref{sec:ABCstructural}.
\end{corol}
\begin{proof}
An operator inequality $I-\Pi_k \le k^{-1} A$ gives 
\be
\la \psi_\ell(t) | I- \Pi_k|\psi_\ell(t)\ra \le   k^{-1} \la \psi_\ell(t) |A|\psi_\ell(t)\ra \le  k^{-1} \cdot \max{(f_1(0), \kappa_1f_0(0))}.
\ee
Here the second inequality follows from Lemma~\ref{lemma:moment_bounds} with $p=1$.
This implies Eq.~(\ref{projection_error}).
\end{proof}

Fix a pair of cutoff parameters  $0<k<\ell$ and define  an "error state"
\be
|e(t)\ra = |\psi_k(t)\ra - |\psi_\ell(t)\ra.
\ee
Our strategy is to prove that the norm of $e(t)$ goes to zero in the limit $k,\ell\to \infty$.
For technical reasons, we shall also need an upper bound on the norm of $A^{p/2}e(t)$ for integers $p=O(1)$.
To this end define moments
\be
h_p(t) = \la e(t)|A^p|e(t)\ra.
\ee
We shall be primarily concerned with $h_0(t)$ which controls the rate of convergence of solutions $\psi_k(t)$ in the limit $k\to \infty$.
Higher order moments $h_p(t)$ with $p\ge1$ are only needed to prove that 
$\psi_k(t)$ "converges to the right thing" in the sense that the limiting point of the sequence $\psi_k(t)$ solves the Kolmogorov equation (the unregularized one).
 In particular, the runtime of our quantum algorithm only depends on the upper bound on $h_0(t)$.
\begin{lemma}
\label{lemma:regularization_p}
For  any  cutoff parameters $0< k<\ell$ and evolution time $t\ge 0$ one has 
\be
\label{h0_upper_bound}
h_0(t)\le 
\frac{12 \gamma}{k^{1/2}} (\la \psi(0)|A|\psi(0)\ra  + \kappa_1 \la\psi(0)|\psi(0)\ra).
\ee
For any integer $p\ge 1$ 
\be
\label{hp_upper_bound}
h_p(t) \le \frac{e^{2\kappa_p t}}{k^{1/2}} poly(s,J,\lambda_1^{-1},\lambda_N)\max_{0\le p'\le 2p+2} \la \psi(0)|A^{p'}|\psi(0)\ra.
\ee
Here  $\gamma$ and $\kappa_p$ are the coefficients defined in Section~\ref{sec:ABCstructural}.
\end{lemma}
\begin{proof}
From the equation of motion Eq.~(\ref{KE_k}) 
and the analogous equation for $\psi_\ell(t)$
one gets
\[
\frac{d}{dt} |e(t)\ra =-A |e(t)\ra  + \Pi_k (B+C)\Pi_k |\psi_k(t)\ra - \Pi_\ell (B+C)\Pi_\ell |\psi_\ell(t)\ra
\]
and
\[
\frac{d}{dt} \la e(t)| =- \la e(t)|A  - \la \psi_k(t)| \Pi_k (B+C)\Pi_k + \la \psi_\ell(t)| \Pi_\ell (B+C)\Pi_\ell.
\]
Here we noted that $A$ is hermitian while $B+C$ is anti-hermitian. 
It follows that
\[
\frac12 \frac{d}{dt} h_p(t) = -h_{p+1}(t) +  \la e(t) | A^p  \Pi_k (B+C)\Pi_k |\psi_k(t)\ra - \la e(t) | A^p  \Pi_\ell (B+C)\Pi_\ell |\psi_\ell(t)\ra.
\]
Clearly, $k<\ell$ implies $\calJ_k\subseteq \calJ_\ell$. Thus one can write
\[
\Pi_\ell = \Pi_k + \Pi_k^\perp, \qquad \Pi_k^\perp \equiv  \sum_{\mm \in \calJ_\ell\setminus \calJ_k} \; |\mm\ra\la \mm|.
\]
Then
\begin{align*}
\frac12 \frac{d}{dt} h_p(t) & = -h_{p+1}(t)  +  \la e(t) | A^p  \Pi_k (B+C)\Pi_k |e(t)\ra
-   \la e(t) | A^p  \Pi_k^\perp  (B+C)\Pi_k |\psi_\ell(t)\ra \\
& -  \la e(t) | A^p  \Pi_k  (B+C)\Pi_k^\perp |\psi_\ell(t)\ra  -  \la e(t) | A^p  \Pi_k^\perp  (B+C)\Pi_k^\perp |\psi_\ell(t)\ra.
\end{align*}
Note that
\[
 \la e(t) | A^p  \Pi_k (B+C)\Pi_k |e(t)\ra =\frac12  \la e(t) | \Pi_k  [A^p,B+C] \Pi_k |e(t)\ra
 \]
for $p\ge 1$. Indeed,  $(1/2)  \Pi_k  [A^p,B+C] \Pi_k$ is the real (hermitian) part of the operator $A^p  \Pi_k (B+C)\Pi_k$.
Applying the commutator bound of Corollary~\ref{corol:commutator_bound_order1} and Lemma~\ref{lemma:AC2} with $\psi=\Pi_k e(t)$ one gets
\[
  \la e(t) | \Pi_k  [A^p,B+C] \Pi_k |e(t)\ra \le \la e(t)| A^{p+1} + \kappa_p A^p |e(t)\ra.
 \]
 We arrive at 
\begin{align}
\frac12 \frac{d}{dt} h_p(t) &  \le 
\kappa_p h_p(t)+
  |\la e(t) | A^p  \Pi_k^\perp  (B+C)\Pi_k |\psi_\ell(t)\ra|\nonumber \\ 
  & + 
|\la e(t) | A^p  \Pi_k  (B+C)\Pi_k^\perp |\psi_\ell(t)\ra| 
  +  |\la e(t) | A^p  \Pi_k^\perp  (B+C)\Pi_k^\perp |\psi_\ell(t)\ra|. \label{h_p_derivative}
\end{align}
Here the term $\kappa_p h_p(t)$ disappears for $p=0$ since $[A^0,B+C]=0$.
 We shall bound each term that depends on $B+C$ 
 using Lemma~\ref{lemma:AC1} and  the moment bounds of Lemma~\ref{lemma:moment_bounds}.
Applying Lemma~\ref{lemma:AC1} with $|\psi\ra = \Pi_k^\perp A^p |e(t)\ra$ and $|\phi\ra = \Pi_k |\psi_\ell(t)\ra$
gives 
\[
|\la e(t) | A^p  \Pi_k^\perp  (B+C)\Pi_k |\psi_\ell(t)\ra| \le \gamma \sqrt{\la e(t)| \Pi_k^\perp A^{2p+1} |e(t)\ra
\la \psi_\ell(t)|A^2 |\psi_{\ell}(t)\ra}.
\]
Here $\gamma$ is
defined in Lemma~\ref{lemma:AC1}.
The restriction of $A$ onto the range of $\Pi_k^\perp$ has eigenvalues at least $k$ (and at most $\ell$).
Hence 
\[
\la e(t)| \Pi_k^\perp A^{2p+1} |e(t)\ra \le k^{-1}  \la e(t) | A^{2p+2}|e(t)\ra.
\]
Using the definition $e(t) = \psi_k(t)-\psi_\ell(t)$  and
 triangle inequality one gets
\[
 \sqrt{\la e(t)|  A^{2p+2} |e(t)\ra} \le
 \sqrt{\la \psi_k(t)|  A^{2p+2} |\psi_k(t)\ra} 
+\sqrt{\la \psi_\ell(t)|  A^{2p+2} |\psi_\ell(t)\ra}.
\]
We arrive at 
\[
|\la e(t) | A^p  \Pi_k^\perp  (B+C)\Pi_k |\psi_\ell(t)\ra| \le  2\gamma k^{-1/2} \cdot  \sqrt{f_{2p+2}(t) f_2(t)}.
\]
By a slight abuse of notations, here we ignore the dependence of the moment functions $f_{2p+2}(t)$ and $f_2(t)$ on the
cutoff  parameters $k$ and $\ell$
(this is justified since the upper bound of Lemma~\ref{lemma:moment_bounds} is uniform in $k$). 
Applying exactly the same argument to the last term in Eq.~(\ref{h_p_derivative})  gives
\[
 |\la e(t) | A^p  \Pi_k^\perp  (B+C)\Pi_k^\perp |\psi_\ell(t)\ra|\le 2\gamma k^{-1/2} \cdot  \sqrt{f_{2p+2}(t) f_2(t)}.
\]
It remains to bound the term $|\la e(t) | A^p  \Pi_k  (B+C)\Pi_k^\perp |\psi_\ell(t)\ra|$ 
 in Eq.~(\ref{h_p_derivative}).
Applying Lemma~\ref{lemma:AC1} with $|\phi\ra = \Pi_k  A^p |e(t)\ra$ and $|\psi\ra = \Pi_k^\perp |\psi_\ell(t)\ra$
gives 
\[
|\la e(t) | A^p  \Pi_k  (B+C)\Pi_k^\perp |\psi_\ell(t)\ra|\le \gamma \sqrt{ \la e(t)|A^{2p+2} |e(t)\ra \la \psi_\ell(t)| A \Pi_k^\perp |\psi_\ell(t)\ra}.
\]
The same arguments as above give $ \la \psi_\ell(t)| A \Pi_k^\perp |\psi_\ell(t)\ra\le k^{-1} f_2(t)$ and
$ \sqrt{\la e(t)|A^{2p+2} |e(t)\ra} \le 2\sqrt{f_{2p+2}(t)}$. Hence 
\[
|\la e(t) | A^p  \Pi_k  (B+C)\Pi_k^\perp |\psi_\ell(t)\ra|\le2\gamma k^{-1/2} \cdot  \sqrt{f_{2p+2}(t) f_2(t)}.
\]
Combining the above bounds and Eq.~(\ref{h_p_derivative})   gives
\be
\label{h_0_derivative1}
\frac{d}{dt} h_0(t) \le 12 \gamma k^{-1/2} f_2(t)
\ee
and
\be
\label{h_p_derivative1}
\frac{d}{dt} h_p(t)\le 2\kappa_p h_p(t) + 12   \gamma k^{-1/2} \sqrt{f_{2p+2}(t) f_2(t)}
\ee
for $p\ge 1$. Consider first the case $p=0$. Integrating Eq.~(\ref{h_0_derivative1}), recalling that $h_0(0)=0$,  and using the second part of Lemma~\ref{lemma:moment_bounds}
gives
\begin{align}
h_0(t) & \le  12 \gamma k^{-1/2} \int_0^t d\tau\,  f_2(\tau)
\le  12 \gamma k^{-1/2} (f_1(0) + \kappa_1 f_0(0)) \nonumber \\
& =  12 \gamma k^{-1/2} (\la \psi(0)|A|\psi(0)\ra  + \kappa_1 \la\psi(0)|\psi(0)\ra). 
\end{align}
Suppose $p\ge 1$. Lemma~\ref{lemma:moment_bounds} implies that
$ \sqrt{f_{2p+2}(t) f_2(t)}\le poly(J,s,\lambda_1^{-1},\lambda_N)\max_{0\le p'\le 2p+2} f_{p'}(0)$ for all $t\ge0$.
Hence Eq.~(\ref{h_p_derivative1}) and  Gronwall inequality gives
\be
h_p(t) \le \frac{e^{2\kappa_p t}}{k^{1/2}} poly(s,J,\lambda_1^{-1},\lambda_N)\max_{0\le p'\le 2p+2} f_{p'}(0).
\ee
This is equivalent to the statement of the lemma.
\end{proof}

From Lemma~\ref{lemma:regularization_p} we derive the following. 
\begin{corol}
\label{corol:limit}
Let $p\ge 1$ be an integer such that $\la \psi(0)|A^{p'}|\psi(0)\ra<\infty$ for all $p'\in [0,2p+2]$.
Then for any $t\ge 0$ there exists a vector function  $\psi \, : \, [0,t]\to   \calF_0$ such that 
\be
\label{convergence_bound_with_A}
\lim_{k\to \infty} \,\max_{\tau \in [0,t]} \| A^{p/2} (\psi_k(\tau) - \psi(\tau))\|=0.
\ee
Furthermore, for all $t\ge 0$
\be
\label{regularization_error}
\| \psi_k(t) - \psi(t)\|^2 \le  \frac{12 \gamma}{ k^{1/2}} (\la \psi(0)|A|\psi(0)\ra  + \kappa_1 \la\psi(0)|\psi(0)\ra)
\ee
Here $\gamma$ and $\kappa_1$ are the coefficients defined in Section~\ref{sec:ABCstructural}.
\end{corol}
\begin{proof}
Consider  a sequence of vector functions $\{A^{p/2} \psi_k(\tau) \}$ mapping $[0,t]$ to $\calF_0$.
For simplicity here we restrict ourselves to integer cutoff  parameters $k\ge 1$.
We claim that the sequence $\{A^{p/2} \psi_k(\tau) \}$  is uniformly  Cauchy.
Indeed,  for any $1\le k<\ell$ one has
\[
\| A^{p/2}\psi_k(\tau) - A^{p/2}\psi_\ell(\tau)\|^2 = \la e(\tau)|A^p |e(\tau)\ra=h_p(\tau),
\]
 where $e(\tau)=\psi_k(\tau)-\psi_\ell(\tau)$
is the error state and $h_p(\tau)$ decays as $1/\sqrt{k}$ in the limit $k \to \infty$ due to  Eq.~(\ref{hp_upper_bound}) of Lemma~\ref{lemma:regularization_p}.
Moreover, the upper bound on $h_p(\tau)$ is uniform with respect to $\tau$ (it depends only on $t$).
Thus $\{A^{p/2} \psi_k(\tau) \}$  is uniformly  Cauchy.
Since $\calF_0$ is the complete metric space, the sequence $\{A^{p/2} \psi_k(\tau)\}$
must have a limit 
\[
\theta_p(\tau)=\lim_{k\to \infty} A^{p/2} \psi_k(\tau) \in \calF_0.
\]
For each $\tau \in [0,t]$ define
\[
\psi(\tau) = A^{-p/2} \theta_p(\tau) \in \calF_0.
\]
This is well-defined  since $A^{-1/2}$ is a bounded operator on $\calF_0$ (note that eigenvalues of $A^{-1/2}$ are at most $\lambda_1^{-1/2}$).
Then $\theta_p(\tau)=A^{p/2} \psi(\tau)$ and 
Eq.~(\ref{convergence_bound_with_A}) is equivalent to
\[
\lim_{k\to \infty} \,\max_{\tau \in [0,t]} \| A^{p/2} \psi_k(\tau) - \theta_p(\tau)\|=0
\]
which follows from the fact that $\{A^{p/2} \psi_k(\tau) \}$  is uniformly  Cauchy and $\theta_p(\tau)$ is its limiting point. 
Since $A^{-1/2}$ is a bounded operator on $\calF_0$, one gets
\[
\psi(\tau) = A^{-p/2} \lim_{k\to \infty} A^{p/2} \psi_k(\tau) = \lim_{k\to \infty} \psi_k(\tau).
\]
For any $k\ge 1$ one has
\[
\| \psi(\tau)-\psi_k(\tau)\|^2 = \lim_{\ell \to \infty} \| \psi_\ell(\tau) - \psi_k(\tau)\|^2 = h_0(\tau).
\]
Thus the bound Eq.~(\ref{regularization_error}) follows from Eq.~(\ref{h0_upper_bound}) of Lemma~\ref{lemma:regularization_p}.
\end{proof}
\begin{lemma}
\label{lemma:KEsolution}
Suppose $\la \psi(0)|A^8|\psi(0)\ra<\infty$. Then
the vector function $\psi\, : \, \RR_{\ge 0} \to \calF_0$ constructed in Corollary~\ref{corol:limit} is differentiable and solves the Kolmogorov equation, that is,
\be
\label{KErestated}
\frac{d}{dt} \psi(t) = (-A+B+C)\psi(t)
\ee
for all $t\ge 0$.
\end{lemma}
\begin{proof}
We can apply Corollary~\ref{corol:limit} with $p=3$ since $\la \psi(0)|A^8|\psi(0)\ra<\infty$ (which implies $\la \psi(0)|A^{p'}|\psi(0)\ra<\infty$ for all $0\le p'\le 8=2p+2$).
From Eq.~(\ref{convergence_bound_with_A}) one gets
\[
\la \psi(t)|A^3|\psi(t)\ra =\lim_{k\to \infty} \la \psi_k(t)|A^3|\psi_k(t)\ra \le f_3(t),
\]
where $f_3(t)$ is the moment function upper bounded by Lemma~\ref{lemma:moment_bounds}.
Thus $\la \psi(t)|A^3|\psi(t)\ra<\infty$ which implies $\| (-A+B+C)\psi(t)\|<\infty$ due to Lemma~\ref{lemma:loose}.
Hence $(-A+B+C)\psi(t)\in \calF_0$.
We need to show that 
\be
\label{time_step_delta_eq1}
\| \psi(t+\delta) - \psi(t) - \delta (-A+B+C)\psi(t)\| = O(\delta^2)
\ee
in the limit $\delta\to 0$.
Consider any  $k\ge 1$. Since $\psi_k(t)$ is infinitely differentiable, one can write
\begin{align}
\label{time_step_delta_eq2}
\| \psi_k(t+\delta) - \psi_k(t) - \delta (-A+\Pi_k(B+C)\Pi_k)\psi_k(t)\| & =\left\| 
\int_t^{t+\delta} d\tau_1 \,  \left( \frac{d}{d\tau_1} \psi_k(\tau_1) - \frac{d}{dt}\psi_k(t)\right) \right\| \nonumber \\
&\le  \int_t^{t+\delta} d\tau_1 \, \int_t^{\tau_1} d\tau_2 \left\| \frac{d^2}{d\tau_2^2} \psi_k(\tau_2)\right\|
\end{align}
From the equation of motion Eq.~(\ref{KE_k}) one gets
\[
\frac{d^2}{d\tau_2^2} \psi_k(\tau_2) = \left(-A+\Pi_k(B+C)\Pi_k\right)^2 \psi_k(\tau_2).
\]
Let $\varphi \equiv \left(-A+\Pi_k(B+C)\Pi_k\right)\psi_k(\tau_2)$. Then
\[
\| \left( -A+\Pi_k(B+C)\Pi_k \right) \varphi \| \le \| A\varphi\| + \|(B+C)\Pi_k \varphi\| \le \|A \varphi\| + \omega_0 \| A^{3/2} \varphi\|
\le (\lambda_1^{-1/2} + \omega_0) \| A^{3/2} \varphi\|.
\]
Here the second inequality uses Lemma~\ref{lemma:loose}.
The definition of $\varphi$  gives
\[
 \| A^{3/2} \varphi\| \le \| A^{5/2} \psi_k(\tau_2)\| + \| A^{3/2} (B+C) \Pi_k \psi_k(\tau_2)\|.
\]
The last term can be upper bounded using Lemma~\ref{lemma:loose} which gives
\[
\| A^{3/2} (B+C) \Pi_k \psi_k(\tau_2)\|\le \omega_{3/2} \| A^3 \psi_k(\tau_2)\|.
\]
Combining the above bounds results in
\[
 \left\| \frac{d^2}{d\tau_2^2} \psi_k(\tau_2)\right\|^2 = 
\| \left(-A+\Pi_k(B+C)\Pi_k\right)^2 \psi_k(\tau_2)\|^2  \le \eta  \la \psi_k(\tau)|A^6|\psi_k(\tau)\ra
\]
for some real number $\eta<\infty$ independent of $k$ and $\delta$ (although $\eta$ depends on $N$ and other parameters of the original ODE). 
Lemma~\ref{lemma:moment_bounds} provides a uniform (independent of $k$ and $\delta$) upper bound on $\la \psi_k(\tau)|A^6|\psi_k(\tau)\ra$
since $\la \psi(0)|A^6|\psi(0)\ra<\infty$. Substituting this into Eq.~(\ref{time_step_delta_eq2}) and taking the double integeral one gets
\be
\label{time_step_delta_eq3}
\| \psi_k(t+\delta) - \psi_k(t) - \delta (-A+\Pi_k(B+C)\Pi_k)\psi_k(t)\|  \le O(\delta^2),
\ee
where the constant  hidden in $O(\delta^2)$ is independent of $k$. Suppose we can prove that 
\be
\label{time_step_delta_eq4}
\lim_{k\to \infty}  (-A+\Pi_k(B+C)\Pi_k)\psi_k(t) = (-A+B+C)\psi(t).
\ee
We have already established that $\lim_{k\to \infty} \psi_k(t)=\psi(t)$ and $\lim_{k\to \infty} \psi_k(t+\delta)=\psi(t+\delta)$.
Taking the limit $k\to \infty$ in Eq.~(\ref{time_step_delta_eq3}) and using Eq.~(\ref{time_step_delta_eq4}) proves
Eq.~(\ref{time_step_delta_eq1}), thereby proving the lemma. Hence it remains to prove Eq.~(\ref{time_step_delta_eq4}).
We divide the proof into two parts.
\begin{prop}
\label{prop:limit1}
\be
\label{time_step_delta_eq5}
 (-A+B+C)\psi(t) = \lim_{k\to \infty}  (-A+B+C) \psi_k(t).
\ee
\end{prop}
\begin{prop}
\label{prop:limit2}
\be
\label{time_step_delta_eq6}
\lim_{k\to \infty}  \| (B+C)\psi_k(t) -\Pi_k(B+C)\Pi_k \psi_k(t)\|=0.
\ee
\end{prop}
Replacing the righthand side of Eq.~(\ref{time_step_delta_eq4})
with $\lim_{k\to \infty}  (-A+B+C) \psi_k(t)$ using Proposition~\ref{prop:limit1},
cancelling the  terms $-A\psi_k(t)$ that appear in both sides, and using Proposition~\ref{prop:limit2} proves Eq.~(\ref{time_step_delta_eq4}).
\begin{proof}[\bf Proof of Proposition~\ref{prop:limit1}]
Applying the triangle inequality and Lemma~\ref{lemma:loose} one gets
\[
\|  (-A+B+C)\psi(t)  - (-A+B+C) \psi_k(t)\| \le \| A(\psi(t)-\psi_k(t))\| + \omega_0 \| A^{3/2} (\psi(t)-\psi_k(t)) \|,
\]
Both terms go to zero in the limit $k\to \infty$
due to Corollary~\ref{corol:limit}.
\end{proof}
\begin{proof}[\bf Proof of Proposition~\ref{prop:limit2}]
Let $\Pi_{>k}= I-\Pi_k$. We have
\[
B+C - \Pi_k (B+C)\Pi_k  = (B+C) (\Pi_k+\Pi_{>k}) - \Pi_k (B+C)\Pi_k = (B+C)\Pi_{>k}  + \Pi_{>k} (B+C) \Pi_{k}.
\]
Thus it suffices to prove that $\|(B+C)\Pi_{>k} \psi_k(t)\|$ and $\|  \Pi_{>k} (B+C) \Pi_{k}\psi_k(t)\|$ go to zero in the limit $k\to \infty$.
From Lemma~\ref{lemma:loose} one gets
\[
\| (B+C)\Pi_{>k} \psi_k(t)\| \le  \omega_0 \| A^{3/2} \Pi_{>k} \psi_k(t)\|.
\]
Since $\Pi_{>k}\le k^{-1} A \Pi_{>k}$, one gets
\[
\| (B+C)\Pi_{>k} \psi_k(t)\| \le   \frac{\omega_0}{k} \| A^{5/2}\psi_k(t)\| = \frac{\omega_0 \sqrt{f_5(t)}}{k},
\]
where $f_5(t)=\la \psi_k(t)|A^5|\psi_k(t)\ra$ has a uniform (independent of $k$)  upper bound of  Lemma~\ref{lemma:moment_bounds}.
Hence 
\[
\lim_{k\to \infty} \| (B+C)\Pi_{>k} \psi_k(t)\|=0.
\]
Similar  arguments give
\begin{align*}
\| \Pi_{>k} (B+C) \Pi_{k} \psi_k(t)\| &  \le k^{-1} \| A \Pi_{>k} (B+C)\Pi_k \psi_k(t)\| =  k^{-1} \|  \Pi_{>k}  A(B+C)\Pi_k \psi_k(t)\|  \\
& \le k^{-1} \| A(B+C)\Pi_k \psi_k(t)\| 
\le \omega_1 k^{-1}  \| A^{5/2}\Pi_k \psi_k(t)\| \le  \frac{\omega_1 \sqrt{f_5(t)}}{k}
\end{align*}
Here the third   inequality uses Lemma~\ref{lemma:loose} with $p=1$. Hence
\[
\lim_{k\to \infty} \| \Pi_{>k} (B+C) \Pi_{k} \psi_k(t)\|=0.
\]
This proves Eq.~(\ref{time_step_delta_eq6}).
\end{proof}
\end{proof}
The statement of  Theorem~\ref{thm:regul} is equivalent to 
Corollary~\ref{corol:limit} with $p=3$, Lemma~\ref{lemma:KEsolution}, and Eq.~(\ref{regularization_error}).
\end{proof}

\section{Readout state}
\label{sec:readout}

Let $|\psi(t)\ra\in \calF_0$ be the solution of the second-quantized Kolmogorov equation  constructed in Section~\ref{sec:Kolmogorov}. 
By construction, the original Kolmogorov equation  Eq.~(\ref{Kolmogorov1}) has a solution
\be
\label{u_via_Phi}
u(t,x) 
= \frac1{\sqrt{\mu(x)}} \sum_{\mm \in \calJ} \la \mm|\psi(t)\ra \Phi_\mm(x).
\ee
Here $\mu(x)$ is the normalized Gaussian distribution defined in Eq.~(\ref{mu(x)}) and
$\Phi_\mm(x)$ are normalized eigenfunctions for a system of $N$ quantum harmonic oscillators such that  the $i$-th oscillator has mass,
$q^{-1}$, spring constant $q^{-1}\lambda_i^2$, and frequency $\lambda_i$. 
A well-known expression for these eigenfunctions is 
\[
 \Phi_\mm(x)=\sqrt{\mu(x)} \HH_\mm(x),
\]
where 
$\HH_\mm(x)$ are (rescaled) $N$-variate Hermite polynomials defined as follows.
Suppose first $x\in \RR$ is a single variable.
We use probabilist's Hermite polynomials defined as
\[
\herm_n(x) = (-1)^n e^{x^2/2} \frac{d^n}{dx^n} e^{-x^2/2},
\]
where $n\ge 0$ is the degree of $\herm_n(x)$. 
For example, $\herm_0(x)=1$, $\herm_1(x)=x$, $\herm_2(x)=x^2-1$, and
$\herm_3(x)=x^3-3x$. Define
\be
\label{Hermite_normalized}
\mathbb{H}_{\multi{m}}(x) = \prod_{i=1}^N  \frac1{\sqrt{(m_i)!}} \herm_{m_i} \left(x_i\sqrt{2\lambda_i/q}\right),
\ee
where $x\in \RR^N$. The rescaling of variables in $\HH_\mm(x)$ ensures that the eigenfunctions $\Phi_\mm(x)$ form an orthonormal family. Equivalently,
\be
\label{orthogonality_relations}
\int_{\RR^N} dx\, \mu(x) \HH_\nn(x)  \HH_\mm(x) = \delta_{\nn,\mm}.
\ee
 From Eq.~(\ref{u_via_Phi}) one gets
\be
\label{u_via_H}
u(t,x) =  \sum_{\mm \in \calJ} \la \mm|\psi(t)\ra \HH_\mm(x)
\ee

\subsection{Noiseless initial conditions}
\label{sec:readout_challenges}

Suppose first that the SDE Eq.~(\ref{SDE}) has noiseless initial conditions, $X(0)=x$.
 For the sake of argument, suppose a quantum simulator can efficiently prepare the exact
solution $|\psi(t)\ra$ of the Kolmogorov equation (up to the normalization). Can one efficiently extract $u(t,x)$ from $|\psi(t)\ra$ ? 
From Eq.~(\ref{u_via_H}) one gets
\be
\label{readout_failed}
u(t,x) = \la \varphi_{out}(x)|\psi(t)\ra,
\ee
where $|\varphi_{out}(x)\ra$ is a "state" with amplitudes $\la \mm|\varphi_{out}(x)\ra = \HH_\mm(x)$. One may hope to construct a quantum
circuit that prepares a normalized version of $|\varphi_{out}(x)\ra$ and estimates the inner product in Eq.~(\ref{readout_failed}) using the Hadamard test. 
Unfortunately, this approach fails since the "state" $|\varphi_{out}(x)\ra$  has an infinite norm. 
For example, suppose $x$ is the all-zero vector, $x=0^N$.
By definition, $\HH_\mm(0^N)=\prod_{i=1}^N (1/\sqrt{m_i!}) \herm_{m_i}(0)$. 
Using the identity $\herm_n(0)=(-1)^{n/2} (n-1)!!$ for even $n$ and $\herm_n(0)=0$ for odd $n$, see Ref.~\cite{andrews1999special}, p.~282, one gets
\[
\frac1{n!} \herm_n(0)^2  = \frac{(n-1)!!}{n!!}  \sim 1/\sqrt{n}.
\]
Thus 
\[
\|\varphi_{out}(0^N)\|^2 =  \sum_{\mm\in \ZZ_{\ge 0}^N} \;  \HH_\mm(0^N)^2\sim\left( \sum_{n=0}^\infty 1/\sqrt{n}\right)^N=\infty.
\]
Using  Mehler's formula one can check that the sum $\sum_{\mm\in \ZZ_{\ge 0}^N} \HH_\mm(x)^2$ diverges for any $x\in \RR^N$.

\subsection{Noisy initial conditions}
\label{sec:readout_solution}

Recall that our quantum algorithm aims to approximate the expected value $v(t,x)=\EE_z u(t,x+z)$, where $z\in \RR^N$ is a random  vector
that models noise in the initial conditions. 
Each variable $z_i$ is drawn from $\calN(0,q/(2\lambda_i))$. Equivalently, 
$z$ is drawn from $\mu(z)$.
We claim that 
for any $x\in \RR^N$ there there exists a state 
$|\psi_{out}(x)\ra\in \calF$ such that 
\be
\label{readout_eq1}
v(t,x) = \la \psi_{out}(x)|\psi(t)\ra
\ee
and
\be\begin{aligned}
\label{readout_state_norm}
    \|\psi_{out}(x)\|&=  \exp{\left[ q^{-1} \|x\|_\lambda^2\right]}.
\end{aligned}\ee
We shall refer to $|\psi_{out}(x)\ra$ as the {\em readout state}.
Indeed, from Eq.~(\ref{u_via_H}) one gets
\[
v(t,x)=
\int_{\RR^N}dy\,  \mu(z) u(t,x+z)
=\sum_{\mm\in \calJ} \la \mm|\psi(t)\ra \int_{\RR^N}dz\,  \mu(z)  \HH_\mm(x+z).
\]
The well-known  umbral identity for  single-variate Hermite polynomials reads as 
\[
\herm_n(x+y) = \sum_{k=0}^n {n \choose k} x^{n-k} \herm_k(y)
\]
for any $x,y\in \RR$. Using this identity and the orthogonality relations Eq.~(\ref{orthogonality_relations})
 gives
\[
\int_\RR dy \, e^{-y^2/2} \herm_n(x+y)  = \int_\RR dy\,  e^{-y^2/2} \herm_n(x+y)   \herm_0(y) = (2\pi)^{1/2} x^n
\]
for any $x\in \RR$.
As a consequence,
\[
\int_{\RR^N} dz \, \mu(z) \HH_\mm(x+z) = \prod_{i=1}^N \frac1{\sqrt{(m_i !)}} \left( x_i \sqrt{\frac{2 \lambda_i}{q}}\right)^{m_i}
\]
for any $x\in \RR^N$. We arrive at
\[
v(t,x) = \sum_{\mm\in \calJ} \psi_\mm(t)  \prod_{i=1}^N \frac1{\sqrt{(m_i !)}} \left( x_i \sqrt{\frac{2 \lambda_i}{q}}\right)^{m_i} = \la \psi_{out}(x)|\psi(t)\ra,
\]
where
\be
\label{psi_out}
|\psi_{out}(x)\ra = \sum_{\mm \in \ZZ_{\ge 0}^N}\; 
|\mm\ra \cdot  \prod_{i=1}^N \frac1{\sqrt{(m_i !)}} \left( x_i \sqrt{\frac{2 \lambda_i}{q}}\right)^{m_i} \in \calF.
\ee
We note that $|\psi_{out}(x)\ra$ coincides with the coherent state embedding of $x$
discussed in~\cite{engel2021linear}. 
One can easily check Eq.~(\ref{readout_state_norm}) by noting that
$|\psi_{out}(x)\ra$ is a tensor product of coherent states associated with each variable $x_i$, the norm is multiplicative under the tensor product,
and $\|\psi_{out}(x)\|^2= e^{2q^{-1} \lambda_1 x_1^2}$
in the case $N=1$.

In the next section we approximate $v(t,x)$ by the inner product
 $\la \psi_{out}(x)|\psi_k(t)\ra$, where
  $\psi_k(t)$ solves the  regularized Kolmogorov equation
  with a suitable cutoff $k$. 
 We shall see that normalized versions of the states $|\psi_{out}(x)\ra$ and $|\psi_k(t)\ra$ can be
efficiently approximated by quantum circuits. 

As a side remark, we note that the expected value $u(t,x)$ with a noise-free initial condition also admits a representation in terms of the readout state $\psi_{out}(x)$ defined in Eq.~(\ref{psi_out}).
Using the identity\footnote{This identity follows from
the well-known formula $\herm_n(x) = \EE_z (x+iz)^n$, where  $z\sim \calN(0,1)$, see for instance Ref.~\cite{andrews1999special}, p.~280.}
\[
\HH_\mm(x) = \int_{\RR^N} dz\, \mu(z)   \prod_{j=1}^N  \frac1{\sqrt{(m_j)!}} \left((x_j+ i z_j)\sqrt{2\lambda_j/q}\right)^{m_j}
\]
with $i=\sqrt{-1}$ one gets
\be
\label{u_coherent_state_rep}
u(t,x) =\int_{\RR^N} dz\, \mu(z) \la \psi_{out}(x+iz)|\psi(t)\ra.
\ee
Unfortunately, we were not able to convert Eq.~(\ref{u_coherent_state_rep}) into an efficient readout procedure capable of approximating $u(t,x)$. 

\section{Quantum algorithm}
\label{sec:algorithm}

Our  algorithm simulates the regularized Kolmogorov equation (RKE)  defined in Section~\ref{sec:regul-sz}.
Using the readout state introduced in Section~\ref{sec:readout} one can express the desired expected value $v(t,x)$ as 
\be
\label{v_quantum_classical_parts}
v(t,x)=\la \psi_{out}(x)|\psi(t)\ra \approx \la \psi_{out}(x)|\psi_k(t)\ra = \la  \psi_{out}(x)|\Pi_k|\psi_k(t)\ra +  \la  \psi_{out}(x)|I-\Pi_k|\psi_k(t)\ra.
\ee
The approximation in the second equality is controlled by the regularization error of Theorem~\ref{thm:regul} and the norm
of the readout state $\psi_{out}(x)$, as discussed in details below.  By definition of the RKE, the state
$(I-\Pi_k) \psi_k(t)$ evolves only under the diagonal Hamiltonian $-A$. Thus 
\be
\label{v_from_Hk_perp}
 \la  \psi_{out}(x)|I-\Pi_k|\psi_k(t)\ra = \sum_{\mm \in \calJ\setminus \calJ_k} e^{-\lambda_\mm t} \la \psi_{out}(x)|\mm\ra\la \mm|\psi(0)\ra,
\ee
where $\lambda_\mm = \la \mm|A|\mm\ra=\sum_{i=1}^N \lambda_i m_i$. 
As discussed below, our assumptions on the observable function $u_0$ imply that $\psi(0)$ is a sparse vector
with $O(1)$ nonzero amplitudes $\la \mm|\psi(0)\ra$ and each amplitude admits a simple analytic formula.  Thus  the second term in Eq.~(\ref{v_quantum_classical_parts})
can be computed classically in time negligible compared with the overall runtime of our algorithm.
To simplify the presentation, below we ignore the second term in Eq.~(\ref{v_quantum_classical_parts}).
Equivalently, we assume that $\psi(0)\in \calH_k$ (if this is not the case, one just need to replace $\psi(0)$ by $\Pi_k\psi(0)$ in all equations).
 Since RKE preserves $\calH_k$, one gets
$\psi_k(t)\in \calH_k$ for all $t\ge 0$. 
Thus it suffices to simulate the Kolmogorov equation projected onto $\calH_k$.

Let $Q$  be the number of qubits one needs to encode the  subspace $\calH_k$.
As discussed in Section~\ref{sec:regul-sz},  $Q\le O(k\lambda_1^{-1} \log{N})$ in the worst case.
In Appendix~\ref{app:encoding} we use techniques from~\cite{engel2021linear,tanaka2023polynomial}
to encode multi-indices $\mm \in \calJ_k$ by $Q$-bit strings.
Details of this encoding are not essential for understanding of our quantum algorithm.

Below  we describe main steps of the quantum algorithm claimed in Theorem~\ref{thm:algorithm}:
initialization,  simulating  time evolution associated with the projected Kolmogorov equation, and preparation of the readout state.
We then combine all these steps to estimate the desired expected value $v(t,x)$ and derive an upper bound on the qubit and gate count. 
For convenience of the reader, we collect all parameters affecting the runtime in the following table.
\begin{center}
\begin{tabular}{|c|c|}
\hline
$N$ & number of variables \\
\hline
$q$ & noise rate \\
\hline
$s$ & sparsity  \\
\hline
$J$ & drift strength \\
\hline
$k$ & regularization cutoff \\
\hline
$\ell$ & number of Trotter steps\\
\hline
$\lambda_1$ & smallest dissipation rate\\
\hline
$\epsilon$ & approximation error \\
\hline
\end{tabular}
\end{center}

\subsection{Initialization}
\label{sec:init}

Our simulation of the projected Kolmogorov equation begins with preparing the initial state  encoding the observable function $u_0\, : \, \RR^n\to \RR$.
By assumption, 
\be
\label{u0_restated}
u_0(x)=\prod_{i=1}^N x_i^{d_i},
\ee
where $d_i\ge 0$ and $\sum_{i=1}^N d_i=O(1)$.
Let 
\[
\overline{u}_0 = \int_{\RR^N} dx \, \mu(x) u_0(x)
\]
be the mean value of $u_0(x)$. From Wick's theorem one gets $\overline{u}_0=0$ unless $d_i$ is even for all $i$, in which case
\[
\overline{u}_0 =  \prod_{i=1}^N \frac{(d_i)!}{2^{d_i/2} (d_i/2) !} (q/2\lambda_i)^{d_i/2}.
\] 
The initial (unnormalized) state of the  Kolmogorov equation  is 
\[
|\psi(0)\ra =  \sum_{\mm \in \calJ_k}\phi_\mm  |\mm\ra,
\]
where $\phi_\mm\in \RR$ are coefficients of the expansion 
\[
u_0(x)-\overline{u}_0 = \sum_{\mm \in \calJ}  \phi_\mm \HH_\mm(x).
\]
Let us consider two examples. 

\noindent
{\em Example 1:}
Suppose our goal is to estimate the expected value of  a single variable $X_i(t)$.
Then $u_0(x)=x_i$ and $\overline{u}_0=0$.
The expansion of $u_0(x)$ in the Hermite basis  is
$u_0(x)= \sqrt{q/2\lambda_i} \HH_\mm(x)$, where $\mm=e^i$ is the multi-index with $m_i=1$ and $m_j=0$ for $j\ne i$.
The  initial state is 
\[
|\psi(0)\ra = \sqrt{\frac{q}{2\lambda_i}} |\mm\ra.
\]

\noindent
{\em Example 2:}
Suppose our goal is to estimate the expected value of $X_i^2(t)$.
Then $u_0(x) = x_i^2$ and $\overline{u}_0 = q/2\lambda_i$.  Let $\mm=2e^i$ be a multi-index such that 
$m_i=2$ and $m_j=0$ for $j\ne i$.
To expand $u_0(x)-\overline{u}_0$ in the Hermite basis note that $\herm_2(x)=x^2-1$ which gives
\[
\HH_{\mm}(x) = \frac1{\sqrt{2}} \herm_2( x_i \sqrt{2\lambda_i/q}) = \sqrt{2}\lambda_i q^{-1} x_i^2 - 1/\sqrt{2}.
\]
Hence $u_0(x) - \overline{u}_0 =  (q/\sqrt{2} \lambda_i)\HH_{\mm}(x)$ and the initial state is
\[
|\psi(0)\ra =  \frac{q}{\sqrt{2} \lambda_i} |\mm\ra.
\]

In the general case, expand  $u_0(x)$ in the Hermite basis by applying the identity
\be
\label{x_i^n}
x_i^n = n! \sum_{0\le m\le n/2} \frac1{2^m m! (n-2m)!} \herm_{n-2m}(x_i)
\ee
to each variable of of $u_0(x)$ and replacing $x_i$ by $x_i \sqrt{2\lambda_i/q}$ to get the rescaled polynomials $\HH_\mm(x)$. 
This shows that  $|\psi(0)\ra$ is a superposition of $O(d^d)=O(1)$ basis vectors $|\mm\ra$, $\mm\in \calJ_k$.
Coefficients in this superposition can be computed in time $O(Q) + poly(\log{N})$.
Indeed, we assumed that
the dissipation rates $\lambda_i$ are computable in time $poly(\log{N})$ and it takes time linear in $Q$ simply to write down a single
multi-index $\mm\in \calJ_k$.
The norm of  $\psi(0)$ can be computed using Wick's theorem:
\be
\label{psi0_norm_analytic}
\|\psi(0)\|^2  = \int_{\RR^N} dx\, \mu(x) (u_0 - \overline{u}_0)^2(x)= \prod_{i=1}^N \frac{(2d_i)!}{2^{d_i} d_i !} (q/2\lambda_i)^{d_i} -\overline{u}_0^2.
\ee
It is known~\cite{mao2024toward} that any $Q$-qubit  superposition of $\chi$ basis states
can be prepared by a quantum circuit of size linear in $Q + Q \chi/\log{(Q)}$. In our case $\chi=O(1)$
and thus the normalized state $\psi(0)/\|\psi(0)\|$ can be prepared by a quantum circuit of size linear in $Q$.
From Eq.~(\ref{psi0_norm_analytic}) one gets 
\be
\label{psi0_norm_bound}
\| \psi(0)\| \le poly(q,\lambda_1^{-1}).
\ee
Recall that the upper bound on the regularization error of Theorem~\ref{thm:regul} depends on both $\la \psi(0)|\psi(0)\ra$ and $\la \psi(0)|A|\psi(0)\ra$.
We claim that 
\be
\label{Apsi0_bound}
\la \psi(0)|A|\psi(0)\ra \le poly(q,\lambda_1^{-1}).
\ee
Indeed, the above arguments show that 
\[
|\psi(0)\ra = \sum_{\mm \in \calS} \psi_\mm |\mm\ra \qquad \mbox{where} \quad |\psi_\mm| \le O(1) \prod_{i=1}^N (q/\lambda_i)^{m_i/2},
\]
for some subset of multi-indices $\calS\subseteq \calJ$ of size $|\calS|=O(1)$
such that $|\mm|=O(1)$  for all $\mm\in \calS$
Let $\lambda_\mm = \sum_{i=1}^N \lambda_i m_i$. For any $\mm \in \calS$ one has
$\lambda_\mm \le O(1) \max_j \lambda_j m_j$ since $\mm$ has $O(1)$ nonzeros. Thus
\[
\la \psi(0)|A|\psi(0)\ra =\sum_{\mm \in \calS} \psi_\mm^2 \lambda_\mm \le O(1)  \max_{\mm \in \calS}  \max_j \lambda_j m_j  \prod_{i=1}^N (q/\lambda_i)^{m_i}
\le poly(q,\lambda_1^{-1}).
\]

\subsection{Time evolution}
\label{sec:time_evolution}

Let $A_k$, $B_k$,  $C_k$ be the restrictions  of $A$, $B$,  $C$ onto the subspace $\calH_k$.
Then the solution of the projected  Kolmogorov equation  is
\be
\label{RKE_eq2}
 |\psi_k(t)\ra = e^{t(-A_k + B_k + C_k)} |\psi(0)\ra.
\ee
The following lemmas provide upper bounds on the norm and the sparsity of 
$B_k+C_k$.
This determines the cost of simulating the time evolution $e^{t (B_k+C_k)}$ by the sparse Hamiltonian simulation 
methods~\cite{berry2014exponential,low2017optimal,low2017quantum,low2019hamiltonian}.
Note that $e^{t (B_k+C_k)}$ is  unitary since $B$ and $C$ are anti-hermitian. 
\begin{lemma}
\label{lemma:BplusC}
The matrix of $B_k+C_k$ in the standard basis has at most 
$4sk\lambda_1^{-1}$ nonzeros in every row and every column. Furthermore,
\be
\label{BC_norm}
\|B _k+C_k\|\le \gamma k^{3/2},
\ee
where $\gamma$  is the coefficient defined in Section~\ref{sec:ABCstructural}.
\end{lemma}
\begin{proof}
Recall that $C=C^\uparrow-C^\downarrow$, where $C^\uparrow$ is a linear combination of operators proportional to
$c_{ijk} a_j^\dag a_k^\dag a_i$ and $C^\downarrow = (C^\uparrow)^\dag$,
see Eq.~(\ref{Csimplified}).
Consider a column of $C_k$ labeled by a multi-index  $\mm \in \calJ_k$.
Then $C^\uparrow|\mm\ra$ is a linear combination of basis vectors  $|\mm -e^i + e^j + e^k\ra$ such that $m_i\ge 1$
and $c_{ijk}\ne 0$. The number of such basis vectors is at most $s|\mm|\le sk\lambda_1^{-1}$. 
Likewise, $C^\downarrow |\mm\ra$ is a linear combination of basis vectors $|\mm +e^i - e^j - e^k\ra$ such that $m_j\ge 1$
and $m_k\ge 1$ and $c_{ijk}\ne 0$. The number of such basis vectors is also at most $s|\mm|\le sk\lambda_1^{-1}$.
Hence $C_k$ has sparsity at most $2sk\lambda_1^{-1}$. Similar arguments show
that $B_k$ has sparsity at most $2sk\lambda_1^{-1}$. This proves the first part of the lemma.
Consider  any normalized states $\psi,\phi\in \calH_k$ such that
$\|B_k+C_k\| = \la \psi|B_k+C_k|\phi\ra$.
 From Lemma~\ref{lemma:AC1} one gets
\[
\|B_k+C_k\|  = \la \phi |B_k+C_k|\psi\ra \le \gamma \sqrt{\la \psi|A_k|\psi\ra  \la \phi|A_k^2|\phi\ra} \le \gamma \|A_k\|^{3/2} \le  \gamma k^{3/2}
\]
since the restriction of $A$ onto $\calH_k$ has eigenvalues at most $k$. 
\end{proof}

For simplicity, we simulate the projected Kolmogorov equation  by  Trotterizing the time evolution.
This is enough to get
 a polynomial time algorithm claimed in Theorem~\ref{thm:algorithm}.
We expect that the polynomial degree can be optimized  using 
 more sophisticated quantum ODE solvers~\cite{krovi2023improved,an2023linear}.
Consider the first-order product formula 
\be
\label{product_formula}
\calS(t) = e^{-tA_k } e^{t(B_k+C_k)}
\ee
approximating $e^{t(-A_k + B_k + C_k)}$ for a short evolution time $t$. 
Divide the time interval $[0,t]$ into $\ell$ Trotter steps
of length $t/\ell$.
Using  well-known upper bounds on the Trotter error, see e.g. Lemma~1 of~\cite{childs2021theory}, 
one gets
\be
\label{trotter_error}
\| \calS(t/\ell)^\ell  \psi(0) -   \psi_k(t)  \| \le  O(t^2/\ell) (\| A_k\| + \|B_k+C_k\| )^2 e^{(t/\ell) (\| A_k\| + \|B_k+C_k\| )}  \|\psi(0)\|
\ee
for any evolution time $t\ge 0$ and any number of Trotter steps $\ell$. 
Choose the number of Trotter steps $\ell$ large enough that 
\[
(t/\ell) (\| A_k\| + \|B_k+C_k\|)=O(1).
\] 
By definition, $\|A_k\|\le k$.
Using the bounds of  Lemma~\ref{lemma:BplusC}  one gets
\be
\label{Trotter_error_final}
\|   \calS(t/\ell)^\ell  \psi(0) -  \psi(t)  \| \le O( t^2 k^3/\ell)\cdot  poly(q,s,J,\lambda_1^{-1}) \cdot  \|\psi(0)\|.
\ee

Let $\tau \equiv  t/\ell$. 
A single Trotter step $e^{\tau (B_k+C_k)}$
can be implemented using any available quantum algorithm for simulating unitary time evolution under a sparse Hamiltonian, see e.g.~\cite{berry2014exponential,low2017optimal,low2017quantum,low2019hamiltonian,haah2019product}.
These methods have cost scaling linearly with the evolution time, Hamiltonian norm, sparsity,
and logarithmically with the desired error tolerance. 

A single Trotter step $e^{-\tau A_k}$ can be implemented  by a quantum circuit with a post-selective measurement.
 To this end  introduce one ancillary qubit such that the full Hilbert space is $\calH_k \otimes \CC_2$.
Let $R$ be the controlled  rotation of the ancilla qubit defined as
\be
\label{postselection_gadget}
R|\mm\ra \otimes |0\ra = |\mm\ra \otimes (c_\mm  |0\ra + s_\mm |1\ra)
\ee
for all $\mm\in \calJ_k$, where $c_\mm = e^{-\tau \lambda_\mm}$ and $s_\mm = \sqrt{1-c_\mm^2}$.
The action of $R$ on states $|\mm\ra \otimes |1\ra$ is irrelevant. 
Consider an arbitrary state $|\psi\ra \in \calH_k$.
Applying $R$ to  $|\psi\ra \otimes |0\ra$, measuring the ancillary qubit, and post-selecting the outcome $0$
implements the desired time evolution $|\psi\ra \to e^{-\tau A_k}|\psi\ra$.
The probability of observing the outcome $0$ is $P_0 = \| e^{-\tau A_k} \psi\|^2/\|\psi\|^2$.
Let $|0_{anc}\ra$ be the all-zero state of all ancillary qubits (including those used to implement the controlled rotation $R$).
The above shows that  a unitary quantum circuit $U_\ell$ satisfying 
\be
\label{U_ell}
\| \calS(t/\ell)^\ell  - \la 0_{anc} |U_\ell |0_{anc}\ra \| \le \epsilon
\ee
has gate complexity polynomial in all parameters
$k,\ell, t, q,s,J, \lambda_1^{-1}, \log{N},\log{\epsilon^{-1}}$.

\subsection{Readout state preparation}
\label{sec:readout_state_prep}

Let $\psi_{out}(x)$ be the readout state 
defined in Eq.~(\ref{psi_out})
and $\mathrm{supp}(x)\subseteq \{1,2,\ldots,N\}$ be the set of indices $i$ such that $x_i\ne 0$. 
By assumption, $|\mathrm{supp}(x)|=O(1)$.
From Eq.~(\ref{psi_out}) one gets
\be
|\psi_{out}(x)\ra = \bigotimes_{i\in \mathrm{supp}(x)} \; |\varphi_i\ra  \bigotimes_{i\notin \mathrm{supp}(x)} \; |0\ra
\ee
with
\be
|\varphi_i\ra =  \sum_{m\ge 0}
 \frac1{\sqrt{(m!)}} \left( x_i\sqrt{\frac{2 \lambda_i}{q}}\right)^{m}  |m\ra.
\ee
Note that $|\varphi_i\ra$ is the coherent state of a single harmonic oscillator. 
One can efficiently prepare a good approximation of $\varphi_{i}$ by truncating the sum over $m$ at some finite order $K$.
Indeed, define a state
\[
|\tilde{\varphi}_{i}\ra =  \sum_{m= 0}^{K}
 \frac1{\sqrt{(m!)}} \left( x_i \sqrt{\frac{2 \lambda_i}{q}}\right)^{m}  |m\ra.
\]
Using the bound 
\[
\left| e^x - \sum_{m=0}^{K} \frac{x^m}{m!}\right|\le \epsilon \quad \mbox{for} \quad
K\ge \frac{4x}{\log{(2)}} + \frac{\log{(1/\epsilon)}}{\log{(2)}}
\]
which holds for any $x\ge0$ and $\epsilon>0$, see Lemma~14 in~\cite{bravyi2021complexity},
one can ensure that 
\[
\| \varphi_{i} -\tilde{\varphi}_{i}\| \le \delta
\]
by choosing the cutoff
\[
K= \max_i \frac{8 \lambda_i x_i^2}{q \log{(2)}} +   \frac{2\log{(1/\delta)}}{\log{(2)}}.
\]
Our approximation of the readout state is defined as
\[
|\phi_{out}(x)\ra =  \bigotimes_{i\in \mathrm{supp}(x)} \; |\tilde{\varphi}_i\ra  \bigotimes_{i\notin \mathrm{supp}(x)} \; |0\ra.
\]
From the triangle inequality one gets
\be
\|  \psi_{out}(x)  -  \phi_{out}(x)\|  \le O(\delta) \| \psi_{out}(x)\|
\ee
Setting  $\delta\sim  \epsilon/ \| \psi_{out}(x)\|$ gives 
\be
\|  \psi_{out}(x)  -  \phi_{out}(x)\|  \le \epsilon
\ee
provided that 
\[
K = O(q^{-1} \|x\|_\lambda^2)+O(\log{\epsilon^{-1}}).
\]
Let $r\ge 1$ be the smallest integer such that $K+1\le 2^r$.
 Then the state $\tilde{\varphi}_i$ can be embedded into an 
$r$-qubit register using the binary encoding of integers. 
General purpose quantum state preparation methods
enable the exact preparation of any 
 $r$-qubit state  by a quantum circuit composed of $O(2^r)$ CNOTs and single-qubit gates,
 see e.g.~\cite{grover2002creating,mottonen2004transformation}. 
Thus  the cost of preparing the normalized version of $\phi_{out}(x)$
 on a quantum computer scales linearly with $q^{-1}\|x\|_\lambda^2$ and $\log{\epsilon^{-1}}$.
 The state $\phi_{out}(x)$ is expressed using $r|\mathrm{supp}(x)| = O(r)$ qubits. 
For simplicity, let us assume that  $\phi_{out}(x)\in \calH_k$ (if this is not the case, project 
$\phi_{out}(x)$ onto $\calH_k$).
Changing the binary encoding of  multi-indices 
to the encoding by $Q$-bit strings described in Appendix~\ref{app:encoding}, one can consider the  readout state
$\phi_{out}(x)$ as a $Q$-qubit state.

\subsection{Putting everything together}
\label{sec:runtime}

Define our approximation of $v(t,x)=\la \psi_{out}(x)|\psi(t)\ra$ as 
\[
v_1(t,x) = \la \phi_{out}(x) \otimes 0_{anc} | U_\ell|\psi(0) \otimes 0_{anc}\ra,
\]
where $|0_{anc}\ra = |00\ldots0\ra$ is the state of all ancillary qubits,
 $U_\ell$ is a unitary quantum circuit satisfying Eq.~(\ref{U_ell}),
 and $|\phi_{out}(x)\ra$ is our approximation of the readout state.
We shall choose the regularization  cutoff $k$ and the number of Trotter steps $\ell$ such that 
\be
\label{error_v1}
|v(t,x)-v_1(t,x)|\le \frac{\epsilon}{2},
\ee
where $\epsilon$ is the desired error tolerance of Theorem~\ref{thm:algorithm}.
Define normalized states
\[
|\hat{\psi}(0)\ra = \frac1{\|\psi(0)\|} |\psi(0)\ra \quad \mbox{and} \quad |\hat{\phi}_{out}(x)\ra = \frac1{\|\phi_{out}(x)\|} |\phi_{out}(x)\ra.
\]
 The inner product 
\[
v_2(t,x) = \la \hat{\phi}_{out}(x) \otimes 0_{anc} | U_\ell|\hat{\psi}(0) \otimes 0_{anc}\ra
\]
can be estimated using the standard Hadamard test since we have already showed
how to prepare the states $|\hat{\psi}(0)\ra$ and $|\hat{\phi}_{out}(x)\ra$.  We have
\[
v_1(t,x) = \| \phi_{out}(x)\| \cdot \|\psi(0)\| \cdot v_2(t,x).
\]
As discussed above, the norms $\|\psi(0)\|$ and $\|\phi_{out}(x)\|$ admit a simple analytic formula.
Thus it suffices to estimate $v_2(t,x)$ with an error 
\[
\epsilon_{v_2} = \frac{\epsilon}{2 \|\psi(0)\| \cdot \|\phi_{out}(x)\|} \ge   \frac{\epsilon}{2 \|\psi(0)\| \cdot \|\psi_{out}(x)\|}
\ge \epsilon \cdot  poly(q^{-1} \lambda_1) \cdot \exp{\left[ -q^{-1} \| x\|_\lambda^2\right]}.
\]
Here we used Eqs.~(\ref{psi0_norm_bound},\ref{readout_state_norm}).
The number of circuit repetitions required for the Hadamard test scales as $O(1/\epsilon_{v_2}^2)$
and the output of the Hadamard test approximates $v(t,x)$ with the desired error $\epsilon$.

It remains to choose the regularization cutoff $k$ and the number of Trotter steps $\ell$ to satisfy Eq.~(\ref{error_v1}).
By the triangle inequality,
\[
|v(t,x)-v_1(t,x)|\le \epsilon_{reg} + \epsilon_{tro} + \epsilon_{sim} +  \epsilon_{out},
\]
where the four terms represent the regularization error,  Trotter error,  Hamiltonian simulation error,
and  the readout error:
\[
\epsilon_{reg} = \|\psi_{out}(x)\| \cdot \| \psi(t) - \psi_k(t)\|,
\]
\[
\epsilon_{tro} =  \|\psi_{out}(x)\| \cdot \| \psi_k(t) - \calS(t/\ell)^\ell \psi(0)\|,
\]
\[
\epsilon_{sim} = \|\psi_{out}(x)\| \cdot   \|  \calS(t/\ell)^\ell  \psi(0) - \la 0_{anc} | U_\ell |\psi(0)\otimes 0_{anc}\ra \|,
\]
\[
\epsilon_{out} = \| \psi(0)\|\cdot \| \psi_{out}(x)-\phi_{out}(x)\|.
\]
One can satisfy Eq.~(\ref{error_v1}) whenever all errors $\epsilon_{reg}, \epsilon_{tro},\epsilon_{sim},\epsilon_{out}$
are at most $\epsilon/8$.
Using Theorem~\ref{thm:regul}
and Eqs.~(\ref{psi0_norm_bound},\ref{Apsi0_bound})
one can make $\epsilon_{reg}\le \epsilon/8$ by choosing the regularization cutoff
\be
\label{choice_of_k}
k=\epsilon^{-2} poly(q,s,J,\lambda_1^{-1})  \exp{\left[ 2q^{-1} \| x\|_\lambda^2\right]})
\ee
Next   use Eq.~(\ref{Trotter_error_final}) to choose the number of Trotter steps $\ell$ such that 
$\epsilon_{tro}\le \epsilon/8$. Clearly, $\ell$ scales polynomially with $t,\epsilon^{-1},q,s,J,\lambda_1^{-1}$,
and  $\exp{\left[ q^{-1} \| x\|_\lambda^2\right]}$.
Setting $\epsilon_{sim}=\epsilon/8$ determines the gate complexity of the Hamiltonian simulation circuit $U_\ell$.
As commented in the end of Section~\ref{sec:time_evolution}, for the  chosen $\ell, k$, the gate
complexity of $U_\ell$ scales polynomially with $t,\epsilon^{-1},q,s,J,\lambda_1^{-1}$,
and  $\exp{\left[ q^{-1} \| x\|_\lambda^2\right]}$.
Finally, the gate complexity of preparing the readout state 
scales linearly with $q^{-1}\|x\|_\lambda^2$ and $\log{(1/\epsilon_{out})}$, see Section~\ref{sec:readout_state_prep}.
This complexity is linear in $q^{-1}\|x\|_\lambda^2$ and logarithmic in $1/\|\psi(0)\| = poly(q^{-1} \lambda_1)$.

To conclude, our quantum algorithm  has gate and qubit count scaling
polynomially with  $\log{N}$, $t$, $\epsilon^{-1}$, $q$,$s$, $J$,  $\lambda_1^{-1}$, 
and exponentially with $q^{-1} \| x\|_\lambda^2$, as claimed in Theorem~\ref{thm:algorithm}.

\section{Proof of BQP-completeness}
\label{sec:BQP}

In this section we prove Theorem~\ref{thm:BQP}.
The proof requires a minor adaptation of well-known circuit-to-Hamiltonian mappings~\cite{kitaev2002classical}.
Suppose all dissipation rates are the same,
$\lambda_1=\ldots=\lambda_N\equiv \lambda$,
and all drift functions are linear, that is, $c_{ijk}\equiv 0$.
The divergence-free conditions are then equivalent to skew-symmetry of the matrix $b$. 
One can easily check that
for a fixed realization of Wiener noise  the system
\[
dX_i(t) = -\lambda X_i(t) dt  +  b_i(X(t))dt  + \sqrt{q}\, dW_t,  \qquad t\ge 0
\]
has a solution
\[
X(t) = e^{-\lambda t} e^{t b} X(0) + \sqrt{q} \int_0^t e^{-\lambda(t-s)}  e^{(t-s)b } dW_s.
\]
For a fixed initial state $x=X(0)\in \RR^N$ let $|x\ra=\sum_{i=1}^N x_i |i\ra$. The expected value of the variable $X_1(t)$ over Wiener noise is 
\[
u(t,x) = \EE_W X_1(t) = e^{-t\lambda } \la 1| e^{t b} |x\ra.
\]
Here we used the assumption  that Wiener noise has zero mean and $u(t,x)$ is a linear function of $x$.
By the same reason,
\[
v(t,x) =\int_{\RR^N} dz\, \mu(z) u(t,x+z) = u(t,x).
\]
Thus the noise does not matter in the linear case.

Let $U=U_m \cdots U_2 U_1$ be a quantum circuit on $n$ qubits composed of $m=poly(n)$ gates
$U_1,\ldots,U_m$.
Each gate $U_j$ is a unitary operator with real matrix elements acting nontrivially  on at most $k$ qubits. 
We choose  the matrix $b=\{b_{jk}\}$ as a  skew-symmetric weighted version of the Feynman-Kitaev 
Hamiltonian~\cite{kitaev2002classical} associated with the circuit $U$, 
\be
\label{B_from_circuit}
b=\sum_{j=1}^{m} \alpha_{j-1} \left( |j \ra\la j -1 |\otimes U_j  -  |j-1\ra\la j |\otimes U_j^\dag \right),
\ee
where $\alpha_j$ are real weights  yet to be defined. The operator $b$ acts on a tensor product Hilbert space
$\CC^{m+1} \otimes (\CC^2)^{\otimes n}$, where $\CC^{m+1}$ and  $(\CC^2)^{\otimes n}$ represent the clock register 
and the $n$-qubit computational register. We label basis vectors of the clock register by $|0\ra,|1\ra,\ldots,|m\ra$.
Basis vectors of the full Hilbert space $\CC^{m+1} \otimes (\CC^2)^{\otimes n}$ are labeled by $|j,i\ra$ where $0\le j\le m$
and $i \in \{0,1\}^n$. Note that 
\[
b = W (\tilde{b} \otimes I)  W^\dag,
\]
where
\[
W = |0\ra\la 0|\otimes I +  \sum_{j=1}^m |j \ra\la j| \otimes U_j \cdots U_2 U_1
\]
is a unitary operator that applies the first $j$ gates of $U$ to the computational register
if the clock register state is in the state $|j\ra$
and
\be
\label{B_0_definition1}
\tilde{b} = \sum_{j=0}^{m-1} \alpha_{j} \left( |j+1\ra\la j |  -  |j\ra\la j+1 |\right)
\ee
describes a quantum walk on a 1D chain of length $m+1$. We shall  choose the coefficients $\alpha_j$ such that 
\be
\label{B_0_definition2}
e^{-\tilde{b}} |0\ra = |m\ra.
\ee
Then 
\[
e^{-b } |0, 0^n\ra = W (e^{-\tilde{b} } \otimes I) W^\dag |0,0^n\ra
=W (e^{-\tilde{b} }\otimes I) |0,0^n\ra= W |m,0^n\ra=|m\ra \otimes U|0^n\ra.
\]
Hence
\be
\label{U00}
\la 0^n |U|0^n\ra = \la m,0^n| e^{- b} |0,0^n\ra = e^{\lambda} \la 0,0^n| e^{-\lambda I + b}|m,0^n\ra.
\ee
This is equivalent to $\la 0^n |U|0^n\ra = e^{\lambda} v(t,x)$ with  $x=|m,0^n\ra$.

Suppose each gate of $U$ acts on at most $k$ qubits. 
Then the  matrix $b$ has sparsity at most $s=2^{1+k}$.
 Indeed, consider a row of $b$ associated with a basis vector $|j,i\ra$ for some
$j\in [0,m]$ and $i\in \{0,1\}^n$. From Eq.~(\ref{B_from_circuit}) one infers that $\la j,i|b|j',i'\ra\ne 0$ only if
$j'=j-1$ and $\la i|U_j|i'\ra \ne 0$
or $j'=j+1$ and $\la i|U_{j+1}^\dag|i'\ra \ne 0$. Since each individual gate has sparsity $2^k$,
one infers that for a fixed pair $(j,i)$ the number of pairs $(j',i')$ 
such that $\la j,i|b|j',i'\ra\ne 0$ is at most $s=2^{1+k}$.
Since $b$ is skew-symmetric, columns and rows have the same sparsity.

It remains to bound the drift strength parameter
\[
J = \max_{1\le i,j\le N} |b_{ij}| \le s \max_{0\le j\le m-1} |\alpha_j|.
\]
Here we used Eq.~(\ref{B_from_circuit}) and noted that matrix elements of any unitary gate
have  magnitude at most one. 
Recall that the coefficients $\alpha_j$ must solve Eqs.~(\ref{B_0_definition1},\ref{B_0_definition2}).
We claim that $e^{-\tilde{b}}|0\ra=|m\ra$ if we choose
\[
\alpha_j = \pi \sqrt{(m - j)(j+1)}, \qquad j=0,1,\ldots,m-1.
\]
Indeed, for this choice of $\alpha_j$ the operator $\tilde{b}=\sum_{j=0}^{m-1} \alpha_j \left( |j+1\ra\la j |  -  |j\ra\la j+1 |\right)$ is proportional to the angular momentum operator $L_y$
for a system with the angular momentum $m/2$. The identity $e^{-\tilde{b}}|0\ra=|m\ra$ follows from the fact that 
time evolution under $L_y$ over time $\pi/2$ flips the $z$-component of the angular momentum. 
We conclude that $J=O(2^{k+1} m)$.
Choose $\lambda=q=1/10$ and $t=1$. Then 
\[
|\la 0^n|U|0^n\ra - v(t,x)| = |\la 0^n|U|0^n\ra - e^{\lambda} v(t,x) + (e^\lambda -1 ) v(t,x)| =  (e^\lambda -1 )|v(t,x)| \le (e^\lambda -1 )e^{-\lambda} \le \frac1{10},
\]
as claimed in Theorem~\ref{thm:BQP}.

\section{Navier–Stokes equations}
\label{sec:examples}

In this section we describe an explicit example of
a quadratic divergence-free ODE:
2D \ac{nse} describing a flow of an incompressible fluid. We also benchmark our quantum algorithm by performing numerical experiments. 

The notation and discretization of \ac{nse} of this example closely follows~\cite{zhuk2023detectability,flandoli1998kolmogorov}. Here we omit all the formal definition of spaces required for applying Galerkin projection~\cite[p.55]{Foias_Manley_Rosa_Temam_2001}, which effectively approximates the solution of the continuous \ac{nse}, $\vec v$, by a linear combination of finitely many eigen-functions of the Stokes operator with time-dependent coefficients, then projects \ac{nse} onto a subspace generated by this eigen-functions to get a non-linear divergence-free ODE describing dynamics of the time-dependent coefficients. This ODE is then stochastically forced and we apply our algorithm to estimate expectations of certain functions $u_0$ of the solution.  

The classical \ac{nse} in 2D is a system of two PDEs defining dynamics of the scalar pressure field $p(x,y)$ and the viscous fluid velocity vector-field $\vec v(t,x,y)=[v_1(t,x,y),v_2(t,x,y)]$ which depends on the initial condition $\vec v(0)=\vec v_0$, forcing $\vec f$, and Boundary Conditions (BC), e.g. periodic BC $v_{1,2}(t,x+\ell_1,y) = v_{1,2}(t,x,y)$, $v_{1,2}(t,x,y+\ell_2) = v_{1,2}(t,x,y)$. In the vector form it reads as follows:
\begin{equation}
  \label{eq:NSE-vector}
  \begin{split}
    \dfrac{d\vec{v}}{dt} &+ (\vec{v} \cdot \nabla) \vec{v} - \nu\Delta\vec{v}+\nabla p = \vec f,\quad \nabla\cdot\vec{v} = 0
  \end{split}
\end{equation}
To eliminate pressure $p$ and obtain an evolution equation just for $\vec{v}$ it is common to use Leray projection~\cite[p.38]{Foias_Manley_Rosa_Temam_2001}: every vector-field $\vec{v}$ in $\RR^2$ admits Helmholtz-Leray decomposition, $\vec{v} = \nabla p + \vec q$ with $\nabla\cdot\vec q = 0$ which in turn defines Leray projector, $P_l(\vec{v})=\vec q$ (e.g.~\cite[p.36]{Foias_Manley_Rosa_Temam_2001}). 
Multiplying~\eqref{eq:NSE-vector} by a smooth divergence-free test function $\vec\phi$, and integrating by parts in $\Omega$ allows one to obtain Leray's weak formulation of NSE in 2D:
\begin{equation}
\label{eq:NSE-var}
    \dfrac{d}{dt}(\vec{v},\vec \phi) + b(\vec{v},\vec{v},\vec\phi) + \nu(A \vec v,\vec \phi) = (\vec f,\vec \phi), 
\end{equation}
with initial condition $(\vec{v}(0),\vec\phi)=(\vec{v}_0,\vec\phi)$. Here the Stokes operator~\cite[p.52]{Foias_Manley_Rosa_Temam_2001} $A$ and trilinear form $b$ are defined as follows: 
\[
    A\vec v = -\Delta \vec v = 
    \left[\begin{smallmatrix}-\Delta v_1\\ -\Delta v_2, 
    \end{smallmatrix}\right], 
    \quad
    b(\vec v,\vec y,\vec \phi) = 
    \int_\Omega (\vec v\cdot \nabla y_1) \phi_1 d\vec\xi + \int_\Omega (\vec v\cdot \nabla y_2) \phi_2 d\vec\xi
\]
The orthonormal set of eigenfunctions of $A$ is given by (see~\cite{flandoli1998kolmogorov})
\begin{align*}
    A\vec{e}_{\vec k} &= \lambda_{\vec k} \vec{e}_{\vec k}, & \lambda_{\vec k} &= 4\pi^2|\vec k|^2_2,&\vec k &= (k_1,k_2)^\top, k_{1}\in\mathbb{N}, k_{2}\in\mathbb{Z}\\
    \vec{e}_{\vec k}&= \frac{\vec{k}^\perp}{|\vec k|_2}\sqrt{2}\sin(2\pi\vec{k}\cdot\vec{\xi}),
    & \vec{k}^\perp&=(k_2,-k_1)^\top,&\vec{k}\cdot\vec{\xi}&=k_1\xi_1+k_2\xi_2, |\vec k|_2^2=\vec k\cdot\vec k
\end{align*}
Additionally, we enumerate 2-dimensional vectors (multi-indices) $\vec k$ with a scalar index $k\in\mathbb{N}$ permitting us to define $\lambda_{k}=\lambda_{\vec k}$ and $\vec{e}_{k}=\vec{e}_{\vec k}$. Furthermore, in what follows we will abuse the notation and write $\xi$ for $\vec \xi$ and $e_k$ for $\vec{e}_{k}$ for simplicity. We note that (see~\cite[p.99,p.101]{Foias_Manley_Rosa_Temam_2001} \[
b(\vec u ,\vec y,\vec \phi)=-b(\vec u,\vec \phi,\vec y), \quad b(\vec u,\vec u,A\vec u)=0. 
\]
Now we approximate $\vec v$ by $\vec v_N=\sum_{k=1}^N x_{k}(t) e_{k}(\xi)$, substitute $\vec v_N$ into~\eqref{eq:NSE-var} instead of $\vec v$, set $\vec\phi=e_{k}$ and obtain a system of finitely many ODEs: 
\begin{align}\label{eq:NSEODE}
    \dot x_k &= -\nu\lambda_k x_k - b(\vec u_N,\vec u_N,e_k) + (\vec f,e_k)\\
    &=-\nu\lambda_k x_k - \sum_{i,j=1}^N b(e_i,e_j,e_k)x_ix_j + (\vec f,e_k)\\
    &=-\nu\lambda_k x_k + b_k(x) +(\vec f,e_k)
\end{align}
with initial conditions $x_k^0=(\vec{v}_0,e_k)$. Since $b(e_k,e_k,\vec u)\equiv 0$ and $b(e_i,e_k,e_k)=0$ we get that $b_k(x)$ is independent of $x_k$. Also $\sum \lambda_k x_k b_k(x) = b(\vec u_N,\vec u_N, A \vec u_N)=0$. Hence, the disretized NSE given by~\eqref{eq:NSEODE} satisfies the divergence-free condition.  

Now, set $\vec f=0$ and consider the stochastic discretized \ac{nse} in which $w_k$ is a standard scalar Wiener process:
\begin{equation} \label{eq:stochastic_disc_nse}
    d X_k = -\nu\lambda_k X_k dt + b_k(X)dt + \sqrt{q} d W_k
\end{equation}
We are interested in computing the following expectation value: $u(t,x)=\EE_W\left(u_0(X(t)|X(0)=x)\right)$, for a given function $u_0:\RR^N\to\RR$. As suggested above this task can be performed by rewriting $u$ by means of Hermite polynomials $u(t,x) = \sum_{\mm\in \calJ} \psi_\mm(t) \HH_\mm(x)$ and introducing the Kolmogorov equation describing dynamics of coefficients $\psi(t) =\{\psi_\mm(t)\}_{\mm\in \calJ}$: 
\be
\label{eq:ke_assoc_nse}
\frac{d}{dt} |\psi(t)\ra = (-A + C)|\psi(t)\ra
\ee
where $A=\text{diag}(\{\lambda_{\mm}\}_{\mm\in\calJ})$ with elements $\lambda_\mm=\sum_{i=1}^N \mm_i 4\pi^2|\vec i|^2_2$, and $C=\{C_{\mm,\nn}\}_{\mm,\nn\in\calJ}$ is a skew-symmetric matrix with elements $C_{\mm,\nn} = \sum_{k=1}^N \int_{\RR^N}  (b_k(x) \partial_{x_k}\HH_\nn)\HH_\mm \mu(dx)$.
An explicit formula for $C_{\mm,\nn}$ is derived in Appendix~\ref{app:NSE_matrix_elements}.
We omit the demonstration of the existence of the solution of KE~\eqref{eq:ke_assoc_nse} (which can be done e.g. by employing~\cite[Lemma 4.1]{flandoli1998kolmogorov}).

\subsection{Numerical experiment}
\begin{figure}[htbp]
    \centering
    \begin{subfigure}[t]{0.32\textwidth}
        \includegraphics[width=\textwidth]{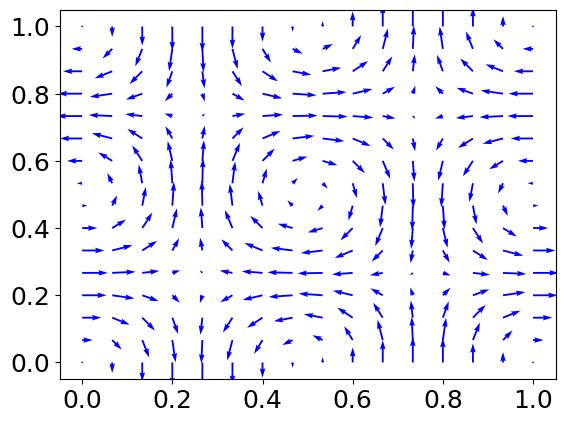}
        \caption{}
        \label{fig:tg_initial}
    \end{subfigure}
    \hfill
    \begin{subfigure}[t]{0.32\textwidth}
        \includegraphics[width=\textwidth]{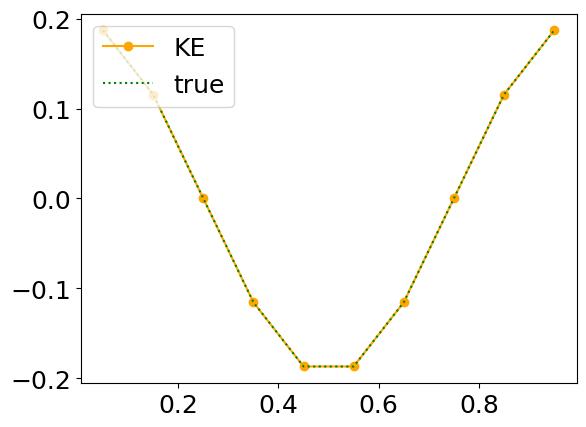}
        \caption{}
        \label{fig:ke_no_noise}
    \end{subfigure}
    \hfill
    \begin{subfigure}[t]{0.32\textwidth}
        \includegraphics[width=\textwidth]{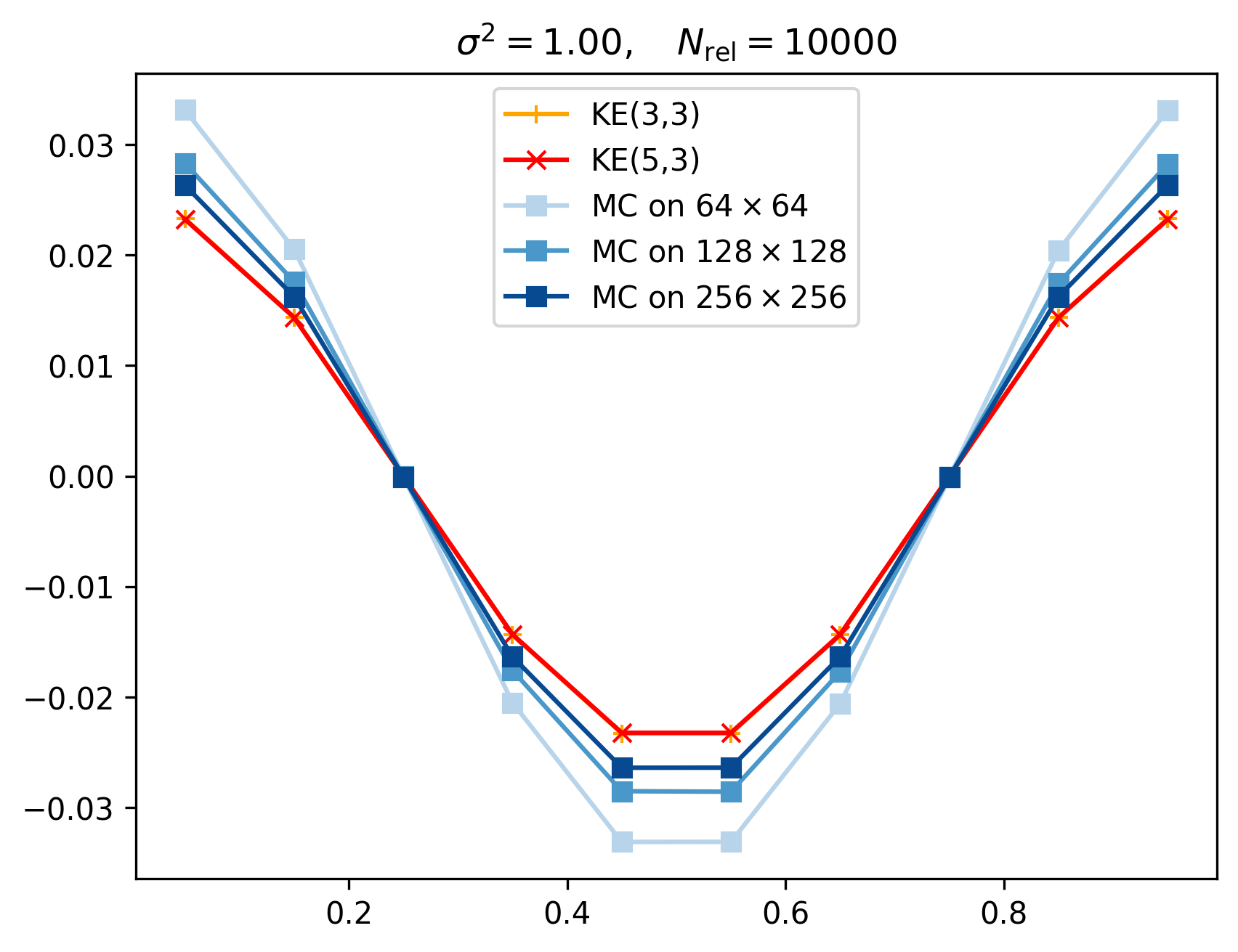}
        \caption{}
        \label{fig:ke_large_noise}
    \end{subfigure}
    \begin{subfigure}[t]{0.72\textwidth}
        \includegraphics[width=\textwidth]{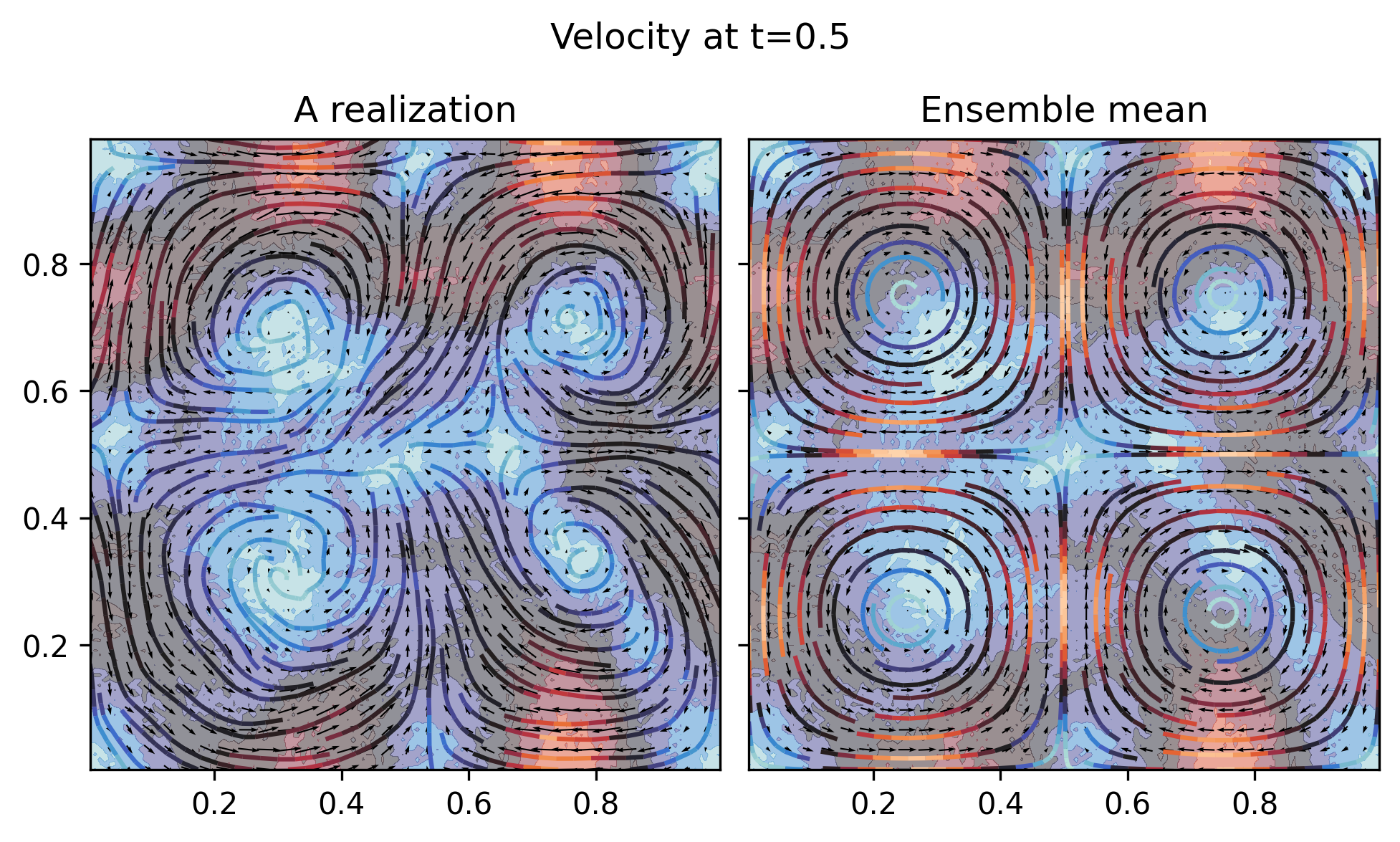}
        \caption{}
        \label{fig:ke_large_noise}
    \end{subfigure}
    \begin{subfigure}[t]{0.72\textwidth}
        \includegraphics[width=\textwidth]{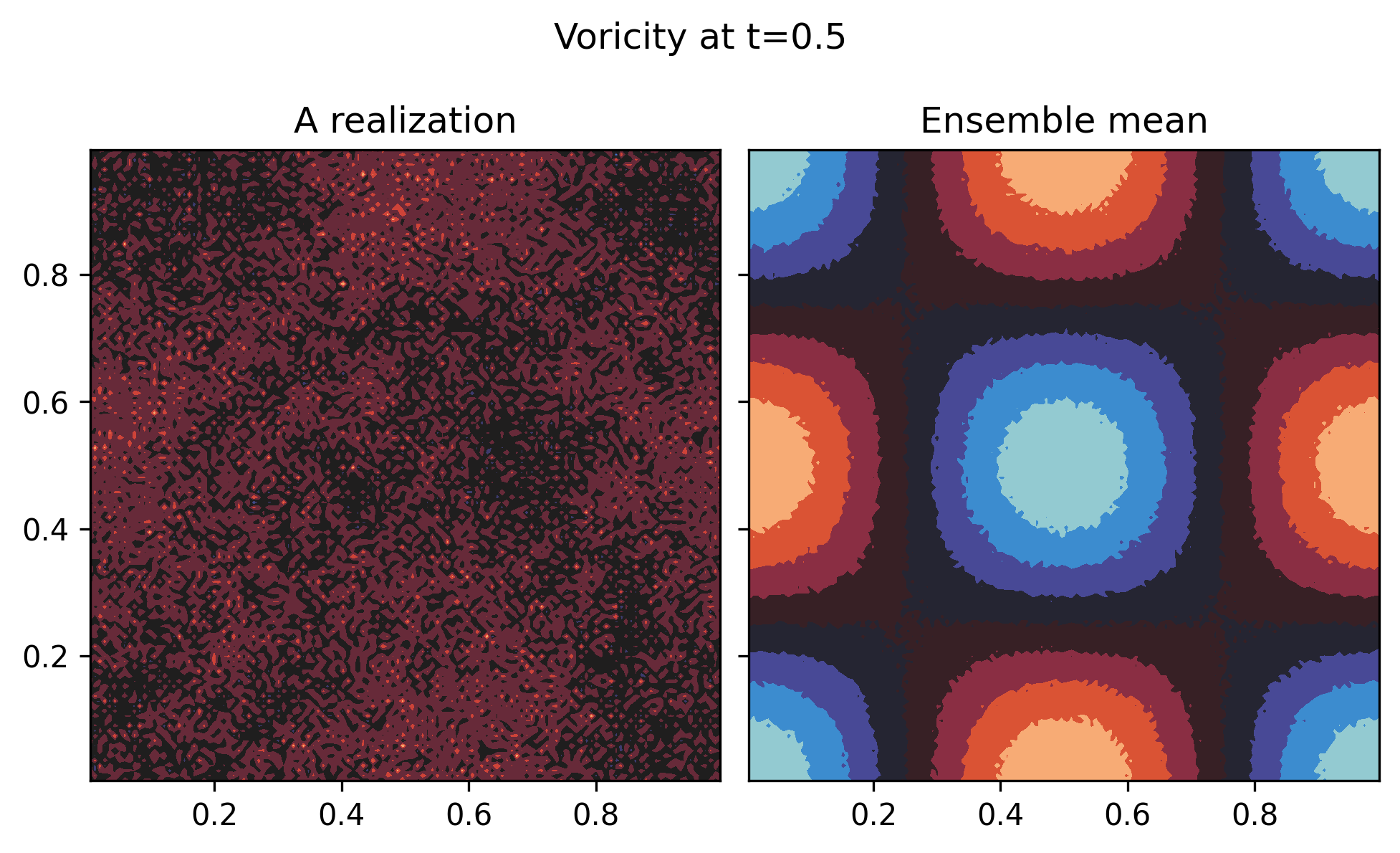}
        \caption{}
        \label{fig:ke_large_noise}
    \end{subfigure}
    \caption{Comparison of the solutions of the \ac{nse} obtained by our approach and compared against true analytical solution: a) initial state of the Taylor-Green vortex flow; b) the case with vanishing noise; c) the case with non-vanishing noise: ensemble of $N_{\mathrm{rel}}$ members at 3 spatial resolutions (shades of blue) is compared agaisnt our algorithm (KE) (orange, red) with different $N$s and 3rd order polynomials; d) example of a realization of the velocity field (left) and the ensemble mean velocity field (right); e) realization of the vorticity and ensemble mean vorticity field.}
\end{figure}

To validate the described quantum algorithm, we performed a numerical experiment that solves the \ac{nse} with a known analytical solution. To this end, we considered a 2-dimensional Taylor-Green vortex flow $\vec v_{true}(t,x,y)=[v_1^{true}(t,x,y), u_v^{true}(t,x,y)]$ defined as
\begin{equation}\begin{aligned} \label{eq:tg_vortex}
    &u_1^{true}(t,x,y) = \sqrt{2} \sin(2\pi x) \cos(2\pi y) \exp(-8\pi^2\nu t)\\
    &u_2^{true}(t,x,y) = -\sqrt{2} \cos(2\pi x) \sin(2\pi y) \exp(-8\pi^2\nu t)
\end{aligned}\end{equation}
for $x,y \in [0;1]$, $t\in[0;0.25]$ and $\nu=0.1$. The flow defined in \eqref{eq:tg_vortex} is divergence-free and satisfies the Navier--Stokes equation \eqref{eq:NSE-vector}. Its initial state, $\vec v(0)$, is depicted in Fig.~\ref{fig:tg_initial}.

To obtain a solution of the underlying \ac{nse}, we configure the corresponding KE to describe the dynamics of a functional $u(t,x)$ taken as
\begin{equation}
    u(t,x)= \EE_W\left(\sum_{k=1}^N X_{k}(t) E_1\cdot \vec e_{k}(\xi_1^p,\xi_2^p)\right)
\end{equation}
where $X(0)=x^0$, $x^0=\{x^0_k\}_{k=1}^N$ is a vector of projection coefficients of Taylor--Green vortex at the initial time in the span of eigenfunctions $e_k,\;k=[1,\dots,N]$, $E_1=(1,0)^\top$, $(\xi_1^p,\xi_2^p)$ is a point in space. In other words, functional $u(t,x)$ is an expected value approximating the first component of the velocity field $\vec u$, solving \ac{nse}, evaluated at the point $(\xi_1^p,\xi_2^p)$. Specifically, we take 10 different points with $\xi_2^p=0.25$ and $\xi_1^p$ uniformly distributed over an interval $[0.05; 0.95]$. We then solve the discretized KE using the matrix $C$ as defined in~\eqref{eq:c_mn_final} computed using $N=40$ eigenfunctions $e_k$ and taking Hermite polynomials of order up to $K=3$. We present the solutions computed at $t=0.25$ across all 10 spatial points in Fig.~\ref{fig:ke_no_noise} and~\ref{fig:ke_large_noise} for two cases with different noise levels.

In the first case, we utilise a vanishing noise with a variance $10^{-5}$ to mimic the dynamics of the deterministic \ac{nse}. As shown in Fig.~\ref{fig:ke_no_noise}, the obtained solution agrees well with the corresponding analytical solution obtained by applying functional $u(t,x)$ to the projections of Taylor--Green vortex~\eqref{eq:tg_vortex} in the span of eigenfunctions $e_k,\;k=[1,\dots,N]$.

In the second case, we introduce noise with a variance equal to $1$. Here, the dynamics of the mean solution of the stochastic \ac{nse} diverges from the deterministic solution, as demonstrated in Fig.~\ref{fig:ke_large_noise}. An example of the dynamics of stochastically forced NSE starting from Taylor-Green vortex is given in Figure~\ref{fig:NSEforward}.  

\begin{figure}[p]
  \centering
  \includegraphics[width=0.9\textwidth]{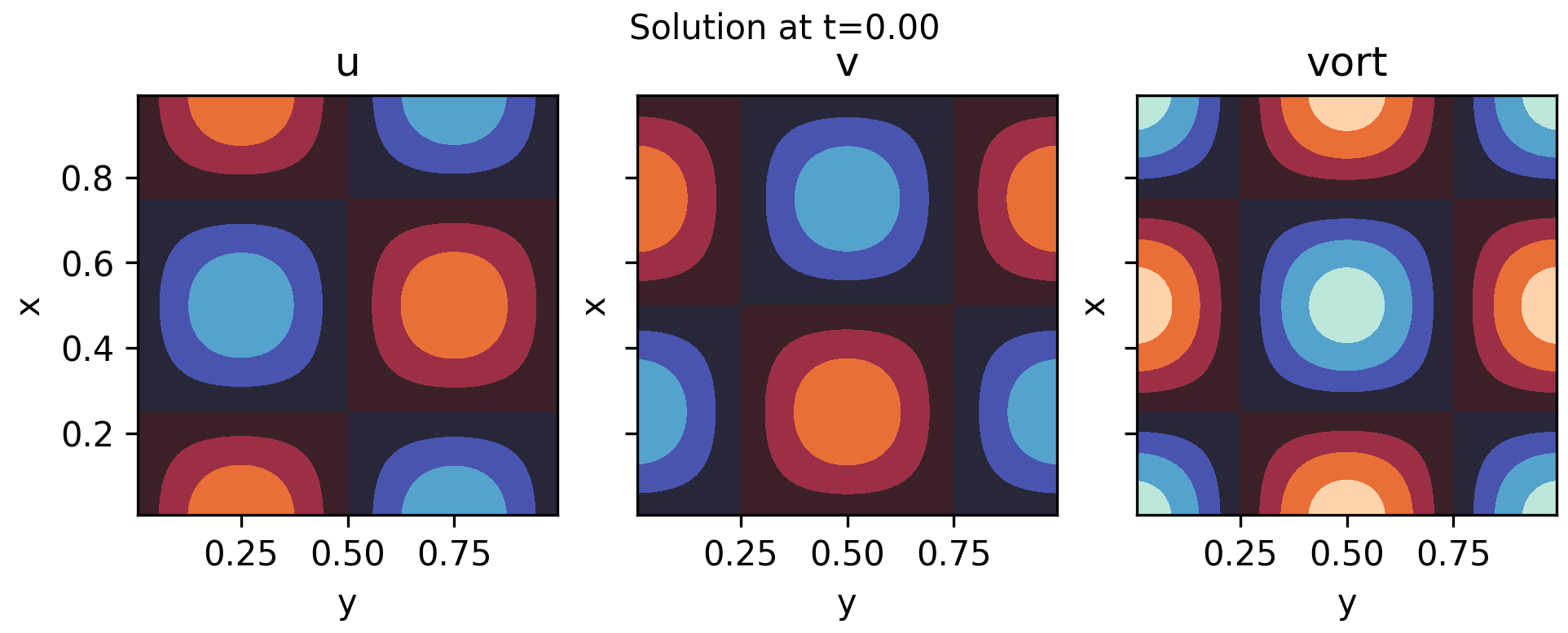}
  \includegraphics[width=0.9\textwidth]{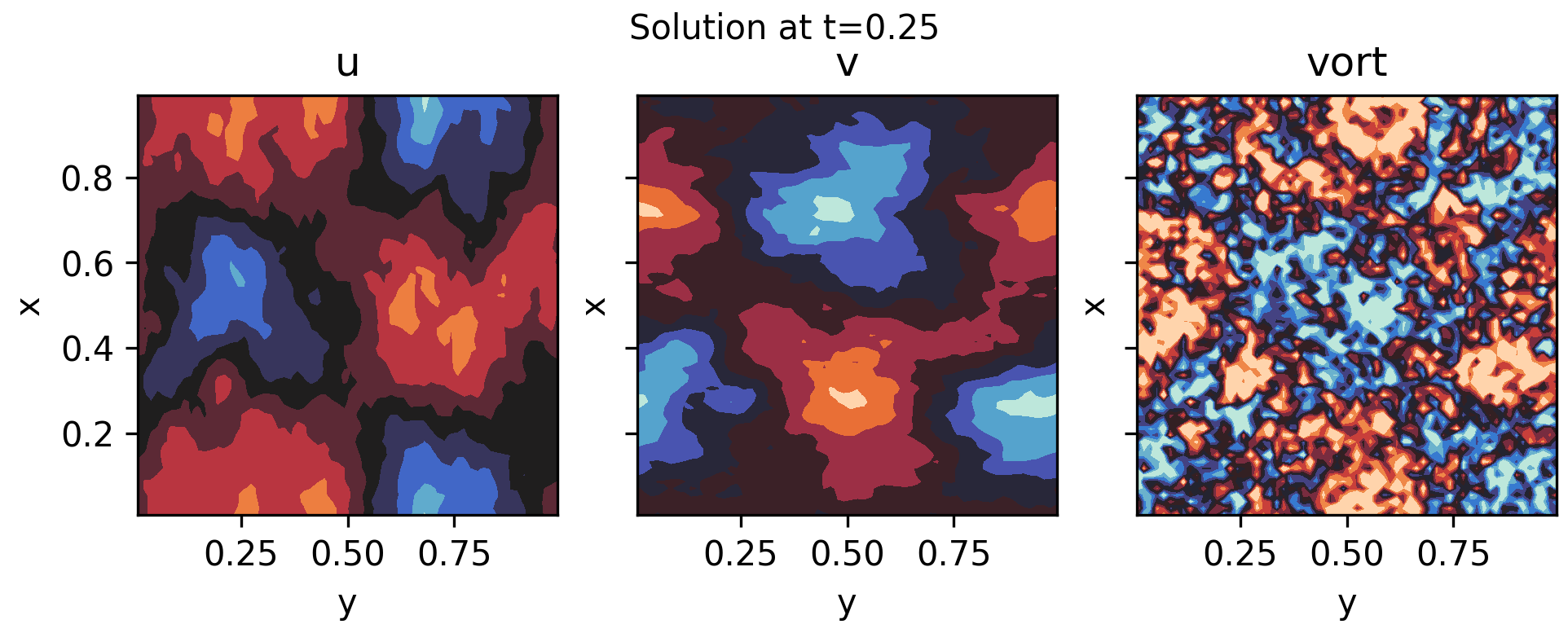}
  \includegraphics[width=0.9\textwidth]{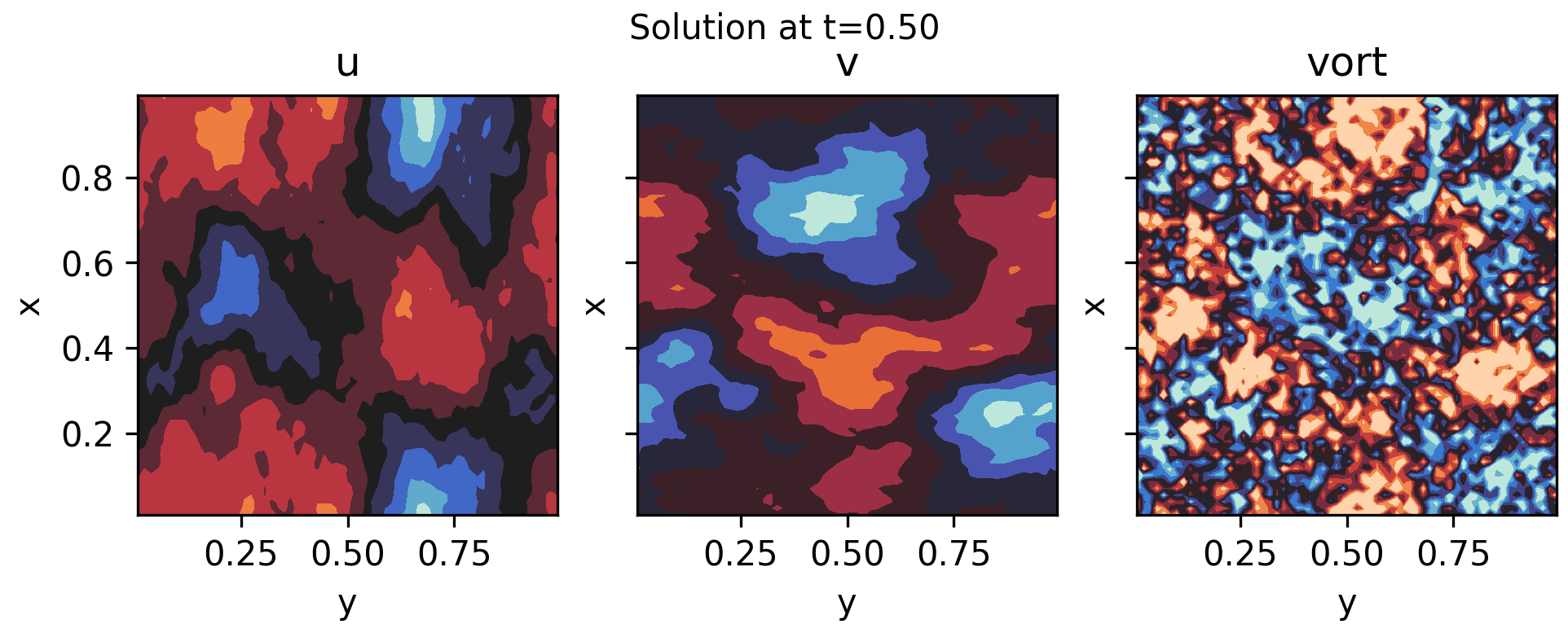}
  
  \hspace{1cm}\includegraphics[width=0.8\textwidth]{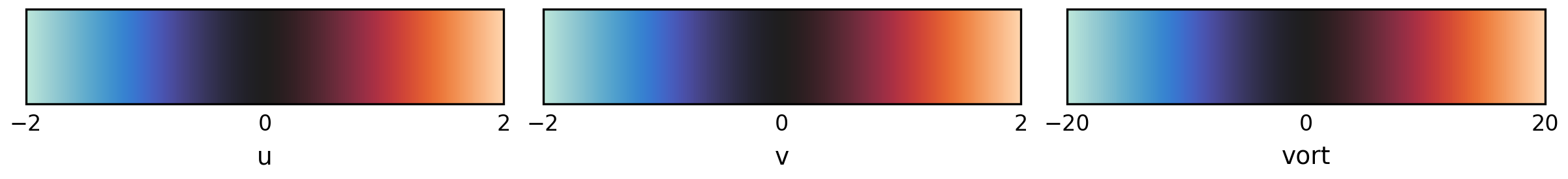}
  \caption{Solution of an incompressible 2D \acf{nse} with $\nu=0.01$ with initial condition described by equation~\ref{eq:tg_vortex} and perturbued with cylindrical Wiener noise with $\sigma^2=1$. The fields respectively represent $x$- and $y$-components of velocity and resulting vorticity.}\label{fig:NSEforward}
\end{figure}


\section*{Acknowledgments}
The authors thank Andrew Childs, Arkopal Dutt, Jay Gambetta, Hari Krovi, Kristan Temme, Martin Mevissen and Juan Bernabe Moreno for helpful discussions. 


\appendix

\section{Proof of Lemma~\ref{lemma:AC1}}
\label{app:lemmaAC1proof}

For convenience of the reader we restate the lemma.
\begin{lemma*}
For any vectors $\psi,\phi\in \calF_0$ one has 
\be
\label{C_overlap_upper_restated}
|\la \phi |B+C|\psi\ra |\le  \gamma \sqrt{\la \psi|A|\psi\ra \la \phi|A^2|\phi\ra}
\ee
where $\gamma = 2q^{1/2} sJ \lambda_1^{-2} (1+ \lambda_1^{1/2} q^{-1/2})$.
\end{lemma*}
\begin{proof}
Let us first upper bound $|\la \phi|C|\psi\ra|$. From Eq.~(\ref{Csimplified}) one gets
\[
C=C^\uparrow -C^\downarrow,
\]
where
\[
C^\uparrow = (q/2)^{1/2}\sum_{i,j,k=1}^N   (\lambda_i\lambda_j  \lambda_k)^{-1/2}   \lambda_i  c_{ijk} a_j^\dag a_k^\dag a_i 
\]
and $C^\downarrow=(C^\uparrow)^\dag$.
The bound $(\lambda_j  \lambda_k)^{-1/2} \le \lambda_1^{-1}$ and
 the triangle inequality give
\[
|\la \phi |C^\uparrow|\psi\ra |\le (q/2)^{1/2} \lambda_1^{-1} \sum_{\mm \in \calJ} \; |\psi_\mm|  \sum_{i,j,k=1}^N   (m_i \lambda_i)^{1/2} |c_{ijk}| \cdot |\la \phi| a_j^\dag a_k^\dag |\mm-e^i\ra|.
\]
Let us rewrite this as 
\[
|\la \phi |C^\uparrow|\psi\ra |\le (q/2)^{1/2} \lambda_1^{-1}  \sum_{\mm \in \calJ} \; \sum_{i,j,k=1}^N x_{ijk\mm} y_{ijk\mm},
\]
where
\[
x_{ijk\mm} = |\psi_\mm| (m_i\lambda_i)^{1/2}   |c_{ijk}|^{1/2} \quad \mbox{and} \quad y_{ijk\mm} =   |c_{ijk}|^{1/2} \cdot |\la \phi| a_j^\dag a_k^\dag |\mm-e^i\ra|.
\]
Let us agree that $y_{ijk\mm} =0$ if $m_i=0$.
Cauchy-Schwartz inequality gives
\be
\label{lemmaAC1_eq1}
|\la \phi |C^\uparrow|\psi\ra | \le (q/2)^{1/2} \lambda_1^{-1}  
\sqrt{  \sum_{\mm \in \calJ}\;  \sum_{i,j,k=1}^N  x_{ijk\mm}^2}
\cdot \sqrt{   \sum_{\mm \in \calJ}\;  \sum_{i,j,k=1}^N  y_{ijk\mm}^2}.
\ee
For any fixed index $i$ there are at most $s$ nonzero coefficients $c_{ijk}$. Hence
\be
\label{lemmaAC1_eq2}
\sum_{\mm \in \calJ}\;  \sum_{i,j,k=1}^N  x_{ijk\mm}^2 \le 
sJ \sum_{\mm \in \calJ} \; \sum_{i=1}^N  \; |\psi_\mm|^2 m_i \lambda_i
 =sJ \la \psi| A |\psi \ra.
\ee
From the definition of creation/annihilation operators one gets
\[
y_{ijk\mm}  = |c_{ijk}|^{1/2}  (m_j+1)^{1/2}(m_k+1)^{1/2}  |\phi_{\mm - e^i + e^j +e^k}|.
\]
Perform a change of variables $\nn= \mm-e^i+e^j+e^k$ in the sum over $\mm$.
Then $m_j+1=n_j$ and $m_k+1=n_k$. Thus
\be
\label{lemmaAC1_eq3}
\sum_{\mm \in \calJ}\;  \sum_{i,j,k=1}^N   y_{ijk\mm}^2
=  \sum_{i,j,k=1}^N\; \sum_{\mm \in \calJ}\;   y_{ijk\mm}^2
\le 
 \sum_{i,j,k=1}^N\; 
\sum_{\nn \in \calJ}\; \phi_\nn^2 |c_{ijk}| n_j n_k \le sJ \lambda_1^{-2} \la \phi|A^2|\phi\ra.
\ee
Substituting Eqs.~(\ref{lemmaAC1_eq2},\ref{lemmaAC1_eq3}) into Eq.~(\ref{lemmaAC1_eq1}) gives
\be
\label{lemmaAC1_eq4}
|\la \phi |C^\uparrow|\psi\ra | \le (q/2)^{1/2} sJ  \lambda_1^{-2}  \sqrt{\la \psi|A|\psi\ra} \sqrt{\la \phi|A^2|\phi\ra}.
\ee
We shall bound  $|\la \phi |C^\downarrow |\psi\ra|$ by rearranging the terms in $C^\downarrow$ and applying the same arguments as above.
By definition,
\[
C^\downarrow = (q/2)^{1/2}\sum_{i,j,k=1}^N   (\lambda_i\lambda_j  \lambda_k)^{-1/2}   \lambda_i  c_{ijk} a_i^\dag a_k a_j.
\]
From the triangle inequality one gets
\[
|\la \phi |C^\downarrow|\psi\ra |\le (q/2)^{1/2} \lambda_1^{-1} \sum_{\mm \in \calJ} |\psi_\mm|  \sum_{i,j,k=1}^N   m_j^{1/2} \lambda_i^{1/2} |c_{ijk}|\cdot |\la \phi| a_i^\dag a_k|\mm-e^j\ra|.
\]
Let us rewrite this as
\[
|\la \phi |C^\downarrow|\psi\ra |\le  (q/2)^{1/2} \lambda_1^{-1}
 \sum_{\mm \in \calJ}\;  \sum_{i,j,k=1}^N x_{ijk\mm} y_{ijk\mm},
\]
where
\[
x_{ijk\mm} = |\psi_\mm|  \cdot
|c_{ijk}| (m_j)^{1/2} \quad \mbox{and} \quad y_{ijk\mm} = |c_{ijk}|^{1/2} \lambda_i^{1/2} \cdot |\la \phi| a_i^\dag a_k|\mm-e^j\ra|.
\]
Cauchy-Schwartz inequality gives
\be
\label{lemmaAC1_eq8}
|\la \phi |C^\downarrow|\psi\ra | \le  (q/2)^{1/2} \lambda_1^{-1}
\sqrt{  \sum_{\mm \in \calJ}\;  \sum_{i,j,k=1}^N  x_{ijk\mm}^2}
\cdot \sqrt{   \sum_{\mm \in \calJ}\;\sum_{i,j,k=1}^N  y_{ijk\mm}^2}.
\ee
For any fixed index $j$ there are at most $s$ nonzero coefficients $c_{ijk}$. 
Thus 
\be
\label{lemmaAC1_eq9}
 \sum_{\mm \in \calJ}\; \sum_{i,j,k=1}^N  x_{ijk\mm}^2 \le sJ  \sum_{\mm \in \calJ}\; \sum_{j=1}^N \psi_\mm^2 m_j \le  sJ \lambda_1^{-1} \la \psi |A|\psi\ra.
\ee
From the definition of creation/annihilation operators one gets
\[
y_{ijk\mm}  = |c_{ijk}|^{1/2} \lambda_i^{1/2} (m_k(m_i+1))^{1/2}  |\phi_{\mm + e^i - e^j -e^k}|.
\]
Perform a change of variables $\nn= \mm+e^i-e^j-e^k$ in the sum over $\mm$.
Then $m_k=n_k+1$ and $m_i+1=n_i$. Thus
\be
\label{lemmaAC1_eq10}
\sum_{\mm \in \calJ}\;  \sum_{i,j,k=1}^N   y_{ijk\mm}^2
\le 
\sum_{\nn \in \calJ}\; \phi_\nn^2 \sum_{i,j,k=1}^N |c_{ijk}|n_i\lambda_i (n_k+1).
\ee
Using the sparsity of $c_{ijk}$ one gets
\[
\sum_{i,j,k=1}^N |c_{ijk}| n_i\lambda_i(n_k+1) \le sJ \sum_{i,k=1}^N n_i \lambda_i n_k+ sJ  \sum_{i=1}^N n_i\lambda_i
\le sJ \lambda_1^{-1} \la \nn|A^2|\nn\ra + sJ  \la \nn|A|\nn\ra
\]
for any $\nn \in \calJ$. Substituting this into Eq.~(\ref{lemmaAC1_eq10}) gives
\be
\label{lemmaAC1_eq11}
\sum_{\mm \in \calJ}\;  \sum_{i,j,k=1}^N   y_{ijk\mm}^2 
\le sJ \lambda_1^{-1} \la \phi|A^2|\phi\ra +  sJ  \la \phi|A|\phi\ra.
\ee
Finally, substituting Eqs.~(\ref{lemmaAC1_eq9},\ref{lemmaAC1_eq11}) into Eq.~(\ref{lemmaAC1_eq8}) one gets
\begin{align*}
|\la \phi |C^\downarrow|\psi\ra | &  \le   (q/2)^{1/2} \lambda_1^{-1}
\sqrt{sJ \lambda_1^{-1} \la \psi|A|\psi\ra}
\cdot
\sqrt{ sJ \lambda_1^{-1} \la \phi|A^2|\phi\ra +  sJ \la \phi|A|\phi\ra }\\
= & 
(q/2)^{1/2} sJ \lambda_1^{-2} \sqrt{\la \psi|A|\psi\ra} \sqrt{\la \phi|A^2|\phi\ra + \lambda_1 \la \phi|A|\phi\ra}.
\end{align*}
Operator inequality $\lambda_1 A \le A^2$  gives $\lambda_1 \la \phi|A|\phi\ra =\la \phi|\lambda_1 A|\phi\ra \le \la \phi|A^2|\phi\ra$.
Hence
\[
|\la \phi |C^\downarrow|\psi\ra |  \le q^{1/2} sJ \lambda_1^{-2} \sqrt{\la \psi|A|\psi\ra} \sqrt{\la \phi|A^2|\phi\ra}.
\]
Combing this, the upper bound on $|\la \phi|C^\uparrow|\psi\ra|$ from Eq.~(\ref{lemmaAC1_eq4}), and $1+1/\sqrt{2}\le 2$  proves 
\be
\label{partCupper}
|\la \phi |C|\psi\ra | \le 2 q^{1/2} sJ \lambda_1^{-2} \sqrt{\la \psi|A|\psi\ra} \sqrt{\la \phi|A^2|\phi\ra}.
\ee

Similar arguments provide a bound on $|\la \phi|B|\psi\ra|$,
By Eq.~(\ref{ABexplicit}), 
$B= \sum_{i,j=1}^N \beta_{i,j} \, a_j^\dag a_i$, 
where 
$\beta_{i,j} = b_{ij}\lambda_j^{-1/2} \lambda_i^{1/2}$.
Thus
\[
|\la \phi|B|\psi\ra| \le \lambda_1^{-1/2} \sum_{\mm \in \calJ} \; |\psi_\mm| \sum_{i,j=1}^N (m_i \lambda_i)^{1/2} |b_{ij}| \cdot |\la \phi|a_j^\dag|\mm-e^i\ra|
=\lambda_1^{-1/2} \sum_{\mm\in \calJ} \; \sum_{i,j=1}^N x_{ij\mm} y_{ij\mm},
\]
where
\[
x_{ij\mm} = |\psi_\mm| (m_i \lambda_i)^{1/2}\cdot |b_{ij}|^{1/2} \quad \mbox{and} \quad
y_{ij\mm} =  |b_{ij}|^{1/2} \cdot  |\la \phi|a_j^\dag|\mm-e^i\ra|.
\]
Let us agree that $y_{ij\mm}=0$ if $m_i=0$. Cauchy-Schwartz gives 
\[
|\la \phi|B|\psi\ra| \le \lambda_1^{-1/2}  \sqrt{ \sum_{\mm \in \calJ}\; \sum_{i,j=1}^N x_{ij\mm}^2}\cdot  \sqrt{ \sum_{\mm \in \calJ}\; \sum_{i,j=1}^N y_{ij\mm}^2}.
\] 
The same arguments as above show that
\[
 \sum_{\mm \in \calJ}\; \sum_{i,j=1}^N x_{ij\mm}^2 \le sJ \la \psi|A|\psi\ra.
\]
We have
\[
y_{ij\mm} =  |b_{ij}|^{1/2} (m_j+1)^{1/2} |\phi_{\mm + e^j-e^i}|.
\]
Perform a change of variables $\nn = \mm +e^j - e^i$ in the sum over $\mm$. Then $m_j+1=n_j$ and we get
\[
 \sum_{\mm \in \calJ}\; \sum_{i,j=1}^N y_{ij\mm}^2 \le  \sum_{\nn \in \calJ}\; \sum_{i,j=1}^N |b_{ij}| n_j |\phi_\nn|^2
 \le sJ \sum_{j=1}^N n_j |\phi_\nn|^2 \le sJ \lambda_1^{-1} \la \phi|A|\phi\ra.
\]
We have shown that
\be
\label{partBupper}
|\la \phi|B|\psi\ra| \le sJ \lambda_1^{-1}  \sqrt{\la \psi|A|\psi\ra \cdot \la \phi|A|\phi\ra} \le sJ \lambda_1^{-3/2}  \sqrt{\la \psi|A|\psi\ra \cdot \la \phi|A^2|\phi\ra}.
\ee
Here the second inequality follows from $\lambda_1 A \le A^2$. 
The triangle inequality $|\la \phi|B+C|\psi\ra|\le |\la \phi|B|\psi\ra| + |\la \phi|C|\psi\ra$ 
and Eqs.~(\ref{partCupper},\ref{partBupper})
prove
the desired bound.
\end{proof}

\section{Proof of Lemma~\ref{lemma:commutator_order1}}
\label{app:B}

For convenience of the reader we restate the lemma.
\begin{lemma*}
For any vector $\psi\in \calF_0$ one has
\be
\la \psi|[A,B+C]|\psi\ra\le 
2Js \lambda_1^{-1/2} 
(1+ 4 q^{1/2}  \lambda_1^{-1/2}  ) \sqrt{ \la \psi|A|\psi\ra \la \psi |A^2|\psi\ra}.
\ee
\end{lemma*}
\begin{proof}
Let us first upper bound $\la \psi|[A,C]|\psi\ra$.
From Eqs.~(\ref{ABexplicit},\ref{Csimplified}) one gets
\be
\label{AC_commutator_eq1}
[A,C] = (q/2)^{1/2}\sum_{i,j,k,\ell=1}^N c_{ijk} \lambda_i^{1/2} (\lambda_j  \lambda_k)^{-1/2}  \lambda_\ell [a_\ell^\dag a_\ell, a_j^\dag a_k^\dag a_i - a_i^\dag a_k a_j ] 
\ee
The canonical  commutation rules give
\be
 [a_\ell^\dag a_\ell, a_j^\dag a_k^\dag a_i ] = (\delta_{j,\ell} + \delta_{k,\ell} - \delta_{i,\ell}) a_j^\dag a_k^\dag a_i.
\ee
Thus
\be
\label{Commutator_eq1}
[A,C] = (q/2)^{1/2}\sum_{i,j,k=1}^N c_{ijk} \lambda_i^{1/2} (\lambda_j  \lambda_k)^{-1/2}  ( \lambda_j + \lambda_k -\lambda_i) (a_j^\dag a_k^\dag a_i + a_i^\dag a_k a_j).
\ee
First let us upper bound the coefficients in the above sum. 
\begin{prop}
\label{prop:coeff_upper}
For any triple of indices $i\ne j\ne k$ one has 
\be
|c_{ijk}|  \lambda_i^{1/2} (\lambda_j  \lambda_k)^{-1/2}  | \lambda_j + \lambda_k -\lambda_i| \le 4J \lambda_1^{-1}(\lambda_i \lambda_j \lambda_k)^{1/2}.
\ee
\end{prop}
\begin{proof}
Recall that $\lambda_i c_{ijk} + \lambda_j c_{jki} + \lambda_k c_{kij}=0$ due to the divergence-free assumption. Hence
\be
|c_{ijk}| | \lambda_j + \lambda_k -\lambda_i| \le |c_{ijk}| (\lambda_j+\lambda_k) + |c_{jki}| \lambda_j + |c_{kij}|\lambda_k  \le 2J(\lambda_j+\lambda_k).
\ee
This gives
\be
|c_{ijk}|  \lambda_i^{1/2} (\lambda_j  \lambda_k)^{-1/2}  | \lambda_j + \lambda_k -\lambda_i| \le 2J \frac{(\lambda_j+\lambda_k)}{\lambda_j \lambda_k} (\lambda_i \lambda_j \lambda_k)^{1/2}
\le 4J \lambda_1^{-1} (\lambda_i \lambda_j \lambda_k)^{1/2}.
\ee
\end{proof}
Define binary coefficients  $\bar{c}_{ijk}\in \{0,1\}$ such that $\bar{c}_{ijk}=1$ if $c_{ijk}$ is nonzero and $\bar{c}_{ijk}=0$ otherwise. 
Substituting the bound of Proposition~\ref{prop:coeff_upper} into Eq.~(\ref{Commutator_eq1}) 
and taking into account that  $|\la \psi| a_j^\dag a_k^\dag a_i|\psi\ra|=|\la \psi| a_i^\dag a_k a_j|\psi\ra|$
one gets
\be
\label{Commutator_eq2}
\la \psi| [A,C] |\psi\ra \le
8 q^{1/2} J \lambda_1^{-1} \omega,
 \ee
where
\[
\omega  = 
 \sum_{i,j,k=1}^N \bar{c}_{ijk}  (\lambda_i \lambda_j \lambda_k)^{1/2}  |\la \psi |a_j^\dag a_k^\dag a_i|\psi\ra|.
\]
Let us  upper bound $\omega$. The argument is nearly identical to the one used in the proof of Lemma~\ref{lemma:AC1}. We have
\[
\la \psi |a_j^\dag a_k^\dag a_i|\psi\ra = \sum_{\mm \in \calJ} \psi_\mm m_i^{1/2} (m_j+1)^{1/2} (m_k+1)^{1/2} \psi_{\mm -e^i + e^j + e^k}^*.
\]
Hence
\[
\omega \le 
\sum_{\mm \in \calJ}\;  \sum_{i,j,k=1}^N x_{ijk\mm} y_{ijk\mm}
\]
where
\[
x_{ijk\mm} = \bar{c}_{ijk} (\lambda_i m_i)^{1/2} |\psi_\mm| \quad \mbox{and} \quad y_{ijk\mm} =\bar{c}_{ijk} (\lambda_j\lambda_k)^{1/2}  (m_j+1)^{1/2} (m_k+1)^{1/2} |\psi_{\mm -e^i + e^j + e^k}|.
\]
Cauchy-Schwartz inequality gives 
\[
\omega  \le  \sqrt{ \sum_{\mm \in \calJ}\; \sum_{i,j,k=1}^N x_{ijk\mm}^2}\cdot  \sqrt{ \sum_{\mm \in \calJ}\; \sum_{i,j,k=1}^N y_{ijk\mm}^2}.
\] 
Here it is understood that the sum of $y_{ijk\mm}^2$ only includes terms $\mm$ with $m_i\ge 1$ (otherwise $y_{ijk\mm}$ is not defined).
From $\sum_{j,k=1}^N \bar{c}_{ijk}\le s$ one gets
\[
 \sum_{\mm \in \calJ}\; \sum_{i,j,k=1}^N x_{ijk\mm}^2 \le 
 s\sum_{\mm \in \calJ} \; \sum_{i=1}^N  \lambda_i m_i \psi_\mm^2= 
 s \la \psi |A|\psi\ra.
\]
Performing a change of variables $\nn = \mm -e^i + e^j +e^k$ in the sum of $y_{ijk\mm}^2$ 
and noting that $m_j+1=n_j$, $m_k+1=n_k$, one gets
\[
 \sum_{i,j,k=1}^N \;  \sum_{\mm \in \calJ}\;
 y_{ijk\mm}^2
 \le \sum_{i,j,k=1}^N \; \bar{c}_{ijk} \sum_{\nn \in \calJ} \;  \lambda_j n_j \lambda_k n_k \psi_\nn^2 \le
  s\la \psi|A^2|\psi\ra.
\]
Hence
\[
\omega  \le s  \sqrt{ \la \psi|A|\psi\ra \la \psi |A^2|\psi\ra}.
\]
Substituting this into Eq.~(\ref{Commutator_eq2}) gives
\[
\la \psi| [A,C] |\psi\ra \le
8 q^{1/2} J s \lambda_1^{-1}  \sqrt{ \la \psi|A|\psi\ra \la \psi |A^2|\psi\ra}.
\]
Next let us upper bound the expected value $\la \psi|  [A,B]|\psi\ra$.
Recall that 
\[
B = \sum_{i,j=1}^N  b_{ij}   \lambda_j^{-1/2} \lambda_i^{1/2} a_j^\dag a_i,
\]
see Eq.~(\ref{ABexplicit}). 
The canonical commutation rules give
\be
\label{app_commutatorAB}
[A,B] = \sum_{i,j=1}^N b_{ij}   \lambda_j^{-1/2} \lambda_i^{1/2} (\lambda_j-\lambda_i) a_j^\dag a_i.
\ee
First let us upper bound the coefficients in the above sum. 
We claim that 
\be
\label{AB_coeff_upper}
|b_{ij}|   \lambda_j^{-1/2} \lambda_i^{1/2} |\lambda_j-\lambda_i|\le 2J (\lambda_i \lambda_j)^{1/2}.
\ee
Indeed,
recall that $\lambda_i b_{ij} + \lambda_j b_{ji}=0$ due to 
the divergence-free assumption.  Hence
\[
|b_{ij}| \cdot  |\lambda_j-\lambda_i| \le( |b_{ij}| + |b_{ji}|) \lambda_j \le 2J \lambda_j.
\]
This proves Eq.~(\ref{AB_coeff_upper}). Define binary coefficients  $\bar{b}_{ij}\in \{0,1\}$ such that $\bar{b}_{ij}=1$ if $b_{ij}$ is nonzero and $\bar{b}_{ij}=0$ otherwise. 
We arrive at 
\[
\la \psi |[A,B]|\psi\ra \le 2J \sum_{i,j=1}^N \bar{b}_{ij} (\lambda_i \lambda_j)^{1/2} |\la \psi|a_j^\dag a_i|\psi\ra|.
\]
From 
\[
\la \psi|a_j^\dag a_i|\psi\ra = \sum_{\mm \in \calJ} \psi_\mm m_i^{1/2} (m_j+1)^{1/2} \psi_{\mm-e^i +e^j}^*
\]
one gets
\[
\la \psi |[A,B]|\psi\ra \le 2J \sum_{\mm\in \calJ} \; \sum_{i,j=1}^N x_{ij\mm} y_{ij\mm},
\]
where
\[
x_{ij\mm}=   \bar{b}_{ij} m_i^{1/2} \lambda_i^{1/2}|\psi_\mm|  \quad \mbox{and} \quad y_{ij\mm} =  \bar{b}_{ij}  (m_j+1)^{1/2}\lambda_j^{1/2}| \psi_{\mm-e^i +e^j}|,
\]
Cauchy-Schwartz inequality gives
\[
\la \psi |[A,B]|\psi\ra \le 2J \sqrt{ \sum_{\mm \in \calJ}\; \sum_{i,j=1}^N x_{ij\mm}^2}\cdot  \sqrt{ \sum_{\mm \in \calJ}\; \sum_{i,j=1}^N y_{ij\mm}^2}.
\]
We have
\[
 \sum_{\mm \in \calJ}\; \sum_{i,j=1}^N x_{ij\mm}^2 \le s\la \psi|A|\psi\ra
\]
Performing a change of variables $\nn=\mm -e^i+e^j$ in the sum of $y_{ij\mm}^2$ and noting that $m_j+1=n_j$ one gets
\[
 \sum_{\mm \in \calJ}\; \sum_{i,j=1}^N y_{ij\mm}^2 \le \sum_{i,j=1}^N \bar{b}_{ij} \sum_{\nn \in \calJ} \; n_j \lambda_j \psi_\nn^2 \le s\la \psi|A|\psi\ra.
 \]
Hence
\[
\la \psi |[A,B]|\psi\ra \le 2Js \la \psi|A|\psi\ra\le 2Js \lambda_1^{-1/2} \sqrt{ \la \psi|A|\psi\ra \la \psi|A^2|\psi\ra}.
\]
Combining the upper bounds on $\la \psi |[A,B]|\psi\ra$ and $\la \psi |[A,C]|\psi\ra$ proves the lemma.
\end{proof}

\section{Proof of Lemma~\ref{lemma:AC2}}
\label{app:lemmaAC2proof}

\begin{lemma*}
For any integer  $p\ge 2$ there exists a real number  $\kappa_p\le poly(s,J,\lambda_1^{-1},\lambda_N)$ 
such that  for any vector $\psi \in \calF_0$ one has 
\be
\label{commutator_bound_restated}
\la \psi| [A^p,B+C]|\psi\ra \le \la \psi | A^{p+1}|\psi\ra  +  \kappa_{p} \la \psi| A^p  |\psi\ra.
\ee
\end{lemma*}
\begin{proof}
Let us first prove Eq.~(\ref{commutator_bound_restated}) for $p=1$. 
Although the case $p=1$ is already covered by Corollary~\ref{corol:commutator_bound_order1},
here we use a different proof strategy that can be generalized to $p\ge 2$.

Given Hermitian operators $X$ and $Y$ let us write $X\le Y$ if $Y-X$ is positive semidefinite.
We need the following implication of Gershgorin circles theorem.
\begin{prop}[\bf Diagonal dominance]
\label{prop:DD}
Suppose $X$ and $Y$ are hermitian operators on $\calF_0$.
Suppose $Y$ is diagonal, $X$ has zero diagonal, and 
\[
\la \mm|Y|\mm\ra \ge \sum_{\nn \in \calJ} |\la \nn|X|\mm\ra|
\]
for all $\mm \in \calJ$. Then $Y-X$ is positive semidefinite, that is, $\la \psi|Y-X|\psi\ra \ge 0$ for all $\psi \in \calF_0$.
\end{prop}
\begin{proof}
\begin{align*}
\la \psi|Y-X|\psi\ra & \ge \sum_{\mm\in \calJ}  \la \mm|Y|\mm\ra |\la \mm|\psi\ra|^2 - \sum_{\nn,\mm\in \calJ} |\la \nn|X|\mm\ra|\cdot |\la \nn|\psi\ra|\cdot |\mm|\psi\ra| \\
& \ge  \sum_{\mm\in \calJ}  \la \mm|Y|\mm\ra |\la \mm|\psi\ra|^2 - \frac12  \sum_{\nn,\mm\in \calJ} |\la \nn|X|\mm\ra| \left(  |\la \nn|\psi\ra|^2 +  |\la \mm|\psi\ra|^2\right)  \\
& =   \sum_{\mm\in \calJ}  \la \mm|Y|\mm\ra |\la \mm|\psi\ra|^2 -  \sum_{\nn,\mm\in \calJ} |\la \nn|X|\mm\ra|  |\la \mm|\psi\ra|^2  \\
& =   \sum_{\mm\in \calJ}   |\la \mm|\psi\ra|^2\left(  \la \mm|Y|\mm\ra  -   \sum_{\nn\in \calJ} |\la \nn|X|\mm\ra|\right) \ge0.
\end{align*}
\end{proof}
Below we write $X\le_d Y$ whenever $X$ and $Y$ are operators satisfying conditions of Proposition~\ref{prop:DD}.
Then $X\le _d Y$ implies $X\le Y$ for any Hermitian operators $X$ and $Y$. 
One can easily check that  $X'\le_d Y'$ and $X''\le_d Y''$ implies $(\alpha X'+\beta X'')\le_d (|\alpha| Y'+|\beta| Y'')$ for any coefficients $\alpha,\beta \in \RR$.

We claim that 
\be
\label{AC_commutator_dbound}
[A,C] \le_d  A^2 + \kappa_C  A
\ee
for some $\kappa_C\le poly(s,J,\lambda_1^{-1},\lambda_N)$.
Indeed, we have already showed that 
\be
\label{commutator_eq1}
[A,C] = (q/2)^{1/2}\sum_{i,j,k=1}^N c_{ijk} \lambda_i^{1/2} (\lambda_j  \lambda_k)^{-1/2}  ( \lambda_j + \lambda_k -\lambda_i) (a_j^\dag a_k^\dag a_i + a_i^\dag a_k a_j).
\ee
see Eq.~(\ref{Commutator_eq1}).
The following proposition proves an analogue of Eq.~(\ref{AC_commutator_dbound}) for a single term in the commutator $[A,C]$.
\begin{prop}
\label{prop:triple}
Consider any triple of indices $i\ne j \ne k$ and hermitian operators
\be
g = \gamma (a_j^\dag a_k^\dag a_i + a_i^\dag a_k a_j) \quad \mbox{and} \quad h = \lambda_i a_i^\dag a_i + \lambda_j a_j^\dag a_j + \lambda_k a_k^\dag a_k,
\ee
where $\gamma \in \RR$.
For any $\beta>0$ one has 
\be
\label{g_upper}
g \le_d  \frac12 |\gamma| \lambda_1^{-3/2} (\beta h^2 + \beta^{-1} h).
\ee
\end{prop}
\begin{proof}
Clearly $h$ is diagonal while $g$ has zero diagonal since it changes the particle number by $\pm 1$.
Define a function 
\[
\sigma(\mm)= \sum_{\nn \in \calJ}  |\la \nn|g|\mm\ra|.
\]
We need to check that 
\be
\label{g_upper_restated}
\sigma(\mm) \le   \frac12 |\gamma| \lambda_1^{-3/2}\la \mm | \beta h^2 + \beta^{-1} h|\mm\ra
\ee
We have 
\[
g |\mm\ra = \gamma m_i^{1/2} ((m_j+1)(m_k+1))^{1/2} |\mm -e^i + e^j + e^k\ra + \gamma (m_i+1)^{1/2} (m_j m_k)^{1/2} |\mm +e^i - e^j - e^k\ra
\]
which gives
\[
\sigma(\mm)= |\gamma| m_i^{1/2} ((m_j+1)(m_k+1))^{1/2} + |\gamma| (m_i+1)^{1/2} (m_j m_k)^{1/2}.
\]
Consider several cases.

\noindent
{\em Case 1:} $m_i=0$.  Then $\sigma(\mm) = |\gamma| (m_jm_k)^{1/2}\le (|\gamma|/2)(m_i+m_j+m_k)$.

\noindent
{\em Case 2:} $m_i\ge 1$, $m_j=m_k=0$.  Then $\sigma(\mm) = |\gamma| m_i^{1/2} \le |\gamma| (m_i+m_j+m_k)$.

\noindent
{\em Case 3:} $m_i\ge 1$, $m_j=0$, $m_k \ge 1$. Then $m_k+1\le 2m_k$ and thus
\[
\sigma(\mm) = |\gamma| m_i^{1/2} (m_k+1)^{1/2}\le \sqrt{2} |\gamma|  (m_im_k)^{1/2}\le ( |\gamma|/\sqrt{2}) (m_i+m_j+m_k).
\]

\noindent
{\em Case 4:} $m_i\ge 1$, $m_j\ge 1$, $m_k=0$. Same argument as above gives $\sigma(\mm) \le   ( |\gamma|/\sqrt{2}) (m_i+m_j+m_k)$.

\noindent
{\em Case 5:} $m_i,m_j,m_k\ge 1$. Then $m_i+1\le 2m_i$ and $m_j+1\le 2m_j$ and $m_k+1\le 2m_k$. In this case
\[
\sigma(\mm) \le 4 |\gamma| (m_i m_j m_k)^{1/2} \le  |\gamma| (m_i+m_j+m_k)^{3/2}.
\]

Thus in all cases one has 
\[
\sigma(\mm) \le  |\gamma| (m_i+m_j+m_k)^{3/2}\le |\gamma| \lambda_1^{-3/2} (m_i \lambda_i + m_j \lambda_j + m_k \lambda_k)^{3/2}.
\]
For any $x\ge 0$ and any $\beta>0$ one has 
$x^{3/2} \le \frac12 (\beta^{-1}  x + \beta x^2)$.
Hence 
\[
\sigma(\mm) \le \frac12 |\gamma| \lambda_1^{-3/2} (\beta x^2 + \beta^{-1} x)
\]
where $x=m_i \lambda_i + m_j \lambda_j + m_k \lambda_k$. This proves Eq.~(\ref{g_upper_restated}) since
$\la \mm|h|\mm\ra = x$ and $\la \mm|h^2|\mm\ra = x^2$. 
\end{proof}
Consider a single term in the righthand side of Eq.~(\ref{commutator_eq1}) of the form 
 $g=\gamma(a_j^\dag a_k^\dag a_i + a_i^\dag a_k a_j)$ with  
\[
\gamma =(q/2)^{1/2} c_{ijk} \lambda_i^{1/2} (\lambda_j  \lambda_k)^{-1/2}  ( \lambda_j + \lambda_k -\lambda_i).
\]
Clearly, 
$|\gamma|\ \le \gamma_{max}\equiv 3(q/2)^{1/2} J \lambda_N^{3/2} \lambda_1^{-1}$.
Proposition~\ref{prop:triple} gives
$g\le_d (1/2)\gamma_{max} \lambda_1^{-3/2} (\beta h_{ijk}^2 + \beta^{-1} h_{ijk})$, where
\[
h_{ijk}  =  \lambda_i a_i^\dag a_i + \lambda_j a_j^\dag a_j + \lambda_k a_k^\dag a_k
\]
and  $\beta>0$ is arbitrary. 
Combining all the terms in Eq.~(\ref{commutator_eq1}) gives 
\[
[A,C] \le_d (1/2) \gamma_{max} \lambda_1^{-3/2} \sum_{i,j,k=1}^N \bar{c}_{ijk} (\beta h_{ijk}^2 + \beta^{-1} h_{ijk}).
\]
Here $\bar{c}_{ijk}=1$ if $c_{ijk}\ne 0$ and $\bar{c}_{ijk}=0$ otherwise. 
Define the particle number operator 
$n_i = a_i^\dag a_i$.
Clearly, $h_{ijk}\le \lambda_N(n_i + n_j+n_k)$ and $h_{ijk}^2 \le \lambda_N^2 (n_i + n_j + n_k)^2$, where the inequality
is understood entrywise, for each diagonal entry of the considered diagonal operators. 
We arrive at
\be
\label{commutator_eq3}
[A,C] \le_d (3/2) (q/2)^{1/2} J (\lambda_N/\lambda_1)^{5/2}\sum_{i,j,k=1}^N \bar{c}_{ijk} (\beta \lambda_N (n_i + n_j +n_k)^2 + \beta^{-1} (n_i + n_j +n_k))
\ee
for any $\beta>0$. 
We also have the following entrywise inequality:
\[
\sum_{i,j,k=1}^N \bar{c}_{ijk} (n_i + n_j +n_k) \le 3s  \sum_{i=1}^N n_i \le 3s\lambda_1^{-1} A.
\]
Here we used the fact that $\bar{c}_{ijk}\in \{0,1\}$ and the assumption that any two-dimensional slice of the tensor $c_{ijk}$ has at most $s$ nonzeros. 
By the same reason, 
\[
\sum_{i,j,k=1}^N \bar{c}_{ijk}   (n_i + n_j +n_k)^2 \le 
\left( \sum_{i,j,k=1}^N \bar{c}_{ijk}  (n_i + n_j +n_k) \right)^2
\le
9s^2 \left( \sum_{i=1}^N n_i \right)^2 \le 9s^2 \lambda_1^{-2} A^2.
\]
Substituting this into Eq.~(\ref{commutator_eq3}) gives
\be
\label{commutator_eq4a}
[A,C] \le_d  (3/2) (q/2)^{1/2} J (\lambda_N/\lambda_1)^{5/2} \left( 
9\beta s^2 \lambda_1^{-2}\lambda_N   A^2 + 3s \beta^{-1} \lambda_1^{-1} A\right).
\ee
We can make the coefficient of $A^2$ equal to one by choosing 
$\beta^{-1} = (27/2) (q/2)^{1/2} s^2 J (\lambda_N/\lambda_1)^{7/2} \lambda_1^{-1}$.
This proves the desired bound Eq.~(\ref{AC_commutator_dbound}) with 
$\kappa_C = (243/8) q  s^3 J^2 \lambda_N^6 \lambda_1^{-8}$.

Next consider a commutator $[A^p,C]$ with $p\ge 2$.
We have 
\be
\label{ApC_commutator_eq1}
[A^p,C] = \sum_{j=1}^p A^{p-j} [A,C] A^{j-1}=\sum_{j=1}^p X_j,
\ee
where $X_j = A^{p-j} [A,C] A^{j-1}$. Given a multi-index $\mm\in \calJ$ let
\be
\label{ApC_commutator_eq1a}
\eta_j(\mm) = \sum_{\nn \in \calJ} |\la \nn |X_j|\mm\ra| = \lambda_\mm^{j-1}  \sum_{\nn \in \calJ} \lambda_\nn^{p-j} |\la \nn |[A,C]|\mm\ra|.
\ee
Here  $\lambda_\mm = \sum_{i=1}^N \lambda_i m_i=\la \mm|A|\mm\ra$. 
From Eq.~(\ref{commutator_eq1}) one infers that $[A,C]|\mm\ra$ is a linear combination of basis states
$|\nn\ra$ such that $\nn = \mm - e^i + e^j + e^k$ or $\nn = \mm +e^i - e^j - e^k$ for some triple of indices $i\ne j\ne k$.
Hence $|\la \nn |[A,C]|\mm\ra|\ne 0$ implies 
\be
\lambda_\nn \le \lambda_\mm + 2\lambda_N \le (1+2\lambda_N\lambda_1^{-1}) \lambda_\mm.
\ee
To get the second inequality we noted that $\lambda_1 \le \lambda_\mm$ for all $\mm \in \calJ$. Hence
\be
\label{ApC_commutator_eq2}
\eta_j(\mm) \le (1+2\lambda_N \lambda_1^{-1})^{p-j} \lambda_\mm^{p-1}  \sum_{\nn \in \calJ} |\la \nn |[A,C]|\mm\ra|.
\ee
From  Eq.~(\ref{commutator_eq4a}) one infers that for any $\beta>0$ 
\be
\label{ApC_commutator_eq3}
\sum_{\nn \in \calJ} |\la \nn |[A,C]|\mm\ra| \le \beta \kappa' \lambda_\mm^2 + \beta^{-1} \kappa'' \lambda_\mm,
\ee
for some coefficients $\kappa',\kappa''\le poly(s,J,\lambda_1^{-1},\lambda_N)$.
 Choose $\beta$ such that 
\[
\beta   p(1+2\lambda_N \lambda_1^{-1})^{p-j} \kappa' \le 1
\]
for all $j=1,\ldots,p$. Clearly, $\beta\le poly(s,J,\lambda_1^{-1},\lambda_N)$ for any $p=O(1)$.
Then Eqs.~(\ref{ApC_commutator_eq2},\ref{ApC_commutator_eq3}) imply
\be
\eta_j(\mm) \le p^{-1} \lambda_\mm^{p+1} + (1+2\lambda_N \lambda_1^{-1})^{p-j} \beta^{-1} \kappa'' \lambda_\mm^p.
\ee
Combining this and Eq.~(\ref{ApC_commutator_eq1}) gives
\be
\sum_{\nn \in \calJ} |\la \nn|[A^p,C]|\mm\ra| \le \sum_{j=1}^p \eta_j(\mm) \le \lambda_\mm^{p+1} + \kappa_{C,p} \lambda_\mm^p,
\ee
where $\kappa_{C,p} = \sum_{j=1}^p  (1+2\lambda_N \lambda_1^{-1})^{p-j} \beta^{-1} \kappa''\le poly(s,J,\lambda_1^{-1},\lambda_N)$.
This proves $[A^p,C] \le_d A^{p+1} + \kappa_{C,p} A^p$.
 Proposition~\ref{prop:DD} then implies
\be
\label{ApC}
\la \psi| [A^p,C] |\psi\ra \le \la \psi|A^{p+1}|\psi\ra +  \kappa_{C,p} \la \psi|A^p|\psi\ra
\ee
for any $\psi \in \calF_0$.

We can apply similar arguments to show that
\be
\label{ABcommutator_eq0}
[A,B]\le_d \kappa_B A
\ee
for some $\kappa_B\le poly(s,J,\lambda_1^{-1},\lambda_N)$.
Indeed, define the particle number operator
$P=\sum_{i=1}^N a_i^\dag a_i$.
Since $\lambda_1 P \le A$ termwise, it suffices to prove that 
\be
\label{ABcommutator_eq1}
[A,B] \le_d \kappa_B  \lambda_1 P.
\ee
From Eq.~(\ref{app_commutatorAB}) one gets 
\be
\label{ABcommutator_eq1a}
[A,B] = \sum_{i,j=1}^N \tilde{\beta}_{i,j} a_j^\dag a_i,
\ee
where  $\tilde{\beta}_{i,j} = b_{ij} \lambda_j^{-1/2} \lambda_i^{1/2} (\lambda_j-\lambda_i)$
is a symmetric matrix. A simple algebra shows that 
\begin{align}
\label{ABcommutator_eq2}
\sum_{\nn \in \calJ} |\la \nn| [A,B]|\mm\ra| & \le \sum_{i,j=1}^N |\tilde{\beta}_{i,j}| (m_i (m_j+1))^{1/2}
 \le 2  \sum_{i,j=1}^N |\tilde{\beta}_{i,j}|  m_i \nonumber \\
 &  \le 2s\|\tilde{\beta}\|_{max} \sum_{i=1}^N m_i =  2s\|\tilde{\beta}\|_{max} \la \mm|P|\mm\ra.
\end{align}
Here
\[
\|\tilde{\beta}\|_{max} \equiv  \max_{i,j} |\tilde{\beta}_{i,j}|\le  2\lambda_N(\lambda_N/\lambda_1)^{1/2} \|b\|_{max}\le 2\lambda_N(\lambda_N/\lambda_1)^{1/2}J.
\]
Comparing Eqs.~(\ref{ABcommutator_eq1},\ref{ABcommutator_eq2}) proves Eq.~(\ref{ABcommutator_eq0}) with
 $\kappa_B =2sJ(\lambda_N/\lambda_1)^{3/2}$.
 
Finally, consider a commutator $[A^p,B]$ with $p\ge 2$.
We have 
\be
\label{ApB_commutator_eq1}
[A^p,B] = \sum_{j=1}^p A^{p-j} [A,B] A^{j-1}=\sum_{j=1}^p X_j,
\ee
where $X_j = A^{p-j} [A,B] A^{j-1}$. Given a multi-index $\mm\in \calJ$ let
\be
\label{ApB_commutator_eq1a}
\eta_j(\mm) = \sum_{\nn \in \calJ} |\la \nn |X_j|\mm\ra| = \lambda_\mm^{j-1}  \sum_{\nn \in \calJ} \lambda_\nn^{p-j} |\la \nn |[A,B]|\mm\ra|.
\ee
From Eq.~(\ref{ABcommutator_eq1a}) one infers that $[A,B]|\mm\ra$ is a linear combination of basis states
$|\nn\ra$ such that $\nn = \mm - e^i + e^j$  for some triple of indices $i\ne j$.
Hence $|\la \nn |[A,B]|\mm\ra|\ne 0$ implies 
$\lambda_\nn \le \lambda_\mm + \lambda_N \le (1+\lambda_N\lambda_1^{-1}) \lambda_\mm$.
From Eq.~(\ref{ABcommutator_eq1}) one gets
\be
\label{ApB_commutator_eq2}
\eta_j(\mm) \le (1+\lambda_N \lambda_1^{-1})^{p-j} \lambda_\mm^{p-1}  \sum_{\nn \in \calJ} |\la \nn |[A,B]|\mm\ra|\le \kappa_B (1+\lambda_N \lambda_1^{-1})^{p-j}  \la \mm|A^p|\mm\ra.
\ee
Substituting this into Eqs.~(\ref{ApB_commutator_eq1},\ref{ApB_commutator_eq1a}) proves
$[A^p,B] \le_d \kappa_{B,p} A^p$
with $\kappa_{B,p}=\sum_{j=1}^p  \kappa_B (1+\lambda_N \lambda_1^{-1})^{p-j}  \le poly(s,J,\lambda_1^{-1},\lambda_N)$ for any $p=O(1)$.
 Proposition~\ref{prop:DD} then gives 
 \be
\label{ApB}
\la \psi| [A^p,B] |\psi\ra \le \kappa_{B,p} \la \psi|A^p|\psi\ra
\ee
for any $\psi \in \calF_0$.
Combining  Eqs.~(\ref{ApC},\ref{ApB}) proves the lemma with $\kappa_p=\kappa_{B,p}+\kappa_{C,p}$.
\end{proof}

\section{Proof of Lemma~\ref{lemma:loose}}
\label{app:D}

\begin{lemma*}
Let $p\ge 0$ be a real number. 
 There exist a real number $\omega_p<\infty$ such that 
\be
\label{ABCloose_restated}
\| A^p(B+C)\psi\| \le \omega_p \| A^{p+3/2}\psi\|
\ee
for any vector $\psi \in \calF_0$. Here $\omega_p$ depends on $N$, $J$, $\lambda_1$, and $\lambda_N$. 
\end{lemma*}
\begin{proof}
Recall that $B$ is a linear combination of finitely many operators $a_j^\dag a_i$ with $i\ne j$. Likewise, $C$ is a linear combination of finitely many operators $a_j^\dag a_k^\dag a_i - a_i^\dag a_k a_j$
with $i\ne j\ne k$.  We have
\[
\|A^p a_j^\dag a_k^\dag a_i \psi\|^2  = \la \psi| X_{ijk}|\psi\ra, \qquad X_{ijk}\equiv a_i^\dag a_k a_j A^{2p} a_j^\dag a_k^\dag a_i.
\]
Note that $X_{ijk}$ is diagonal in the Fock basis $|\mm\ra$. From Eq.~(\ref{ai_action}) one gets
\begin{align*}
\la \mm| X_{ijk}|\mm\ra & = m_i (m_j+1)(m_k+1)  \la \mm+e^j+e^k-e^i|A^{2p} |\mm +e^j+e^k-e^i\ra \\ 
&\le 2^{2p} m_i (m_j+1)(m_k+1)
(\la \mm |A^{2p}|\mm\ra + (2\lambda_N)^{2p}).
\end{align*}
Here the inequality follows from $(y+z)^{2p}\le (2y)^{2p} + (2z)^{2p}$ which holds for all $y,z\ge 0$.
Using a bound $m_i \le \lambda_1^{-1} \la \mm|A|\mm\ra$ and the analogous bound for $m_j,m_k$ one gets
\[
X_{ijk} \le 2^{2p}\lambda_1^{-3}  A(A + \lambda_1I) (A+\lambda_1I) ( A^{2p} +  (2\lambda_N)^{2p} I).
\]
For any integer $1\le \ell \le 2p+3$ one has   $A^{\ell} \le \lambda_1^{-2p-3+\ell} A^{2p+3}$.
Thus
$X_{ijk}\le \omega_p' A^{2p+3}$, where $\omega_p'<\infty$ depends only on $p$, $\lambda_1$, and $\lambda_N$. 
Hence 
$\|A^p a_j^\dag a_k^\dag a_i \psi\|^2 \le \omega_p' \la \psi|A^{2p+3}|\psi\ra$.
Similar arguments show that  $\|A^p a_i^\dag a_k a_j\psi\|^2$ and $\|A^p  a_j^\dag a_i\psi\|^2$
are upper bounded by $\omega_p' \la \psi|A^{2p+3}|\psi\ra$.
We can now bound $\| A^p(B+C)\psi\|$ using the triangle inequality arriving at Eq.~(\ref{ABCloose_restated}).
\end{proof}

\section{Encoding of multi-indices by bit strings}
\label{app:encoding}

A  multi-index $\mm \in \calJ$ can be considered as a multiset\footnote{Recall that a multiset is
a set that may contain multiple copies of each element. The cardinality of a multiset is 
the sum of the multiplicities of all its elements.}
 $M \subseteq \{1,2,\ldots,N\}$
of cardinality $|M|=|\mm|=\sum_{i=1}^N m_i$ that contains $m_i$ copies of an integer $i$,  
\[
M = \{ \underbrace{1,1,\ldots,1}_{m_1},  \underbrace{2,2,\ldots,2}_{m_2}, \ldots,  \underbrace{N,N,\ldots,N}_{m_N}\}.
\]
Recall that $\calJ_k=\{\mm \in \calJ\, : \, \sum_{i=1}^N \lambda_i m_i \le k\}$, where $k$ is the cutoff parameter of
our regularization scheme. 
Let 
$p = \max_{\mm \in \calJ_k}\; |\mm|$
be the maximum number of particles in Fock basis states  $\mm\in \calJ_k$. We have
$p \le \lambda_1^{-1} k$. 
Define a $p$-tuple of integers  $(M_1,\ldots,M_p)$ such that 
$M_j$ is the $j$-th largest  element of  $M$ 
if $1\le j\le |\mm|$, and $M_j=0$ if $|\mm|<j\le p$. 
As a concrete example, suppose $N=7$ and $p=4$. Consider a multi-index
$\mm = (2,0,0,0,0,0,1)$. The corresponding multiset $M=\{1,1,7\}$
has cardinality $|M|=3$. Thus 
$M_1=7$, 
$M_2=M_3=1$, and $M_4=0$. The tuple $(M_1,\ldots,M_p)$ uniquely identifies $\mm$ since
$m_i$ is equal to the number of times an integer $i$ appears in $(M_1,\ldots,M_p)$.
Using the binary encoding of integers
one can identify $M_j$ with a bit string of length $\ell$, 
where  $\ell$ is the smallest integer such that $2^\ell \ge N+1$.
Then $(M_1,\ldots,M_p)$ is a bit string of length $\ell p$
that uniquely identifies $\mm$.  The total number of qubits used to express $\calH_k$ is
\[
Q = \ell p \approx k\lambda_1^{-1}  \log_2{(N)}.
\]
In the above example $N=7$ and thus $\ell=3$.
A multi-index $\mm = (2,0,0,0,0,0,1)$  is encoded by a qubit basis state
\[
|111\ra \otimes |100\ra \otimes |100\ra \otimes |000\ra
\]
where $p=4$ registers store the binary encodings of  $M_1=7$, $M_2=M_3=1$, and $M_4=0$
(the least significant bit is the leftmost) respectively. 

\section{Matrix element calculations}
\label{app:NSE_matrix_elements}

Let us compute $C_{\mm,\nn}$. To this end we note that 
\begin{align}
    b(e_i,e_j,e_k) &= 2(2\lambda_j)^{\frac12}B_{i,j}^k \\
    B_{i,j}^k &= \int_\Omega \sin(2\pi\vec k\cdot \xi) \sin(2\pi\vec i\cdot \xi) \cos(2\pi\vec j\cdot \xi)
    \frac{(\vec{i}^\perp\cdot\vec j)(\vec j^\perp\cdot\vec k^\perp)}{|\vec i|_2||\vec k|_2|\vec j|^2_2}d\xi \\
    &= \frac{(\vec{i}^\perp\cdot\vec j)(\vec j^\perp\cdot\vec k^\perp)}{|\vec i|_2||\vec k|_2|\vec j|^2_2}
    \int_\Omega \sin(2\pi\vec k\cdot \xi) \sin(2\pi\vec i\cdot \xi) \cos(2\pi\vec j\cdot \xi) d\xi \\
    &= \frac14 \frac{(\vec{i}^\perp\cdot\vec j)(\vec j^\perp\cdot\vec k^\perp)}{|\vec i|_2||\vec k|_2|\vec j|^2_2}
    \left( \delta_{\vec k-\vec i+\vec j,0} + \delta_{\vec k-\vec i- \vec j,0} - \delta_{\vec k+\vec i+\vec j,0} - \delta_{\vec k+ \vec i- \vec j,0} \right) \\
    &= \frac14 \frac{(\vec{i}^\perp\cdot\vec j)(\vec j^\perp\cdot\vec k^\perp)}{|\vec i|_2||\vec k|_2|\vec j|^2_2}
    \left( \delta_{\vec k,\vec i+\vec j} + \delta_{\vec k,\vec i- \vec j} - \delta_{\vec k, \vec j- \vec i} \right) 
\end{align}
where the last expression is obtained since multi-indices $\vec k,\vec i,\vec j$ are taken such that either $\vec k_1+\vec i_1+\vec j_1 > 0$ or $\vec k_1=0, \vec k_2>0, \vec i_1=0, \vec i_2>0, \vec j_1 =0, \vec j_2 > 0$  resulting in $\delta_{\vec k+\vec i+\vec j,0}\equiv0$ and so 
\begin{align}
    C_{\mm,\nn} &= \sum_{k=1}^N \int_{\RR^N}  (b_k(x) \partial_{x_k}\HH_\nn)\HH_\mm \mu(dx) \\
    &= -\sum_{\substack{k,i,j\ge 1\\i\ne k,j\notin\{k,i\}}}^N b(e_i,e_j,e_k) \int_{\RR^N} x_ix_j \partial_{x_k}\HH_\nn \HH_\mm \mu(dx)\\
    &= -\sum_{\substack{k,i,j\ge 1\\i\ne k,j\notin\{k,i\}}}^N 2(2\lambda_j)^{\frac12} B_{i,j}^k
    \int_{\RR^N} \sqrt{\nn_k 2\lambda_k/q}\sqrt{q/(2\lambda_i)}\sqrt{q/(2\lambda_j)}\HH_{1_i} \HH_{1_j} \HH_{\nn-1_k} \HH_\mm \mu(dx) \\
    &= -\sum_{\substack{k,i,j\ge 1\\i\ne k,j\notin\{k,i\}}}^N 2B_{i,j}^k (\nn_k \lambda_k/\lambda_i q)^\frac{1}{2}
    \int_{\RR^N} \HH_{1_i} \HH_{1_j} \HH_{\nn-1_k} \HH_\mm \mu(dx)
\end{align}


Notice that for $i\neq j$, $k \notin {i,j}$:
\begin{align}
    (\HH_{1_i} \HH_{1_j}, \HH_{\nn-1_k} \HH_\mm )_\mu 
    &= ((1+\nn_i)^{\frac12}\HH_{1_j}\HH_{\nn-1_k+1_i}+\nn_i^{\frac12}\HH_{1_j}\HH_{\nn-1_k-1_i}, \HH_\mm)_\mu\\
    &=(1+\nn_i)^{\frac12}(1+\nn_j)^{\frac12} (\HH_{\nn-1_k+1_i+1_j}, \HH_\mm)_\mu \\
    &+(1+\nn_i)^{\frac12}\nn_j^{\frac12} (\HH_{\nn-1_k+1_i-1_j}, \HH_\mm)_\mu\\
    &+\nn_i^{\frac12}(1+\nn_j)^{\frac12} (\HH_{\gamma-1_k-1_i+1_j}, \HH_\mm)_\mu \\
    &+\nn_i^{\frac12} \nn_j^{\frac12} (\HH_{\nn-1_k-1_i-1_j}, \HH_\mm)_\mu\\
    &=(1+\nn_i)^{\frac12}(1+\nn_j)^{\frac12}\delta_{\nn-1_k+1_i+1_j,\mm}
    +(1+\nn_i)^{\frac12}\nn_j^{\frac12}\delta_{\nn-1_k+1_i-1_j,\mm}\\
    &+\nn_i^{\frac12}(1+\nn_j)^{\frac12} \delta_{\nn-1_k-1_i+1_j,\mm}
    +\nn_i^{\frac12}\nn_j^{\frac12} \delta_{\nn-1_k-1_i-1_j,\mm}
\end{align}

Then

\begin{align}
    C_{\mm,\nn} = -\frac 12 \sum_{\substack{k,i,j\ge 1\\i\ne k,j\notin\{k,i\}}}^N
    \left(\nn_k q \frac{\lambda_k}{\lambda_i} \right)^\frac{1}{2}
    \frac{(\vec{i}^\perp\cdot\vec j)(\vec j\cdot \vec k)}{|\vec i|_2||\vec k|_2|\vec j|^2_2} 
    \bigl( &\delta_{\vec k,\vec i+\vec j} + \delta_{\vec k,\vec i-\vec j} -\delta_{\vec k,\vec j-\vec i} \bigr) \\
    \biggl(
    (1+\nn_i)^{\frac12}(1+\nn_j)^{\frac12}&\delta_{\nn-1_k+1_i+1_j,\mm}\\
    +(1+\nn_i)^{\frac12}\nn_j^{\frac12} &\delta_{\nn-1_k+1_i-1_j,\mm}\\
    +\nn_i^{\frac12}(1+\nn_j)^{\frac12} &\delta_{\nn-1_k-1_i+1_j,\mm}\\
    +\nn_i^{\frac12}\nn_j^{\frac12} &\delta_{\nn-1_k-1_i-1_j,\mm}
    \biggr)
\end{align}

Finally we notice that the sum of coefficient in front of $\delta_{\nn-1_k-1_i-1_j,\mm}$ should equals to zero since the matrix $C_{\mm,\nn}$ is skew-symmetric. Accounting for that we obtain the final representation:
\begin{align} \label{eq:c_mn_final}
    C_{\mm,\nn} = -\frac 12 \sum_{\substack{k,i,j\ge 1\\i\ne k,j\notin\{k,i\}}}^N
    \left(\nn_k q \frac{\lambda_k}{\lambda_i} \right)^\frac{1}{2}
    \frac{(\vec{i}^\perp\cdot\vec j)(\vec j\cdot \vec k)}{|\vec i|_2||\vec k|_2|\vec j|^2_2} 
    \bigl( &\delta_{\vec k,\vec i+\vec j} + \delta_{\vec k,\vec i-\vec j} - \delta_{\vec k,\vec j-\vec i} \bigr)\\
    \biggl(
    (1+\nn_i)^{\frac12}(1+\nn_j)^{\frac12}&\delta_{\nn-1_k+1_i+1_j,\mm}\\
    +(1+\nn_i)^{\frac12}\nn_j^{\frac12} &\delta_{\nn-1_k+1_i-1_j,\mm}\\
    +\nn_i^{\frac12}(1+\nn_j)^{\frac12} &\delta_{\nn-1_k-1_i+1_j,\mm}
    \biggr)
\end{align}

\end{document}